\documentclass[a4paper, reqno, 11pt]{amsart} 
\usepackage{graphicx} 

\usepackage{amssymb}

\usepackage{epsfig,bbm}  		
\usepackage{epic,eepic}         

\usepackage{bm}                 
\usepackage{amsfonts}
\usepackage{amsthm}
\usepackage{amsmath}
\usepackage{amscd}
\usepackage[latin2]{inputenc}
\usepackage{t1enc}
\usepackage[mathscr]{eucal}
\usepackage{indentfirst}
\usepackage{graphicx}
\usepackage{graphics}
\numberwithin{equation}{section}
\usepackage[margin=2.9cm]{geometry}
\usepackage{hyperref}
\usepackage{dsfont}
\usepackage{csquotes}
\usepackage{stmaryrd}
\usepackage{blindtext}
\usepackage{tikz-cd}
\usepackage{mathrsfs}
\usepackage{cite}
\usepackage{mathtools}
\usepackage{extarrows}
\usepackage{epstopdf}
\usepackage[T1]{fontenc}
\usepackage{braket}
\usepackage{yfonts}
\usepackage{array}
\usepackage{multirow}

\usepackage{tikz}
\usepackage{color}

\let\oldtocsection=\tocsection
 
\let\oldtocsubsection=\tocsubsection
 
\let\oldtocsubsubsection=\tocsubsubsection
 
\renewcommand{\tocsection}[2]{\hspace{0em}\oldtocsection{#1}{#2}}
\renewcommand{\tocsubsection}[2]{\hspace{1em}\oldtocsubsection{#1}{#2}}
\renewcommand{\tocsubsubsection}[2]{\hspace{2em}\oldtocsubsubsection{#1}{#2}}

\usepackage{scalerel,stackengine}
\stackMath
\newcommand\reallywidehat[1]{%
\savestack{\tmpbox}{\stretchto{%
  \scaleto{%
    \scalerel*[\widthof{\ensuremath{#1}}]{\kern.1pt\mathchar"0362\kern.1pt}%
    {\rule{0ex}{\textheight}}
  }{\textheight}%
}{2.4ex}}%
\stackon[-6.9pt]{#1}{\tmpbox}%
}

\theoremstyle{definition}
\newtheorem{definition}[equation]{Definition}
\newtheorem{example}[equation]{Example}
\newtheorem{proposition}[equation]{Proposition}
\newtheorem{lemma}[equation]{Lemma}
\newtheorem{theorem}[equation]{Theorem}
\newtheorem{remark}[equation]{Remark}
\newtheorem{corollary}[equation]{Corollary}

\newcommand{\sfgamma}{\mathsf{\Gamma}}
\newcommand{\cbrak}[1]{\llbracket #1 \rrbracket}
\newcommand{\ip}[1]{\langle #1 \rangle}
\newcommand{\cf}[1]{\mathcal{#1}}
\newcommand{\red}[1]{\underline{#1}}
\newcommand{\frk}[1]{\mathfrak{#1}}

\newcommand{\de}{\mathrm{d}}

\newcommand{\del}{\partial}

\newcommand{\e}{\mathrm{e}}
\newcommand{\ii}{\mathrm{i}}
\newcommand{\rmp}{\mathtt{p}}

\newcommand{\IZ}{\mathbbm{Z}}
\newcommand{\IC}{\mathbbm{C}}

\newcommand{\IN}{\mathbbm{N}}
\newcommand{\IR}{\mathbbm{R}}

\newcommand{\IT}{\mathbb{T}}

\newcommand{\IV}{\mathbb{V}}
\newcommand{\IW}{\mathbb{W}}

\newcommand{\frg}{\mathfrak{g}}
\newcommand{\frf}{\mathfrak{f}}

\newcommand{\frh}{\mathfrak{h}}

\newcommand{\frd}{\mathfrak{d}}

\renewcommand{\Im}{\ensuremath{\mathrm{im}}}

\newcommand{\Cl}{\mathsf{Cl}}
\newcommand{\gric}{\mathrm{GRic}}

\def\e{{\,\rm e}\,}

\newcommand{\cN}{{\mathcal N}}

\newcommand{\cS}{{\mathcal S}}
\newcommand{\cB}{{\mathcal B}}

\newcommand{\cH}{{\mathcal H}}

\newcommand{\cQ}{{\mathcal Q}}
\newcommand{\cR}{{\mathcal R}}

\newcommand{\cF}{{\mathcal F}}

\newcommand{\cK}{{\mathcal K}}

\newcommand{\cG}{{\mathcal G}}

\newcommand{\cJ}{{\mathcal J}}

\newcommand{\ccR}{{\mathscr R}}

\newcommand{\ccG}{{\mathscr G}}

\newcommand{\sfD}{{\mathsf{D}}}
\newcommand{\sfS}{{\mathsf{S}}}
\newcommand{\sfs}{{\mathsf{s}}}

\newcommand{\sfT}{{\mathsf{T}}}

\newcommand{\sfH}{{\mathsf{H}}}
\newcommand{\sfG}{{\mathsf{G}}}

\newcommand{\sfO}{{\mathsf{O}}}

\newcommand{\sfU}{{\mathsf{U}}}
\newcommand{\sfGamma}{{\mathsf{\Gamma}}}
\newcommand{\sfOmega}{{\mathsf{\Omega}}}
\newcommand{\sfCl}{\mathsf{Cl}}

\newcommand{\unit}{\mathds{1}}   			
\newcommand{\midwedge}{\text {\Large$\wedge$}}

\newcommand{\midodot}{\text {\Large$\odot$}}

\usepackage{enumitem}

\newcommand{\gr}{\text{gr}}
\newcommand{\ann}{\mathrm{Ann}}
\newcommand{\rel}{\dashrightarrow}

\newcommand{\sfgammabas}{\sfgamma_{\text{bas}}}

\newcommand{\rk}{\mathrm{rk}}

\newcommand{\ad}{\mathsf{ad}}

\renewcommand{\div}{{\rm div}}

\newtheoremstyle{case}{}{}{}{}{}{:}{ }{}
\theoremstyle{remark}

\makeatletter
\newcommand{\oset}[3][0ex]{%
  \mathrel{\mathop{#3}\limits^{
    \vbox to#1{\kern-3\ex@
    \hbox{$\scriptstyle#2$}\vss}}}}
\makeatother

\makeatletter
\newcommand{\osett}[3][0ex]{%
  \mathrel{\mathop{#3}\limits^{
    \vbox to#1{\kern-2\ex@
    \hbox{$\scriptstyle#2$}\vss}}}}
\makeatother

\makeatletter
\newcommand{\osettt}[3][0ex]{%
  \mathrel{\mathop{#3}\limits^{
    \vbox to#1{\kern-1\ex@
    \hbox{$\scriptstyle#2$}\vss}}}}
\makeatother

\makeatletter

\DeclareRobustCommand\mathflip[1]{%
    \mathpalette\@mathflip{#1}%
}

\newcommand\@mathflip[2]{%
    \mskip4mu
    \pdfsave
    \pdfsetmatrix{-1 0 .5 1}%
    \hb@xt@\z@{\hss$\m@th#1 #2$\hss}%
    \pdfrestore
    \mskip7mu
}

\DeclareRobustCommand\mathflipnu[1]{%
    \mathpalette\@mathflipnu{#1}%
}

\newcommand\@mathflipnu[2]{%
    \mskip2mu
    \pdfsave
    \pdfsetmatrix{-1 0 .5 1}%
    \hb@xt@\z@{\hss$\m@th#1 #2$\hss}%
    \pdfrestore
    \mskip8mu
}

\DeclareRobustCommand\mathflips[1]{%
    \mathpalette\@mathflips{#1}%
}

\newcommand\@mathflips[2]{%
    \mskip4mu
    \pdfsave
    \pdfsetmatrix{-1 0 0 1}%
    \hb@xt@\z@{\hss$\m@th#1 #2$\hss}%
    \pdfrestore
    \mskip7mu
}

\makeatother

\mathtoolsset{showonlyrefs}

\title{Generalised Complex and Spinor Relations}

\author[T.~C.~De Fraja]{Thomas C.~De Fraja}
\address[Thomas C.~De Fraja]
{Department of Mathematics and Maxwell Institute for Mathematical
  Sciences\\ Heriot-Watt
  University\\ Edinburgh EH14 4AS\\ United Kingdom}
\email{tcd2000@hw.ac.uk}

\author[V.~E.~ Marotta]{Vincenzo Emilio Marotta}
\address[Vincenzo Emilio Marotta]
{{Dipartimento di Matematica e Geoscienze, Universit\`{a} di Trieste,   Via A. Valerio 12/1, 34127 Trieste, Italy}}
\email{vincenzoemilio.marotta@units.it}

\author[R.~J. Szabo]{Richard J.~Szabo}
  \address[Richard J.~Szabo]
  {Department of Mathematics and Maxwell Institute for Mathematical Sciences\\
  Heriot-Watt University\\
  Edinburgh EH14 4AS \\
  United Kingdom}
  \email{R.J.Szabo@hw.ac.uk}

\date{}

\begin{document}

\begin{abstract}
    Courant algebroid relations are used to define notions of relations between Dirac structures and spinors. It is shown under which circumstances a spinor relation gives a Courant algebroid relation and how it descends to a relation between Dirac structures. A converse to this result is proved: a T-duality relation induces a spinor relation that links the Dirac generating operators defining  T-dual Courant algebroids, generalising the isomorphism of twisted cohomology associated with topological T-duality. We  introduce the notion of relation between generalised complex structures and  characterise their reduction. We also define relations between generalised K{\"a}hler structures, and rephrase them in terms of bi-Hermitian structures which induce T-duality relations between  $\cN=(2,2)$ supersymmetric sigma-models. We prove existence results for T-dual structures, and demonstrate  compatibility of T-duality relations with Type~II supergravity equations.
\end{abstract}

\maketitle

{\baselineskip=12pt
\tableofcontents
}

\section{Introduction}
Courant algebroids, as introduced in~\cite{Courant1990, Weinstein1997, Severa-letters, Hitchin:2003cxu}, have  proved to be an indispensable geometric tool for understanding many aspects of supergravity and string theory, see e.g.~\cite{Coimbra:2011nw,Jurco2016courant} and references therein, as well as aiding in the understanding of Hamiltonian reduction and its various generalisations, see e.g.~\cite{courant1988beyond,bursztyn2007dirac, LiBland2009,vsevera2001poisson}. In this paper we continue our explorations~\cite{DeFraja:2023fhe,deFraja2025Ricci} of the uses of Courant algebroid relations in providing geometric formulations of T-duality and its  generalisations in string theory. To put the content of the present paper into context, let us start by briefly summarising the history of Courant algebroids following~\cite{kosmann2013courant}.

\medskip

\subsection{Courant Algebroids}~\\[5pt]
\label{sub:CAhistory}
Let $M$ be a smooth manifold. Courant algebroids were first described by Courant and Weinstein in~\cite{courant1988beyond,Courant1990} in their search for integrability conditions for certain subbundles, called \emph{Dirac structures}, of the bundle $TM \oplus T^*M$ endowed with a natural skew-symmetric bracket arising from the Lie bracket of vector fields and the Lie derivative of one-forms along vector fields, which was subsequently called the \emph{Courant bracket}.
They were looking for a suitable geometric description of Dirac's theory of constrained dynamical systems. 
These bundles provide a natural framework for unifying Poisson and (pre)symplectic structures.\footnote{The graphs of Poisson bivectors and presymplectic two-forms define isotropic subbundles of $TM\oplus T^*M$, with respect to the canonical pairing between vectors and one-forms, making them key examples of  Dirac structures.} At the same time,  Dorfman was independently studying Dirac structures from a more algebraic perspective~\cite{dorfman1987dirac}, seeking the same integrability condition. The integrability condition they sought is given by closure under the Courant bracket --- analogous to the integrability of distributions of the tangent bundle $TM$ equipped with the Lie bracket --- which describes the integrability of both Poisson and symplectic structures.

After the emergence of Lie algebroids as the infinitesimal counterparts of Lie groupoids, MacKenzie and Xu \cite{mackenzie1994lie} defined  Lie bialgebroids as the infinitesimal counterparts of Poisson groupoids. Since the double of a Lie bialgebra is a Drinfel'd double \cite{drinfel1988quantum}, this naturally led Liu, Weinstein and Xu to consider how this correspondence generalises to the algebroid setting, i.e. they sought to describe the double of a Lie bialgebroid. In~\cite{Weinstein1997} they found this to be a Courant algebroid, and the first axiomatic definition of a Courant algebroid was given therein: A Courant algebroid is a vector bundle $E$ over a manifold $M$, equipped with a bracket on its space of sections, a symmetric non-degenerate pairing, and an \emph{anchor map} taking sections of $E$ to vector fields on $M$, together satisfying a  set of compatibility conditions.

The definition of Courant algebroids subsequently evolved with Roytenberg's introduction of the non-skew-symmetric Dorfman bracket \cite{roytenberg1999courant}, denoted by $\cbrak{ \, \cdot \, , \, \cdot \,}$, replacing the skew-symmetric Courant bracket, where skew-symmetry is sacrificed in order to satisfy the Jacobi identity. The Dorfman bracket became the defining bracket for Courant algebroids, and the number of defining axioms was slowly whittled down from five to three \cite{uchino2002remarks,kosmann2005quasi}.

As these developments occurred, other manifestations of Courant algebroids emerged, including a one-to-one correspondence with a certain type of differential graded (dg) symplectic manifolds \cite{roytenberg2002structure}. 

More relevant to the subject of the present paper, in his letters to Weinstein \cite{Severa-letters}, \v Severa explains how Courant algebroids appear in the variational problems for two-dimensional non-linear sigma-models. Also relevant to this work, Kosmann-Schwarzberg and Roytenberg \cite{Kosmann-Schwarzbach:2003skv, roytenberg2002structure} showed how one can recover the Courant bracket as a ``derived bracket'' by introducing spinors associated with the pseudo-Euclidean vector bundle structure of Courant algebroids. 

Spinors for the generalised tangent bundle $$\IT M := TM \oplus T^*M$$ are given by sections of $\midwedge^\bullet T^*M$ and the Clifford action of sections of $TM \oplus T^*M$ is induced by the contraction with the vector field component and the exterior product with the one-form component:
\begin{align}\label{eqn:Cl-action}
    (X + \alpha) \cdot \omega = \iota_X \omega + \alpha \wedge \omega \ .
\end{align}
Then the Dorfman bracket on $TM \oplus T^*M$ is the derived bracket given by two nested commutators of the de Rham differential with this Clifford action:
\begin{align}
    \cbrak{X + \alpha, Y +\beta} \cdot \omega \coloneqq [[\de \, , (X+\alpha) \,\cdot\,], (Y + \beta) \,\cdot\,] \, \omega \ ,
\end{align}
for all $X+ \alpha, Y + \beta \in \sfgamma(TM \oplus T^*M)$ and $\omega \in \sfOmega^\bullet(M)$. 
This bracket can also be twisted by replacing $\de$ with the twisted differential $$\de_H = \de + H \, \wedge \ , $$ where $H$ is a closed three-form.

The role of spinors in generalised geometry is two-fold. On the physics side, the spinor fields define the Ramond--Ramond sector of Type II supergravity, and the twisted differential defines the integrability condition for these fields. This, among other structures placed on the Courant algebroid, will allow for the complete geometric description of T-duality for the full bosonic sector of Type II supergravity.
On the mathematical side, they give an equivalent description of generalised complex geometry, as introduced by Gualtieri in his thesis~\cite{gualtieri:tesi}. 

\medskip

\subsection{Generalised Complex Geometry}~\\[5pt]
Generalised complex geometry concerns the generalisation of complex structures on $TM$ to the bundle $T M \oplus T^*M$; it allows for the unification of symplectic and complex structures in the framework of Courant algebroids. These are the two structures that are interchanged in mirror symmetry, which has a very similar flavour to the T-duality of string theory. For instance, the Strominger-Yau-Zaslow (SYZ) conjecture~\cite{strominger1996mirror} asserts that mirror symmetry has a geometric incarnation as T-duality. This is expanded on in \cite{ben2004mirror} using the language of generalised complex geometry, or more precisely generalised K\"ahler geometry \cite{gualtieri2014generalized} and generalised Calabi-Yau geometry \cite{Hitchin:2003cxu}. 
Though we do not discuss mirror symmetry in detail in this paper, generalised complex geometry is an important tool for studying it, and thus a concrete mathematical framework for relating generalised complex structures is an important step. 

Generalised K{\"a}hler structures can be introduced by mirroring the definition of K{\"a}hler structure on the vector bundle $TM \oplus T^*M$.  
Gualtieri \cite{gualtieri:tesi} then defined them as described above and proved the remarkable result that any generalised K{\"a}hler structure corresponds to a bi-Hermitian structure.   
Generalised complex and K{\"a}hler geometries arise naturally in the context of two-dimensional supersymmetric sigma-models. In fact, generalised K{\"a}hler structures were discovered (in the language of bi-Hermitian structures) in the pioneering work \cite{Gates:1984nk} characterising the target spaces which support $\cN=(2,2)$  sigma-models.

A non-linear sigma-model in two dimensions with $\cN=(2,2)$ supersymmetry has a pair of $\cN=2$ supersymmetries (right-handed and left-handed). It is shown in \cite{Zabzine:2005qf} that one of the $\cN=2$ supersymmetry generators can always be written in terms of a generalised complex structure on the (twisted) standard Courant algebroid $\IT M$ on the target space $M$ of the sigma-model. Moreover, $\cN=2$ supersymmetry exists if and only if $M$ admits a generalised complex structure \cite{Zabzine:2005qf}.
It was shown by Gates, Hull and Ro{\v c}ek in \cite{Gates:1984nk} that the non-linear sigma-model is invariant under $\cN=(2,2)$ supersymmetry transformations when the target space $M$ admits a generalised K{\"a}hler structure which is defined by the generalised complex structure associated with one of the supersymmetry generators and the generalised metric appearing in the action functional of the sigma-model.

\medskip

\subsection{Courant Algebroid Relations}~\\[5pt]
Courant algebroid relations are inspired by the arrows in Weinstein's symplectic ``category'', which are given by Lagrangian correspondences. Symplectic manifolds $(M_i, \omega_i)$ with symplectic two-forms $\omega_i$ for $i=1,2$ are in \emph{Lagrangian correspondence} if there is a Lagrangian submanifold $L$ of the symplectic manifold $(M_1 \times M_2, \rmp_1^*\, \omega_1 - \rmp_2^*\,\omega_2)$:\footnote{Here we use the canonical projections $\rmp_i \colon M_1 \times M_2 \to M_i$, for $i=1,2$.}
\begin{align}\label{eqn:Lagcorresp}
    L \ \subset \ M_1 \times M_2  \qquad \text{and} \qquad (\rmp_1^* \, \omega_1 - \rmp_2^*\,\omega_2)\big|_L = 0\ .
\end{align}
For example, the graph $\gr(\varphi) \subset M_1 \times M_2$ of a symplectomorphism $\varphi \colon (M_1, \omega_1) \to (M_2, \omega_2)$ always gives a Lagrangian correspondence. 

If a Lie group $\sfG$ with Lie algebra $\frg$ has a Hamiltonian action on a symplectic manifold $(M,\omega)$ with moment map $\mu \colon M \to \frg^*$, then $\mu^{-1}(0)$ gives a Lagrangian correspondence between $(M,\omega)$ and its symplectic reduction $(\mu^{-1}(0)/\sfG, \red \omega)$. Symplectic reduction can be generalised to any Poisson manifold $P$ by considering a Lie group-valued moment map $\mu \colon P \to \sfG^*$, where $\sfG$ is now a Poisson-Lie group and $\sfG^*$ integrates the dual space $\frg^*$ \cite{lu1991momentum}. 
This can be generalised further to quasi-Poisson manifolds, as introduced in \cite{alekseev1998lie}, where the Schouten-Nijenhuis bracket of the quasi-Poisson bivector is given by a trivector rather than vanishing as in the Poisson case. When the action of $\sfG$ is `quasi-Hamiltonian', the moment map is then a map $\mu \colon P \to \sfD/\sfG$, where $(\frd, \frg)$ is a Manin pair\footnote{A Manin pair $(\frd, \frg)$ consists of a Lie algebra $\frd$ with symmetric pairing of split signature, and a Lagrangian subalgebra $\frg$.} and $\sfD$ integrates $\frd$ \cite{alekseev2000manin,bursztyn2007dirac}. 

As mentioned in Section~\ref{sub:CAhistory}, the theory of Manin pairs (or equivalently Lie bialgebras) can be encapsulated by the study of pairs $(E,L)$, where $E$ is a Courant algebroid equipped with a Dirac structure $L\subset E$.\footnote{The usual notion of a Manin pair is obtained by considering Courant algebroids over a point.} Thus an adequate theory of quasi-Poisson reduction should resemble the Lagrangian correspondence of symplectic reduction, replacing the symplectic reduction by Courant algebroid reduction. In \cite{LiBland2009, bursztyn2008courant} this reduction was described as a \emph{Courant morphism}, defined as a Dirac structure in $E \times \overline{\red E}$ supported on the graph of a function between the underlying manifolds. Here $\red E$ is the reduced Courant algebroid and $\overline{\red E}$ flips the sign of the pairing on $\red E$, similar to flipping the sign of $\omega_2$ in the Lagrangian correspondence in \eqref{eqn:Lagcorresp}. 
Courant algebroid relations are the generalisation of Courant morphisms to cases when there is no map between the spaces on which the Courant algebroids are fibred over. Thus the theory of quasi-Hamiltonian reduction is encompassed by Courant algebroid relations. 

{\v S}evera argues that T-duality should be an \enquote{almost isomorphism} between two quasi-Hamiltonian spaces,  with the quasi-Hamiltonian spaces being obtained through quasi-Hamiltonian reduction, and the \enquote{almost isomorphism} described by a Courant algebroid relation \cite{Severa-unpubl}. In \cite{DeFraja:2023fhe, Vysoky2020hitchiker} the description of T-dualities in terms of Courant algebroid relations is completed and extended. In particular, T-duality is described in \cite{DeFraja:2023fhe} as a Courant algebroid relation $R$ between $\red E{}_1 \to \cQ_1$ and $\red E{}_2 \to \cQ_2$, where the Courant algebroids $\red E{}_1$ and $\red E{}_2$ come from the reduction of a Courant algebroid $E \to M$. The base manifold $M$ is assumed to be foliated by foliations $\cF_1$ and $\cF_2$ with respect to which the reductions are performed, with $\cQ_1=M/\cF_1$ and $\cQ_2=M/\cF_2$. The main result of \cite{DeFraja:2023fhe} is that, given a generalised metric on $\red E{}_1$ satisfying some invariance conditions, there exists a unique generalised metric on $\red E{}_2$ such that the Courant algebroid relation $R$ behaves like an ``isometry'' between them. This fully encodes the well-known T-duality prescription for non-linear sigma-models of closed bosonic strings.

\medskip

\subsection{The Fourier-Mukai Transform}~\\[5pt]
\label{sub:TDualityFM_intro}
In Section~\ref{sub:CAhistory} we encountered the space of spinors for the exact Courant algebroid $\IT M$, namely the module $\midwedge^\bullet T^*M$ over the Clifford algebra $\sfCl(\IT M)$. Consider Type~II supergravity with $H$-flux given by a closed three-form $H\in \sfOmega_{\tt cl}^3(M)$. The Ramond--Ramond fluxes are spinors $F \in \sfOmega^\bullet(M)$ satisfying the Bianchi identity
\begin{align}
    \de F + H \wedge F = 0 \ .
\end{align}
As before, the $H$-twisted exterior derivative is $\de_H = \de + H \, \wedge\,$,  
thus the Ramond--Ramond fluxes are closed under $\de_H$, and hence of interest is their (twisted) cohomology class in the cohomology group $\sfH^\bullet(M,H)$ of the differential complex $(\sfOmega^\bullet(M), \de_H)$.

The Fourier-Mukai transform characterises T-duality as an isomorphism of twisted cohomology of torus bundles. Let $\sfT^k$ be a torus of dimension $k$.
In \cite{Bouwknegt2003topology, cavalcanti2011generalized} T-duality is formulated for two principal $\sfT^k$-bundles $\cQ_1$ and $\cQ_2$ over a common base manifold $\cB$. This allows for the introduction of the \emph{correspondence space} $$M = \cQ_1 \times_\cB \cQ_2 \ , $$ with respective projections $\varpi_1$ and $\varpi_2$ to $\cQ_1$ and $\cQ_2$. The correspondence space is a principal $\sfT^{2k}$-bundle over $\cB$ which sits in the commutative diagram 
\begin{equation}\label{eqn:correspondence}
\begin{tikzcd}[row sep = 1cm]
 & M \arrow[dl,"\varpi_1",swap] \arrow[dr,"\varpi_2"] & \\
 \cQ_1 \arrow[dr] &  & \cQ_2 \arrow[dl] \\ 
 & \cB &
\end{tikzcd}
\end{equation}

Suppose that both $\cQ_1$ and $\cQ_2$ are endowed with  $\sfT^k$-invariant closed three-forms $\red H{}_1$ and $\red H{}_2,$ respectively.
Then $(\cQ_1,\red H{}_1)$ and $(\cQ_2,\red H{}_2)$ are said to be T-dual if 
\begin{align}
 \varpi_1^*\, \red H{}_1 - \varpi_2^*\, \red H{}_2 = \de B \ ,
\end{align}
for some $\sfT^{2k}$-invariant two-form $B \in \mathsf{\Omega}^2_{\sfT^{2k}}(M)$ whose restriction to $\ker(\varpi_{1*}) \otimes \ker(\varpi_{2 *})$ is non-degenerate. 

The Fourier-Mukai transform $$\varrho \colon \sfOmega^\bullet_{\sfT^k}(\cQ_1) \longrightarrow \sfOmega^\bullet_{\sfT^k}(\cQ_2)$$ is then defined by
\begin{align}\label{eqn:FMTintroduction}
 \varrho(\omega) \coloneqq \int_{\sfT^k}\, \e^{B}\, \wedge \varpi_1^*\,\omega 
\end{align}
for any \smash{$\omega \in \mathsf{\Omega}^\bullet_{\sfT^k}(\cQ_1)$}, where the fibrewise integration is the pushforward of forms by the projection $\varpi_{2} \colon M\to\cQ_2$; this integration is well-defined because the fibres are compact manifolds. It  is shown in~\cite{Bouwknegt2003topology} that $\varrho$ defines a degree-shifting isomorphism between the twisted differential complexes \smash{$\big(\mathsf{\Omega}^\bullet_{\sfT^k}(\cQ_1) , \de_{\red H{}_1}\big)$} and 
\smash{$\big(\mathsf{\Omega}^\bullet_{\sfT^k}(\cQ_2) , \de_{\red H{}_2}\big)$}. It describes the transformation of Ramond--Ramond fields in Type~II string theory under fibrewise T-duality.
    
The Fourier-Mukai transform $\varrho$ becomes an isomorphism of irreducible Clifford modules once an isomorphism 
\begin{align}
    \mathscr{R} \colon \mathsf{\Gamma}_{\sfT^k}(\IT \cQ_1) \longrightarrow \mathsf{\Gamma}_{\sfT^k}(\IT\cQ_2)    
\end{align}
 of $C^\infty(\cB)$-modules is chosen such that
\begin{align}\label{eqn:Cliffordmoduleiso}
 \varrho(e \cdot \omega) = (-1)^k \, \mathscr{R}(e) \cdot \varrho(\omega) \ ,   
\end{align}
for any \smash{$e \in \mathsf{\Gamma}_{\sfT^k}(\IT\cQ_1)$} and \smash{$\omega \in \mathsf{\Omega}^\bullet_{\sfT^k}(\cQ_1),$} where $e \cdot \omega$ is the Clifford action defined in~\eqref{eqn:Cl-action}. In~\cite{cavalcanti2011generalized} this isomorphism was shown to exist and to be unique.

As shown in \cite{DeFraja:2023fhe}, applying the construction of the T-duality relation outlined earlier, where the foliations $\cF_1$ and  $\cF_2$ are generated by the two $\sfT^k$-actions on $M = \cQ_1 \times_\cB \cQ_2$, we obtain a Courant algebroid relation $R \colon \IT \cQ_1 \rel \IT \cQ_2$. In terms of $\sfT^k$-invariant sections of $\IT \cQ_i$, this recovers $\mathscr{R}$. From there we can consider the Clifford algebra  $\Cl(R)$ generated by $R$. Sections of $\Cl(R)$ act on $\sfOmega^\bullet(\cQ_1) \times \sfOmega^\bullet(\cQ_2)$.\footnote{Since $R \subset \IT \cQ_1 \times \overline{\IT \cQ_2}$ is isotropic, the Clifford algebra $\sfCl(R)$ is defined by considering $R\subset \IT \cQ_1 \times \IT \cQ_2$. The isotropic condition ensures that the Clifford action is well defined.} If $(e_1, e_2) \in R$, then by \eqref{eqn:Cliffordmoduleiso} it follows that
\begin{align}
    (e_1, e_2) \cdot \big(\omega, \varrho(\omega)\big) :=  \big(e_1 \cdot \omega, (-1)^k\,e_2 \cdot \varrho(\omega)\big) = \big(e_1 \cdot \omega, \varrho( e_1 \cdot \omega)\big) \ \in \ \gr(\varrho) \ .
\end{align}
Thus the graph $\gr(\varrho)\subset\sfOmega_{\sfT^k}^\bullet(\cQ_1) \times \sfOmega_{\sfT^k}^\bullet(\cQ_2)$ is invariant under this action, and hence defines a spinor representation. One can show that this is an irreducible representation. 

This motivates the definition of an \emph{$R$-Clifford relation}, for an arbitrary Courant algebroid relation $R \colon E_1 \rel E_2$. Here $E_i$ are Courant algebroids with spinor bundles $\sfS_{E_i}$ for $\Cl(E_i)$, for  $i=1,2$. An $R$-Clifford relation is then a subbundle $\sfS_R \subset \sfS_{E_1} \times \sfS_{E_2}$ providing an irreducible representation for $\Cl(R)$, i.e. $\sfS_R$ is a spinor bundle for $\Cl(R)$.

\medskip

\subsection{Outline and Summary of Main Results}~\\[5pt]
Let us briefly outline the structure of this paper. In \textbf{Section \ref{sect:section2}} we review the language of Clifford algebras over finite-dimensional vector spaces and introduce the notion of relation between spinors,  linking it to the notion of relation between quadratic vector spaces. \textbf{Section \ref{sect:section3}} is dedicated to the basics of Courant algebroids with particular emphasis on the description of Dirac structures, generalised complex structures and generalised K{\"a}hler structures. In \textbf{Section \ref{sect:section4}} we turn to the description of Courant algebroid relations and introduce a notion of relation between Dirac structures, generalising existing notions of morphisms of Dirac structures. We then describe in detail the application of Courant algebroid relations to T-duality. 
\textbf{Section \ref{sect:section5}} tackles the notion of Courant algebroid relation from the spinor standpoint. The notion of Clifford algebra relation is introduced and used to describe related Dirac generating operators. This is then applied to the case of T-duality, in order to extend some of its properties to the setting of \cite{DeFraja:2023fhe}, and supergravity. In \textbf{Section \ref{sect:section6}} Courant algebroid relations for generalised complex structures are introduced and their applications to reduction and T-duality are discussed. Similarly, in \textbf{Section \ref{sect:section7}} the notion of relation between generalised K{\"a}hler structures is introduced and its main example in terms of T-duality is given, together with explicit examples. \textbf{Appendix~\ref{appendix}} discusses some features of fibred integration that are used in the calculations of the main text.

Let us now describe the main results of this paper in more detail, glossing over several technical points which are addressed in the main text.
We start off in {Section \ref{sect:section2}} by introducing a notion of relation between spinors on quadratic vector spaces, i.e. linear Clifford relations. 
We then transfer this notion to pseudo-Euclidean vector bundles and find its connection with Courant algebroid relations, which is the main topic of {Section \ref{sect:section4}}. Since Dirac structures are naturally characterised as spinor line bundles, we also find a suitable notion of relation between Dirac structures descending from the Courant algebroid relation between their ambient bundles.  

Let $R$ be an isotropic subbundle of the product of anchored pseudo-Euclidean vector bundles $E_1 \to M_1$ and $E_2 \to M_2$ supported on $C \subseteq M_1 \times M_2$ with $\rho(R^\perp) \subseteq TC$. By taking the intersection of $R$ with $E_1$ times the zero section of $E_2$ we may define the kernel of $R$ and denote it by ${\rm K}_R$. We similarly define the cokernel of $R$ and denote it by ${\rm CK}_R$. 

\begin{definition}[\textbf{Definition \ref{def:Diracrel}}]
   Let $R \colon E_1 \dashrightarrow E_2$ be a Courant algebroid relation supported on a submanifold $C \subseteq M_1 \times M_2$.  Let $L_1$ and $L_2$ be Dirac structures for $E_1$ and $E_2,$ respectively. The restriction of $R$ to $L_i$ for $i=1,2$ is given by
\begin{align}
R \big\rvert_{L_i} \coloneqq \set{(e_1, e_2) \in R \ | \ e_i \in L_i} \ .  
\end{align}
Then $R$ is a \emph{relation between the Dirac structures $L_1$ and $L_2$}, or a \emph{Dirac relation}, if
\begin{align}
R \big\rvert_{L_1} = R\cap (L_1 \times L_2) \big\rvert_C + {\rm CK}_R \qquad \text{and} \qquad  R\big\rvert_{L_2} = R\cap (L_1 \times L_2) \big\rvert_C + {\rm K}_R  
\end{align}
are regular vector subbundles supported on $C.$ In that case we write $L_1\sim_R L_2$.
\end{definition}

We also show that this definition extends the previous notions of ``morphism'' between Dirac structures, such as morphisms of Manin pairs, and that it fits the description of reduction of Dirac structures.

We then turn to a description of T-duality between Dirac structures as encompassed by

\begin{proposition}[\textbf{Proposition \ref{prop:TdualDirac}}]
  Under the assumptions of \cite[Theorem 5.27]{DeFraja:2023fhe}, let $\red L{}_1$ be an invariant Dirac structure on $\red E{}_1$. Then there exists a unique Dirac structure $\red L{}_2$ on $\red E{}_2$ such that the T-duality relation $R:\red E{}_1\rel\red E{}_2$ is a Dirac relation.  
\end{proposition}

In \textbf{Section \ref{sect:section5}} we turn to the setting of anchored pseudo-Euclidean vector bundles and their spinor relations, in order to build a well-defined notion of Courant algebroid relation from the standpoint of spinors.

\begin{definition}[\textbf{Definition \ref{def:Rcliffrelation}}]
    Let $\sfS_{E_i}$ be spinor modules for $E_i$ for $i=1,2$. A vector subbundle ${\sf S}_R \subset (\sfS_{E_1}\times\sfS_{E_2})  \rvert_C$ is an \emph{$R$-Clifford relation} if it is a spinor module of the typical fibre of $\Cl(R)$, where the Clifford action is defined with the fibres of the bundle $\big(\Cl(E_1)\times \Cl(E_2)\big) \big\rvert_C$ acting on the fibres of $(\sfS_{E_1}\times\sfS_{E_2}) \rvert_C $. Elements $\nu \in\sfS_{E_1}$ and $\mu \in \sfS_{E_2}$ are \emph{$R$-related}, denoted $\nu \sim_R \mu$, if $(\nu,\mu) \in {\sf S}_R$.
\end{definition}

This definition fits the reduction of spinors discussed in \cite{Drummond:2011eq}. 
We  connect this definition with Dirac relations via

\begin{proposition}[\textbf{Proposition \ref{prop:LreliffUrelsmooth}}]
    Let $R \colon E_1 \rel E_2$ be a Courant algebroid relation supported on $C$ with an $R$-Clifford relation $\sfS_R \subset (\sfS_{E_1} \times \sfS_{E_2}) \rvert_C$. Let $L_1 \subset E_1$ and $L_2 \subset E_2$ be Dirac structures, and let $\Omega_1 \subset \sfS_{E_1}$ and $\Omega_2 \subset \sfS_{E_2}$ be the corresponding pure spinor line bundles. If $\Omega_1 \sim_R \Omega_2$, then $L_1 \sim_R L_2$.  
\end{proposition}

In order to complete the description of  Courant algebroid relations and Dirac relations in terms of spinors, we characterise them in terms of related Dirac generating operators (Definition~\ref{def:DGOrelation}) through

\begin{proposition}[\textbf{Proposition \ref{prop:DGOimpliesCArel}}]
    Let $R \subset E_1 \times \overline{E}_2$ be an almost Dirac structure supported on $C\subseteq M_1 \times M_2$ with trivial kernel and cokernel. Suppose that $\Cl(R)$ has a faithful irreducible spinor representation $\sfS$ and that $\sfS$ is a $(\sfD_1, \sfD_2)$-DGO relation, where $\sfD_i$ is a Dirac generating operator for $E_i$. Then $R$ is a Courant algebroid relation.
\end{proposition}

The converse to this result happens to be elusive in this general setting. We thus resort to proving it in the realm of T-duality relations. We first show 

\begin{proposition}[\textbf{Proposition \ref{prop:regCAonB}}]
\label{prop:5.17_intro}
    In the setting of Section~\ref{sub:TDualityFM_intro}, there is a transitive Courant algebroid $R_\cB$ over $\cB$ whose fibre is isomorphic (as a vector bundle) to a fibre of $R$, where 
\begin{equation}
\begin{tikzcd}
 R \colon (\IT\cQ_1,\red H{}_1) \ar[r,dashed] & (\IT\cQ_2,\red H{}_2)
 \end{tikzcd}
\end{equation} 
is a T-duality relation.
\end{proposition}

This can be taken as motivation for working in the setting of T-duality relations, since the existence of this bundle is not guaranteed in any other setting.
We then show

\begin{theorem}[\textbf{Theorem \ref{prop:canonicalDGOforR}}]
    Let $R\colon (\IT \cQ_1, \red H{}_1) \rel (\IT \cQ_2, \red H{}_2)$ be a T-duality relation, and let $R_\cB$ be the transitive Courant algebroid constructed in Proposition~\ref{prop:5.17_intro}; then $R_\cB$ is isomorphic to $T\cB \oplus \ccG \oplus T^*\cB$ where $\ccG$ is a bundle of quadratic Lie algebras over $\cB$ (see Section~\ref{ssec:SpinorrelationTduality}). Let $\sfS_R$ be the $R$-Clifford relation defined in \eqref{eqn:Tdualityrelspinrep}, and let $\lambda\in \det(\ccG^\star_1)$ be such that the spinor line bundle $\Omega^\lambda$ generated by $\lambda$ is invariant. Then $\sfS_R$ is a $(\de_{\red H{}_1}, \de_{\red H{}_2})$-DGO relation.
\end{theorem}

All these elements give the isomorphism of twisted cohomology that is typical of T-duality, with the notable difference that the topology of the fibres of the dual spaces are not constrained to be compact. These results thus not only perfectly fit  T-duality for torus bundles, but also Poisson-Lie T-duality where the requirement on $\lambda$ translates into the existence of an invariant divergence on one of the quotient spaces.

Motivated by these results, we formulate T-duality for Type II supergravity as

\begin{proposition}[\textbf{Proposition \ref{prop:TdualitytypeII}}]
     Let $(g_A,\red H{}_A,\phi_A, F_A)$ and $(g_B, \red H{}_B, \phi_B, F_B)$ be Type~II bosonic fields on T-dual principal circle bundles $\cQ_A \to \cB$ and $\cQ_B\to\cB$ over the same base $\cB$, respectively. Let
    \begin{equation}
    \begin{tikzcd}
        R \colon (\IT \cQ_A, \red H{}_A, V^+_{g_A}, \div_{\phi_A})\ar[r,dashed] & (\IT \cQ_B, \red H{}_B, V^+_{g_B}, \div_{\phi_B})
        \end{tikzcd}
    \end{equation}
    be a T-duality relation with $R$-Clifford relation $\sfS_R$ such that $F_A \approx_R F_B$ (see Section~\ref{sub:TDualitySUGRA}).
    Then $(g_A,\red H{}_A,\phi_A, F_A)$ satisfy the Type II supergravity equations~\eqref{eqn:typeII1}--\eqref{eqn:typeII2} for $\nu = -1$ if and only if $(g_B,\red H{}_B,\phi_B, F_B)$ satisfy the Type II supergravity equations~\eqref{eqn:typeII1}--\eqref{eqn:typeII2} for $\nu = \ii\,$.
\end{proposition}

Having developed a robust notion of Dirac relation in both usual and spinor terms, we define relations between generalised complex structures in {Section \ref{sect:section6}}.

\begin{definition}[\textbf{Definition \ref{def:generalisedcomplexrel}}]
   Let $R \colon E_1 \dashrightarrow E_2$ be a Courant algebroid relation supported on $C \subseteq M_1 \times M_2$. Let $\cJ_1$ and $\cJ_2$ be generalised complex structures on $E_1$ and  $E_2$, respectively. 
Then $R$ is a \emph{generalised complex relation} if $(\cJ_1\times\cJ_2)(R) = R$. 
\end{definition}

We describe many of the properties of generalised complex relations, which turn out to be very similar to the properties of generalised isometries shown in \cite{DeFraja:2023fhe}.
This definition perfectly fits in the spinor framework via

\begin{proposition}[\textbf{Proposition \ref{prop:UreliffJrel}}]
    Let $R\colon E_1 \rel E_2$ be a Courant algebroid relation, and let $\cJ_i$ be a generalised  complex structure on $E_i$ with associated pure spinor line bundle $\Omega_i$ for $i=1,2$. If  $\Omega_i\sim_{R} \Omega_2$, then $R$ is a generalised  complex relation.
\end{proposition}

Since T-duality is one of our main frameworks, we may use its tools to prove that a T-duality prescription exists for generalised complex structures as well through

\begin{proposition}[\textbf{Proposition \ref{prop:TdualGCS}}]
    If $\cJ_1$ is an invariant generalised complex structure on $\cQ_1$, then there is a generalised complex structure $\cJ_2$ on $\cQ_2$ such that 
    \begin{equation}
    \begin{tikzcd}
    R\colon(\IT \cQ_1, \red H{}_1, \cJ_1) \ar[r,dashed] & (\IT \cQ_2, \red H{}_2, \cJ_2)
    \end{tikzcd}
    \end{equation}
    is a generalised complex relation.
\end{proposition}

In Proposition \ref{prop:Tdualspinor} we prove that the T-duality relation also keeps track of the type of the generalised complex structures: the type of the T-dual generalised complex structure can be computed from the type of the starting generalised complex structure and data from the decomposition of spinors induced by the existence of a T-duality relation. That is, we fully recover the type change of generalised complex structures in our general T-duality framework.

We naturally complete in {Section \ref{sect:section7}} our discussion by converging on the description of relations for generalised K{\"a}hler structures.

\begin{definition}[\textbf{Definition \ref{def:genkahlerrel}}]
     Let $R\colon E_1 \rel E_2$ be a Courant algebroid relation supported on $C \subseteq M_1 \times M_2$. Let $\cK_1 = (\cJ_1^+, \cJ_1^-)$ and $ \cK_2 = (\cJ_2^+,\cJ_2^-)$ be generalised  K\"ahler structures on $E_1$ and $E_2$, respectively.
    Then $R$ is a \emph{generalised K\"ahler relation} if $R$ is a generalised complex relation for both $(E_1, \cJ^+_1) \rel (E_2, \cJ^+_2)$ and $(E_1, \cJ_1^-) \rel (E_2, \cJ_2^-)$.
    If $\cK_1$ and $\cK_2$ are generalised Calabi-Yau structures, then $R$ is a \emph{generalised Calabi-Yau relation}.
\end{definition}

This definition naturally implies that the generalised K{\"a}hler relation is also a generalised isometry between the associated generalised metrics. Moreover, it gives a natural definition for the relation between the associated bi-Hermitian structures. 

Again the realm of T-duality is rich with opportunities to test our results, which we do in

\begin{proposition}[\textbf{Proposition \ref{prop:TdualGKS}}]
    If $\cQ_1$ has an invariant generalised K\"ahler structure $\cK_1$, then there is a generalised K\"ahler structure $\cK_2$ on $\cQ_2$ such that 
    \begin{equation}
    \begin{tikzcd}
    R \colon (\IT \cQ_1, \red H{}_1, \cK_1) \ar[r,dashed] & (\IT \cQ_2, \red H{}_2, \cK_2)
    \end{tikzcd}
    \end{equation}
    is a generalised K\"ahler relation.
\end{proposition}

This result can be interpreted as encoding T-duality for $\cN=(2,2)$ supersymmetric sigma-models. A similar result holds for generalised Calabi-Yau structures as well (Proposition~\ref{prop:TdualGCY}). 

Lastly, we discuss in detail some examples of T-duality for generalised K{\"a}hler structures in our formulation, such as mirror symmetry (Example~\ref{ex:mirror}) and semi-flat Calabi-Yau metrics (Example~\ref{ex:semi-flat}). We also formulate a  new example in which, starting from a trivial $\sfT^2$-bundle over $\cB$, we twist it with a closed three-form and endow it with a generalised Calabi-Yau structure. We  construct its T-dual generalised Calabi-Yau manifold as a principal circle bundle over $\cB \times \sfT^1$ which admits genuine twisted symplectic structures (Example~\ref{ex:harmonic-symplectic}). In particular, this produces a local non-K\"ahler example of a generalised Calabi-Yau structure.

\medskip

\subsection{Acknowledgements}~\\[5pt]
We thank Jeff Streets for comments on the manuscript.
The work of T.C.D. was supported by an EPSRC Doctoral Training Partnership Award. The work of V.E.M. was supported by PNRR MUR projects PE0000023-NQSTI and in part by the GACR Grant EXPRO 19-28268X. This article is based upon work from COST Actions CaLISTA CA21109 and THEORY-CHALLENGES CA22113 supported by COST (European Cooperation in Science and Technology). 

\section{Linear Algebra of Spinors and Relations} \label{sect:section2}

We begin with a discussion on the linear algebra of relations and spinors, and introduce the language of relations for spinors.

\medskip

\subsection{Linear Relations}~\\[5pt]
\label{sub:linear_relations}
Let $(W_1,\ip{\,\cdot \, ,\, \cdot \, }_1)$ and $(W_2, \ip{\, \cdot \, ,\, \cdot \,}_2)$ be finite-dimensional real vector spaces endowed with a symmetric non-degenerate bilinear pairing. Let $W_1 \times \overline{W}_2$ be the product vector space $W_1 \times W_2$ endowed with the pairing $\ip{\, \cdot \, , \, \cdot \, }_1 - \ip{\, \cdot \, , \, \cdot \,}_2$.

\begin{definition}
    A \emph{linear relation} between $(W_1,\ip{\,\cdot \, ,\, \cdot \, }_1)$ and $(W_2, \ip{\, \cdot \, ,\, \cdot \,}_2)$, denoted $W_1 \sim_U W_2$, is a Lagrangian subspace $U \subset W_1 \times \overline{W}_2$. We also write $U:W_1\rel W_2$.
\end{definition}

\begin{example}
    Let $f\colon W_1 \to W_2$ be a linear map such that $\ip{f(w),f(w')}_2 = \ip{w, w'}_1$ for every $w, w' \in W_1$. Then $\gr(f) \subset W_1 \times \overline{W}_2$ defines a Lagrangian subspace, and hence a linear relation.
\end{example}

\begin{example}\label{ex:natural}
    Let $V_1$ and $V_2$ be finite-dimensional vector spaces and let $f \colon V_1 \to V_2$ be a linear map. Consider the vector spaces $\IV_i \coloneqq V_i \oplus V^*_i$, for $i=1,2$, endowed with the pairing 
    \begin{align}
    \ip{v_i + \alpha_i, w_i+ \beta_i}_i = \tfrac12\,(\iota_{v_i} \beta_i + \iota_{w_i} \alpha_i) \ ,
    \end{align}
    for all $v_i, w_i \in V_i$ and $\alpha_i, \beta_i \in V^*_i$, where $\iota$ denotes the canonical dual pairing (contraction) between vectors and forms. Then $\IV_1$ and $\IV_2$ are related by
    \begin{align}
        U_f \coloneqq \set{(v_1+\alpha_1, v_2+ \alpha_2) \in \IV_1 \times \IV_2 \, \, \vert \, \, v_2 = f(v_1) \ , \ \alpha_1 = f^*(\alpha_2)  } \ ,
    \end{align}
    that is a Lagrangian subspace of $\IV_1 \times \overline{\IV}_2$.
\end{example}

We shall now introduce linear relations between Lagrangian subspaces. To do this, let $U$ be a linear relation between $W_1$ and $W_2$. We introduce the \emph{kernel} and \emph{cokernel of $U$} which are  given by
\begin{align} \label{eqn:KandCK}
    {\rm K}_U \coloneqq (W_1 \times \set{0}) \cap U  \qquad \text{and} \qquad {\rm CK}_U \coloneqq (\set{0} \times W_2) \cap U \ ,
\end{align}
respectively. We also denote $\ker(U) = \rmp_1({\rm K}_U)$ and ${\rm coker}(U) = \rmp_2({\rm CK}_U)$, where $\rmp_i \colon W_1 \times W_2 \to W_i$ are the natural projections.

\begin{definition}\label{def:linearlagrel}
Let  $W_1 \sim_U W_2$ and let $V_i \subset W_i$ be a Lagrangian subspace, for $i=1,2$. Define
  \begin{align}
      U \rvert_{V_i} \coloneqq \set{(w_1, w_2) \in U \  \rvert  \ w_i \in V_i } \ ,
  \end{align}
for $i=1,2.$ Then $U$ is a \emph{linear relation} between the Lagrangian subspaces $V_1$ and $V_2$, denoted $V_1 \sim_U V_2$, if 
  \begin{align}
      U \rvert_{V_1} = U\cap (V_1 \times V_2) + {\rm CK}_U \qquad \text{and} \qquad
     U \rvert_{V_2} = U\cap (V_1 \times V_2) + {\rm K}_U \ ,
 \end{align}
as  vector subspaces.
\end{definition}

Note that the inclusion $\supseteq$ always holds for both equations. 

\begin{remark}
    If $W_1 \sim_U W_2$ and $V_1 \sim_U V_2$ for Lagrangian subspaces $V_1 \subset W_1$ and $V_2 \subset W_2$, then if $(v_1, w_2) \in U\rvert_{V_1}$, there is some $v_2 \in V_2$ such that $(v_1,v_2) \in U$.

    If $U = \gr(f)$ is the graph of a linear map $f\colon W_1 \to W_2$, then ${\rm CK}_U = \set 0$, ${\rm K}_U = \gr(f\rvert_{\ker(f)})$ and $\ker(\gr(f)) = \ker(f)$. For Lagrangian subspaces $V_i \subset W_i$, the condition that $U\rvert_{V_1} = U\cap (V_1 \times V_2)$ is equivalent to $f(V_1) \subseteq V_2$. The condition that $U\rvert_{V_2} = U \cap (V_1 \times V_2)$ is equivalent to $f^{-1}(V_2) \subseteq V_1 + \ker(f)$. In the case that $f$ is an isomorphism these conditions reduce to $f(V_1) = V_2$.
\end{remark}

\medskip

\subsection{Clifford Algebras}~\\[5pt]
\label{sub:Clifford_alg}
Let $V$ be a finite-dimensional real vector space. The vector space $$\IV = V\oplus V^*$$ has a Clifford algebra structure generated by the relation $$\texttt{v} \cdot \texttt{v} = \ip{\texttt{v},\texttt{v}} \ , $$ for $\texttt{v}\in \IV$, where $\ip{\,\cdot\,,\,\cdot\,}$ is the natural pairing from Example~\ref{ex:natural}. We denote this Clifford algebra by~$\Cl(\IV)$. 

A \emph{Clifford module} over $\Cl(\IV)$ is a vector space $\sfS$ together with an algebra homomorphism $\Psi \colon \Cl(\IV) \to {\sf End}(\sfS)$. The elements of $\sfS$ are called \emph{spinors}. A Clifford module $\sfS$ is called a \emph{spinor module} if it is irreducible, i.e. if there are no non-trivial proper subspaces $\sfS' \subset \sfS$ such that $\sfS'$ is also a Clifford module. 

\begin{remark}[\textbf{Spinor Lines}] \label{rmk:spinorlines}
Let ${\rm Lag}(\IV)$ be the set of Lagrangian subspaces of $\IV$. For any $L \in {\rm Lag}(\IV)$, the subspace $\Omega^L \subset \sfS$ defined by
\begin{align}
    \Omega^L \coloneqq \set{\psi \in \sfS \ | \ \texttt{v}\cdot \psi = 0 \ , \ \forall\, \texttt{v} \in L } 
\end{align}
is one-dimensional. For any $\psi \in \sfS$, we define its \emph{null space} $N_\psi \subset \IV$ by
\begin{align}
    N_\psi \coloneqq \set{\texttt{v} \in \IV \ | \ \texttt{v} \cdot \psi=0} \ .
\end{align}
Note that $N_\psi$ is always an isotropic subspace. We  call $\psi \in \sfS$ a \emph{pure spinor} if $N_\psi$ is a Lagrangian subspace. If $\psi$ is a pure spinor, then so is $c\, \psi$ for any $c \in \IR$. Thus each pure spinor has an associated \emph{spinor line} $\Omega^\psi:=\Omega^{N_\psi} \subset \sfS$. Hence there is a one-to-one correspondence between Lagrangian subspaces $L \in {\rm Lag}(\IV)$ and pure spinor lines $\Omega^L \subset \sfS$.
\end{remark}

\begin{example}
    The space of forms $\sfS = \midwedge^\bullet V^*$ defines a spinor module for $\Cl(\IV)$, with $\Psi ( \texttt{v})(\sfs) = \texttt{v} \cdot \sf s$, for every $\texttt{v}\in \IV$ and $\sfs \in \sfS$. Explicitly, the $\Cl(\IV)$-module structure is given via the action
\begin{align}
    (v+\alpha) \cdot \nu = \iota_v \nu + \alpha \wedge \nu \ .
\end{align}
More generally, if $W$ is any vector space with split signature bilinear form, and $L \subset W$ is any Lagrangian subspace of $W$, then $\sfS = \midwedge^\bullet L^*$ defines a spinor module for $\Cl(W)$, with $\Psi$ similarly defined.

    Consider $1 \in \midwedge^\bullet V^*$, the unit spinor. Then
    \begin{align}
        N_1 = \set{v+ \alpha\in\IV \ \vert \ \iota_v 1 + \alpha \wedge 1 = 0} = V \ ,
    \end{align}
    and hence $1$ is a pure spinor, with associated Lagrangian subspace $V$. The associated line is $\Omega^V = \IR$ as a subspace of $\midwedge^\bullet V^*$. 
    
    At the opposite extreme, let $\lambda \in \det V^*$. Then $N_\lambda = V^*$ and $\Omega^{V^*} = \det V^*$, hence $\lambda$ is a pure spinor with associated Lagrangian subspace $V^*$.
    See \cite[Section 2.5]{gualtieri:tesi} for a full description of pure spinors associated to Lagrangian subspaces of $\IV$.
\end{example}

\begin{remark}[\textbf{Mukai Pairing on Spinors}]
\label{rem:Mukai}
    There is a bilinear form on spinors, called the Mukai pairing, which is well-behaved under the spin representation. For spinors $\sfs,\sfs'\in\sfS=\midwedge^\bullet V^*$, it is defined by
    \begin{align}
        (\!(\sfs,\sfs')\!) = \big(\vartheta\,\sfs\wedge \sfs'\big)^{\rm top} \ \in \ \det V^* \ ,
    \end{align}
    where $\vartheta$ is the antiautomorphism of the Clifford algebra $\Cl(\IV)$ determined by the map of tensor products $\mathtt{v}_1\otimes\cdots\otimes\mathtt{v_k}\mapsto \mathtt{v}_k\otimes\cdots\otimes\mathtt{v_1}$, and $(\,\cdot\,)^{\rm top}$ indicates taking the top degree component of a form. See~\cite[Section~2.4]{gualtieri:tesi} for further details and properties of this pairing.
\end{remark}

\medskip

\subsection{Linear Clifford Relations}\label{ssec:linearspinorrel}~\\[5pt]
Let $V$ and $W$ be  vector spaces. We can define two distinct $\Cl(\IV)\times \Cl(\IW)$-module structures on the product $\midwedge^\bullet V^* \times \midwedge^\bullet W^*$ through
\begin{align}\label{eqn:productcliffordaction}
    (\texttt{v} , \texttt{w})\cdot(\nu ,\mu) = (\texttt{v}\cdot \nu ,\pm\, \texttt{w}\cdot \mu) \ ,
\end{align}
for $\texttt{v} \in \IV$ and $\texttt{w}\in \IW$. Then
\begin{align}
    (\texttt{v} , \texttt{w})^2\cdot(\nu ,\mu) = (\texttt{v}^2\cdot \nu , \texttt{w}^2\cdot \mu) \ ,
\end{align}
and so this defines a Clifford action if and only if $\texttt{v}^2=\ip{\texttt{v},\texttt{v}}_\IV = \ip{\texttt{w},\texttt{w}}_\IW = \texttt{w}^2$. 

Let $U \subset \IV \times \overline{\IW}$ be an isotropic subspace. Then $\ip{\texttt{v},\texttt{v}}_\IV = \ip{\texttt{w},\texttt{w}}_{\IW}$ for each $u = (\texttt{v}, \texttt{w}) \in U$. It follows that $U$ has a Clifford algebra structure defined by $$u^2 = \ip{u,u}:=\ip{\texttt{v},\texttt{v}}_\IV = \ip{\texttt{w},\texttt{w}}_\IW \ . $$ Denote by $\Cl^+(U)$ the algebra $\Cl(U)$ acting on $\midwedge^\bullet V^* \times \midwedge^\bullet W^*$ with $+$ parity in Equation \eqref{eqn:productcliffordaction}, and by $\Cl^-(U)$ the algebra with opposite parity. In general, $\midwedge^\bullet V^* \times \midwedge^\bullet W^*$ is a reducible representation of $\Cl^+(U)$ and~$\Cl^-(U)$.

\begin{definition}\label{def:RCliffrel}
    A subspace ${\sf S}_U \subset \midwedge^\bullet V^* \times \midwedge^\bullet W^*$ is a \emph{U-Clifford relation of parity $i\in \set{-,+} = \IZ_2$} if $\sfS_U$ is a spinor module for $\Cl^i(U)$. Elements $\nu \in \midwedge^\bullet V^* $ and $\mu \in \midwedge^\bullet W^*$ are said to be \emph{$U$-related}, denoted $\nu \sim_U \mu$, if $(\nu,\mu) \in {\sf S}_U$.
\end{definition}

We show that this definition leads to a notion of \enquote{isomorphism} of spinor modules through

\begin{proposition}\label{prop:Uisomorphism}
    Let $\rmp_i$ be the projections from $\IV \times \IW$ to the first and second factor for $i=1,2$ respectively, and denote by the same symbols the projections from $\midwedge^\bullet V^* \times \midwedge^\bullet W^* $. Let $\sfS_U$ be a $U$-Clifford relation of parity $j$, with $\rmp_i(U) \neq \set{0}$ and $\rmp_i(\sfS_U) \neq \set{0}$ for $i=1,2$. Then
    \begin{align}
        \sfS_U \cap (\set{0} \times \midwedge^\bullet W^* ) = \sfS_U \cap (\midwedge^\bullet V^* \times \set{0} ) = \set{(0,0)} \ .
\end{align}
\end{proposition}

\begin{proof}
    Suppose $(0, \mu) \in \sfS_U$ for some $\mu \in \midwedge^\bullet W^*$. By the invariance of $\sfS_U$, we obtain
    \begin{align}
        \sfS' \coloneqq \mathsf{Cl}^j(U) \cdot (0, \mu) \ \subset \ \sfS_U \ .
    \end{align} 
    The set $\sfS'$ is invariant under $\mathsf{Cl}^j(U)$, thus by the irreducibility of $\sfS_U$ either $\sfS' = \sfS_U$ or $\sfS' = \set{(0,0)}$. By the assumptions $\rmp_2(U) \neq \set{0}$ and $\rmp_1(\sfS_U) \neq \set{0}$, the latter holds, hence $\mu = 0$. The other equality follows similarly.
\end{proof}

\begin{corollary}\label{cor:URinjective}
    If $\sfS_U$ is a $U$-Clifford relation with $\rmp_i(U) \neq \set{0}$ and $\rmp_i(\sfS_U) \neq \set{0}$ for $i=1,2$, and $(\nu, \mu), (\nu, \mu')\in \sfS_U$, then $\mu = \mu'$.
\end{corollary}

When the  kernel ${\rm K}_U$ and cokernel ${\rm CK}_U$ of $U$ are non-trivial, there is more freedom in choosing $\sfS_U$ to account for their presence. To address this freedom, we can consider only the injective part of the relation. We define $\red{\rm CK}_U = (\set 0 \times \IW ) \cap (U / {\rm K}_U)$ and
\begin{align}
    U^\times = (U/{\rm K}_U)/\red{\rm CK}_U \ \subset \ \IV/{\rm ker}(U) \times \IW / {\rm coker}(U) \ .
\end{align}
We denote by $[\texttt{v},\texttt{w}]$ the equivalence class of $(\texttt{v},\texttt{w}) \in U$, and by $[(\texttt{v},\texttt{w})]$ the projection under the quotient map $U \to U^\times$. This projection descends to the Clifford algebra, since $\ip{u+k, u+k'} = \ip{u,u}$ for any $u\in U$ and $k,k'$ in ${\rm K}_U$ or ${\rm CK}_U$. The quotient is then taken componentwise on the tensor product, i.e. $[u\otimes u'] = [u]\otimes [u']$. This therefore yields the Clifford algebra $\Cl(U^\times)$.

Suppose that there exists an action of $\Cl(U^\times)$ on $\midwedge^\bullet V^* \times \midwedge^\bullet W^*$. We say that this action is \emph{compatible with $\Cl^i(U)$} if for each $[\texttt{v},\texttt{w}] \in \Cl(U^\times)$ there is an element $(\texttt{v},\texttt{w})\in \Cl^i(U)$ with $[(\texttt{v},\texttt{w})] = [\texttt{v},\texttt{w}]$ such that
\begin{align}
    (\texttt{v},\texttt{w}) \cdot (\nu, \mu) = [\texttt{v},\texttt{w}] \cdot (\nu,\mu) \ .
\end{align}

\begin{definition}\label{def:RCliffrel2}
    A subspace $\sfS_U \subset \midwedge^\bullet V^* \times \midwedge^\bullet W^*$ is a \emph{weak $U$-Clifford relation of parity $i$} if there exists an action of $\Cl(U^\times)$ on $\midwedge^\bullet V^* \times \midwedge^\bullet W^*$ compatible with $\Cl^i(U)$ for which $\sfS_U$ is a spinor module.
\end{definition}

\begin{proposition}\label{prop:strong=weak}
    A $U$-Clifford relation defines a weak $U$-Clifford relation.
\end{proposition}

\begin{proof}
    Suppose that $\sfS_U$ is a $U$-Clifford relation in the sense of Definition \ref{def:RCliffrel}. For $(k,0) \in {\rm K}_U$ and $(\nu, \mu)\in \sfS_U$, we obtain
    \begin{align}
        (k,0)\cdot (\nu, \mu) = (k\cdot \nu, 0) \ .
    \end{align}
    Then $k\cdot \nu = 0$ by Proposition \ref{prop:Uisomorphism}, hence ${\rm K}_U$ acts trivially on $\sfS_U$, and similarly for $\red{\rm CK}_U$. Hence $\sfS_U$ descends to a spinor module for $\Cl(U^\times)$.
\end{proof}

\begin{example}\label{ex:barf}
    Suppose  $\boldsymbol{f} \colon \IV \to \IW$ is an isomorphism given by $$\boldsymbol{f} = f^{-1} \oplus f^* \ , $$ where $f \colon W \to V$ is an isomorphism. Let \smash{$\sfS_{\gr(\boldsymbol{f})}$} be a $\gr(\boldsymbol{f})$-Clifford relation of parity $+$. Then, noting that $\Cl(\gr(\boldsymbol{f}))$ is the graph of the induced isomorphism from $\Cl(\IV)$ to $\Cl(\IW)$, 
    \begin{align}
        \big(v+\alpha, \boldsymbol{f}(v+\alpha)\big) \cdot (\nu, \mu) = \big(\iota_v\nu + \alpha \wedge \nu\,,\, \iota_{f^{-1}(v)}\mu + f^*(\alpha) \wedge \mu \big) \ ,
    \end{align}
    for every \smash{$(\nu , \mu) \in \sfS_{\gr(\boldsymbol{f})}$}, and $v+\alpha \in \IV$.
    
    If $\nu$ is a zero-form then $\mu$ is also a zero-form, otherwise there is $w\in W$ such that $f(w)\cdot\mu = 0$, contradicting Corollary \ref{cor:URinjective}. It follows that, on zero-forms, \smash{$\sfS_{\gr(\boldsymbol{f})}$} is generated by $(1,\kappa)$ for some $\kappa\in \IR\setminus\set{0}$. 
    
    Similarly, if $\nu\in V^*$ then $\mu \in W^*$ and $\mu(w) = \kappa\,\nu(f(w))$. Let $w\in W$ be such that $\nu(f(w)) = 1$. Then
    \begin{align}
        (\nu,\mu) = (\nu,f^* \nu)\cdot\big(f(w), w\big)\cdot (\nu,\mu) = (\nu, f^*\nu) \cdot \big(1, \mu(w)\big) = (\nu, \kappa \, f^*\nu) \ .
    \end{align}
    Thus \smash{$\nu \sim_{\gr(\boldsymbol{f})} \kappa\, f^*\nu$}, and hence $\mu = \kappa\, f^*\nu$ by Corollary \ref{cor:URinjective}. 
    
    By induction, it follows that \smash{$\sfS_{\gr(\boldsymbol{f})} = \gr(\kappa\, f^*)$}.
\end{example}

\begin{remark}\label{rem:kappa}
    In light of Example~\ref{ex:barf}, if $\sfS_U\subset \midwedge^\bullet V^* \times \midwedge^\bullet W^*$ is a $U$-Clifford relation, then
\begin{align}
    \sfS_U^\kappa = \set{(\nu, \kappa \, \nu_U) \ \vert \  (\nu, \nu_U) \in \sfS_U}
\end{align}
 for each $\kappa \in \IR \setminus \set 0$ is also a $U$-Clifford relation of the same parity. 
This is similar to the ambiguity addressed in \cite[Part~II]{Cortes:2019roa} regarding the choice of connection on the spinor bundle.

In order to remove this ambiguity, let $n=\dim_\IR V$ and consider
\begin{align}
    \det (\sfS_U^\kappa)^* = \set{(\mu, \tfrac{1}{\kappa^{2^n}} \mu_U) 
 \ \vert \ (\mu, \mu_U) \in \det \sfS_U^*} \ .
\end{align}
It follows that $\sfS_U \otimes |\det \sfS_U^*|^{\frac{1}{2^n}} = \sfS_U^\kappa \otimes |\det (\sfS_U^\kappa)^*|^{\frac{1}{2^n}}$ for every $\kappa \in \IR \setminus \set 0$, where $|\det \sfS_U^*|^{\frac{1}{2^n}}$ indicates the one-dimensional vector space of $\frac1{2^n}$-densities on $\sfS_U$, i.e. maps $\mu:\det\sfS_U\to\IR$ such that \smash{$\mu(\lambda\,\sfs) = |\lambda|^{\frac1{2^n}}\,\mu(\sfs)$} for all $\lambda\in\IR$ and $\sfs\in\det \sfS_U$. This motivates
\end{remark}

\begin{definition}\label{def:canonicalRclifford}
    If $\sfS_U$ is a weak $U$-Clifford relation of parity $i$ for $U \subset \IV \times \overline{\IW}$, the \emph{canonical $U$-Clifford relation of parity $i$} is
    \begin{align}
    \cS_U = \sfS_U \otimes |\det \sfS_U^* |^{\frac{1}{2^n}}
\end{align}
where $n = \dim_\IR V$. As a spinor module, $\Cl^i(U)$ acts trivially on the second component of $\cS_U$.

If $\Omega_\IV \subset \midwedge^\bullet V^*$ and $\Omega_\IW \subset \midwedge^\bullet W^*$ are pure spinor lines, then $\Omega_\IV$ and $\Omega_\IW$ are \emph{$U$-related} if
\begin{align}
    (\Omega_\IV \times \Omega_\IW)\otimes |\det \sfS_U^* |^{\frac{1}{2^n}} \ \subset \ \cS_U \ ,
\end{align}
 denoted as $\Omega_\IV \sim_U \Omega_\IW$.
\end{definition}

\begin{remark}
Definition \ref{def:canonicalRclifford} does not depend on the choice of $\sfS_U$. Indeed, suppose there exists some $U$-Clifford relation $\sfS_U$. Let $$(u_1, u_2) \otimes (\mu_1, \mu_2) \ \in \ (\Omega_\IV \times \Omega_\IW)\otimes |\det \sfS_U^* |^{\frac{1}{2^n}} \ , $$ and $\kappa \in \IR \setminus \set 0$. Then
\begin{align}
    (u_1, u_2) \otimes (\mu_1, \tfrac 1\kappa\, \mu_2) = (u_1, \tfrac 1 \kappa\, u_2) \otimes (\mu_1, \mu_2) \ \in \ (\Omega_\IV \times \Omega_\IW)\otimes |\det \sfS_U^* |^{\frac{1}{2^n}} \ .
\end{align}
Hence $(\Omega_\IV \times \Omega_\IW)\otimes |\det (\sfS_U^\kappa)^* |^{\frac{1}{2^n}} \subseteq (\Omega_\IV \times \Omega_\IW)\otimes |\det \sfS_U^* |^{\frac{1}{2^n}}$. The opposite inclusion holds similarly, therefore
\begin{align}
    (\Omega_\IV \times \Omega_\IW)\otimes |\det \sfS_U^* |^{\frac{1}{2^n}} = (\Omega_\IV \times \Omega_\IW)\otimes |\det (\sfS_U^\kappa)^* |^{\frac{1}{2^n}} \ .
\end{align}
\end{remark}

\begin{proposition}\label{prop:LreliffUrel}
    Let $U \colon \IV \rel \IW$ be a linear relation. Let $L_1 \subset \IV$ and $L_2 \subset \IW$ be Lagrangian subspaces, with $\Omega_1 \subset \midwedge^\bullet V^*$ and $\Omega_2 \subset \midwedge^\bullet W^*$ the corresponding pure spinor lines. Suppose there is a $U$-Clifford relation $\sfS_U \subset \midwedge^\bullet V^* \times \midwedge^\bullet W^*$.
    If $\Omega_1 \sim_U \Omega_2$, then $L_1 \sim_U L_2$. 
\end{proposition}

\begin{proof}
    Take $(\psi_1,\psi_2) \in \Omega_1 \times \Omega_2$ such that $(\psi_1,\psi_2) \in \sfS_U$ and consider $(\texttt{l}, \texttt{w})\in U|_{L_1}$. Then
    \begin{align}
        (\texttt{l}, \texttt{w}) \cdot (\psi_1, \psi_2) = (0, \pm\,\texttt{w}\cdot \psi_2) \ .
    \end{align}
    Since $(\psi_1,\psi_2) \in \sfS_U$, by Proposition \ref{prop:Uisomorphism} it follows that $\texttt{w}\cdot \psi_2 = 0$ and hence $\texttt{w}\in L_2$, i.e. $U|_{L_1} \subset L_1 \times L_2$. Similarly, $U|_{L_2} \subset L_1 \times L_2$, and hence $L_1 \sim_U L_2$.
\end{proof}

\begin{remark}\label{rem:fullrankcase}
    Suppose that $L_1 \sim_U L_2$ and that either $\rmp_1(U) \supseteq L_1$ or $\rmp_2(U) \supseteq L_2$; when $U$ is a Lagrangian subspace with trivial kernel and cokernel, this holds automatically. Assume without loss of generality the latter holds. Let $\psi_1\in\Omega_1$.
    By Proposition \ref{prop:Uisomorphism}, there is some $\psi_2'$ such that $(\psi_1, \psi_2') \in \sfS_U$. Since $\rmp_2(U) \supseteq L_2$, for each $\texttt{l}_2 \in L_2$ there is $\texttt{l}_1 \in L_1$ such that $(\texttt{l}_1, \texttt{l}_2) \in U\cap (L_1 \times L_2)$, and hence
    \begin{align}
        (\texttt{l}_1, \texttt{l}_2) \cdot (\psi_1, \psi_2') = (0, \pm\,\texttt{l}_2 \cdot \psi_2')
    \end{align}
    must vanish. Thus $\psi_2' \in \Omega_2$ so that $\Omega_1 \sim_U \Omega_2$.
\end{remark}

\begin{remark}\label{rem:weakcase}
    For a weak $U$-Clifford relation $\sfS_U$, given $(\texttt{l}, \texttt{w}) \in U|_{L_1}$ and $[\texttt{l},\texttt{w}] \in U|_{L_1}^\times$, in general the equality
    \begin{align}\label{eqn:L1compatibility}
        [\texttt{l},\texttt{w}] \cdot (\psi_1, \psi_2) = (\texttt{l},\texttt{w}) \cdot (\psi_1, \psi_2)
    \end{align}
    does not hold.
    In this light, when $U \colon \IV \rel \IW$ is a linear relation and $\IV$ (resp. $\IW$) is equipped with a Lagrangian subspace $L_1$ (resp. $L_2$), we say that $\sfS_U$ is an \emph{$L_1$-weak} (resp. \emph{$L_2$-weak}) \emph{$U$-Clifford relation} if $\sfS_U$ is a weak $U$-Clifford relation satisfying the condition \eqref{eqn:L1compatibility} (resp. the corresponding condition with $L_2$ in place of $L_1$). Then the conclusion of Proposition~\ref{prop:LreliffUrel} holds.
\end{remark}

\begin{remark}
    Thus far we have been working over the field of real numbers $\IR$. Everything discussed up to this point in this section holds equally well over the field of complex numbers $\IC$. For later use, let us introduce some notation in the complex case. If $W$ is a real vector space, its complexification $W\otimes_\IR \IC$ is denoted by $W^\IC$. For a subspace $V \subset W^\IC$, we denote the complex conjugate of $V$ by $\overline V$.
\end{remark}

\section{Courant Algebroids, Dirac Structures and Generalised Complex Geometry} \label{sect:section3}
We follow \cite{gualtieri:tesi, Alekseev2001} for the brief review of this section, discussing how Courant algebroids and spinor bundles are deeply interconnected, as well as some aspects of generalised complex geometry. This builds on the notions from Section~\ref{sub:Clifford_alg}.

\medskip

\subsection{Spinors and Courant Algebroids}~\\[5pt]
Let $M$ be a smooth manifold with ${\rm dim}_\IR(M)=n$. 
Any pseudo-Euclidean vector bundle $E\to M$ of rank $2n$ with split signature pairing always admits a ${\sf Spin}(n,n)$-structure \cite{gualtieri:tesi}. 

For instance, consider the generalised tangent bundle $$\IT M \coloneqq TM \oplus T^*M$$ equipped with the canonical split signature metric $\ip{\, \cdot\, , \, \cdot\,}$ induced by the duality pairing:
\begin{align}
    \ip{X+\alpha, Y+\beta} \coloneqq \tfrac12\,(\iota_X \beta+ \iota_Y \alpha) \ ,
\end{align}
for all $X+ \alpha, Y + \beta \in \sfgamma(\IT M)$.
Then for any $X + \xi \in \sfGamma(\IT M)$ and $\omega \in \sfOmega^\bullet(M),$ the action 
\begin{align} \label{eqn:cliffaction}
    (X + \xi) \cdot \omega \coloneqq \iota_X \omega + \xi \wedge \omega
\end{align}
satisfies
\begin{align}
   (X + \xi)^2 \cdot \omega = \ip{X+ \xi,X + \xi} \, \omega \ .
\end{align}
This makes $\sfGamma(\midwedge^\bullet T^*M)$ into an irreducible Clifford module, so that
\begin{align}
    \Cl(TM \oplus T^*M) \simeq {\sf End}(\midwedge^\bullet T^*M) \ .
\end{align}

Let $H \in \sfOmega_{\tt cl}^3(M)$ be a closed three-form. Introduce the \emph{twisted differential} $\de_H \colon \sfOmega^\bullet(M) \to \sfOmega^\bullet(M)$ given by 
\begin{align}
    \de_H \omega \coloneqq \de \omega + H \wedge \omega \ ,
\end{align}
for any $\omega \in \sfOmega^\bullet(M).$ We may then define the twisted Dorfman bracket $\cbrak{ \,\cdot \,,\, \cdot\,}_H$ as a derived bracket of $\de_H$ and the Clifford action \eqref{eqn:cliffaction} given by
\begin{align}\label{eqn:Dorfderived}
    \cbrak{e_1, e_2}_H \cdot \omega \coloneqq [[\de_H , e_1 \,\cdot\,], e_2 \,\cdot\, ] \, \omega \ ,
\end{align}
for all $e_1, e_2 \in \sfGamma(\IT M)$ and $\omega \in \sfOmega^\bullet(M)$. The bracket $[\,\cdot\,,\,\cdot\,]$ on the right-hand side of \eqref{eqn:Dorfderived} is the graded commutator of operators in ${\sf End}(\midwedge^\bullet T^*M)$, i.e. if $A,B \in {\sf End}(\midwedge^\bullet T^*M)$ are of degree $|A|$ and $|B|$ respectively, and $\omega \in \sfOmega^\bullet(M)$, then
\begin{align}
    [A,B] \cdot \omega= A\cdot B\cdot \omega - (-1)^{|A|\, |B|}\, B\cdot A\cdot \omega \ .
\end{align}
Explicitly
\begin{align}\label{eqn:Hbracket}
    \llbracket X + \alpha, Y + \beta \rrbracket_H = [X,Y] + \pounds_X \beta - \iota_Y\, \de \alpha + \iota_Y\, \iota_X H \ ,
\end{align}
where $\pounds_X$ denotes the Lie derivative along a vector field $X\in\sfGamma(TM)$. When $H=0$ we drop the subscript in the notation for the bracket.

For any pseudo-Euclidean vector bundle $E \to M$ with inner product $\ip{\,\cdot \, , \, \cdot\,}$ of split signature, we can always find a Clifford bundle $\Cl(E) \to M$ admitting an irreducible Clifford module $\sfS \to M,$ with the Clifford action giving an isomorphism $\Cl(E) \simeq {\sf End}(\sfS).$ The role of the twisted differential $\de_H$ in the general case is taken by

\begin{definition} \label{def:dgo}
    A \emph{Dirac generating operator} $\sfD$ for the pseudo-Euclidean vector bundle $E \to M$ is an odd operator on $\sfGamma(\sfS)$ satisfying
    \begin{enumerate}
        \item $[\sfD, f] \in \sfGamma(E),$ for all $f \in C^\infty(M),$ i.e. $[\sfD,f]$ is the operator given by the Clifford action of some section of $E;$ \\[-3mm]
        \item $[[\sfD, e_1], e_2] \in \sfGamma(E),$ for all $e_1, e_2 \in \sfGamma(E);$ and \\[-3mm]
        \item $\sfD^2 \in C^\infty(M),$ i.e. $\sfD^2$ acts as a multiplication by a function.
    \end{enumerate}
\end{definition}

The structures on the generalised tangent bundle $\IT M$ are replaced by

\begin{definition}\label{def:CourantAlg}
A \emph{Courant algebroid} is a quadruple\footnote{Throughout this paper we will denote the Courant algebroid $(E,\llbracket\,\cdot\,,\,\cdot\,\rrbracket,\langle\,\cdot\,,\,\cdot\,\rangle ,\rho)$ simply by $E$ when there is no ambiguity. If there is more than one Courant algebroid $E$ involved in the discussion, we label its operations with a subscript $_E\,$.} $(E,\llbracket\,\cdot\,,\,\cdot\,\rrbracket,\langle\,\cdot\,,\,\cdot\,\rangle ,\rho)$, where $E$ is a vector bundle over a manifold $M$
with a fibrewise non-degenerate pairing $\langle\,\cdot\,,\,\cdot\,\rangle \in\mathsf{\Gamma}(\midodot^2E^*),$ a vector bundle morphism 
$\rho \colon E \rightarrow TM $
called the \emph{anchor}, and an $\IR$-bilinear bracket operation   
\begin{equation} 
\llbracket\,\cdot\,,\,\cdot\,\rrbracket \colon \mathsf{\Gamma}(E) \times \mathsf{\Gamma}(E) \longrightarrow \mathsf{\Gamma}(E)
\end{equation}
called the \emph{Dorfman bracket}, which together satisfy
\begin{enumerate}[label=(\roman{enumi})]
    \item $\pounds_{\rho(e)}\langle e_1, e_2\rangle  = \langle\llbracket e,e_1\rrbracket, e_2\rangle  + \langle e_1, \llbracket e, e_2\rrbracket\rangle  \ ,$  \label{eqn:metric1} \\[-3mm]
    \item $\langle\llbracket e, e\rrbracket, e_1\rangle  = \tfrac{1}{2}\, \pounds_{\rho(e_1)}\langle e, e\rangle  \ ,$  \label{eqn:metric2} \ and \\[-3mm]
    \item $\llbracket e, \llbracket e_1,e_2 \rrbracket  \rrbracket  = \llbracket \llbracket e,e_1\rrbracket  ,e_2 \rrbracket  +  \llbracket e_1, \llbracket e,e_2 \rrbracket  \rrbracket  \ , $ \label{eqn:Jacobiid}
\end{enumerate}
for all $e,e_1, e_2 \in \mathsf{\Gamma}(E)$. 

A Courant algebroid is \emph{regular} if its anchor $\rho:E\to TM$ has constant rank, and \emph{transitive} if $\rho$ is surjective.
\end{definition}

\begin{remark}
The Jacobi identity~\ref{eqn:Jacobiid} together with the pairing preserving condition~\ref{eqn:metric1} imply that the anchor $\rho$ is a bracket homomorphism:
\begin{align}
\rho(\llbracket e_1, e_2\rrbracket)= [\rho(e_1), \rho(e_2)] \ ,
\end{align}
for all $e_1 , e_2 \in \mathsf{\Gamma}(E)$, where the bracket on the right-hand side is the usual Lie bracket of vector fields in $\sfGamma(TM)$. We refer to \cite[Section 2]{Kotov:2010wr} for a complete account of all the main properties of Courant algebroids.
\end{remark}

The general result connecting Definitions \ref{def:dgo} and \ref{def:CourantAlg} can  be formulated as

\begin{theorem}[\!\!\!\textbf{\cite[Corollary 2.3]{Alekseev2001}}] \label{thm:CourantDirac}
    Let $\sfD$ be a Dirac generating operator for the pseudo-Euclidean vector bundle $E \to M.$ Then there is a canonical Courant algebroid structure on $E$ whose anchor and Dorfman bracket are respectively given by
    \begin{align}
        \pounds_{\rho(e)} f \coloneqq [[\sfD, f], e] \ ,
    \end{align}
    for all $e \in \sfGamma(E)$ and $f \in C^\infty(M),$ and
    \begin{align}
        \cbrak{e_1, e_2} \coloneqq [[\sfD, e_1], e_2] \ ,
    \end{align}
    for all $e_1, e_2 \in \sfGamma(E).$
\end{theorem}

A converse to Theorem~\ref{thm:CourantDirac} was first formulated in \cite{Alekseev2001}. Its circumscription to {regular} Courant algebroids is given by

\begin{theorem}[\!\!\!\textbf{\cite[Theorem 41]{Cortes:2019roa}}] \label{thm:DGOconverse}
    Let $E$ be a regular Courant algebroid with split signature pairing. Every spinor bundle $\sfS$ over $\Cl(E)$  locally admits a Dirac generating operator.
\end{theorem}

Since any vector bundle with split signature pairing admits a spinor bundle, we can thus construct a Dirac generating operator over any regular split signature Courant algebroid. We will make use of this result in Section \ref{ssec:twisteddifferential}.

\begin{remark}\label{rem:Courantcond}
For any Courant algebroid $E$ over $M$ there is a map $\rho^*\colon T^*M\to E$ given by
\begin{align}\label{eqn:rhostar}
\langle\rho^*(\alpha),e\rangle  \coloneqq \iota_e\, \rho^{\tt t}(\alpha) \ ,
\end{align}
for all $\alpha\in \mathsf{\Omega}^1(M)$ and $e\in\mathsf{\Gamma}(E)$, where $\rho^{\tt t}\colon T^*M\to E^*$ is the transpose of $\rho.$ Then the symmetric part of the Dorfman bracket can be expressed as
\begin{align}
    \llbracket e_1, e_2\rrbracket + \llbracket e_2, e_1\rrbracket = \rho^*\,\de \, \ip{e_1 , e_2} \ ,
\end{align}
for all $e_1,e_2\in\sfGamma(E)$.

The subbundle $\ker(\rho)^\perp\subset E$ orthogonal to $\ker(\rho)$ with respect to the pairing  coincides with the image
\begin{align}
    \ker(\rho)^\perp = \Im(\rho^*) \ .
\end{align}
The pairing preserving condition~\ref{eqn:metric1} of Definition~\ref{def:CourantAlg} implies
\begin{align}
    \rho\circ\rho^* = 0 \ , 
\end{align}
from which it follows that the kernel of the anchor is a coisotropic subbundle of $E$:
\begin{align}
    \ker(\rho)^\perp \ \subset \ \ker(\rho) \ .
\end{align}
\end{remark}

Remark~\ref{rem:Courantcond} motivates

\begin{definition}
A Courant algebroid $E$ over $M$ is \emph{exact} if the chain complex
\begin{align}\label{eqn:exCou}
0 \longrightarrow T^*M \xlongrightarrow{\rho^*} E \xlongrightarrow{\rho} TM \longrightarrow 0
\end{align}
is a short exact sequence of vector bundles.
\end{definition}

By definition, an exact Courant algebroid is transitive and $\ker(\rho)=\ker(\rho)^\perp$ is a Lagrangian subbundle of $E$.

An isotropic splitting $$\sigma\colon TM \longrightarrow E$$ of the short exact sequence \eqref{eqn:exCou} induces a closed three-form $H_\sigma \in \mathsf{\Omega}_{\tt cl}^3(M)$ given by
\begin{align}\label{eqn:severaclass}
    H_\sigma(X,Y,Z) =  \ip{\llbracket\sigma(X),\sigma(Y)\rrbracket, \sigma(Z)}
\end{align}
for $X,Y,Z \in \mathsf{\Gamma}(TM)$. Here $H_\sigma = H$ from the twisted differential $\de_H$, hence $\llbracket\,\cdot\,,\, \cdot\, \rrbracket = \llbracket\,\cdot\,,\, \cdot\, \rrbracket_H$ from Equation \eqref{eqn:Dorfderived}. An isotropic splitting $\sigma$ induces an isomorphism $\sigma \oplus \tfrac{1}{2}\,\rho^*$ between the exact Courant algebroid $E$ and the twisted standard Courant algebroid $(\IT M, H_\sigma)$ with anchor the canonical projection $\rmp_1:\IT M\to TM$.

In a splitting $E \simeq \IT M$, automorphisms $\Phi$ of exact Courant algebroids are given by diffeomorphisms $\varphi:M\to M$ and $B$-field transformations as $$\Phi = \big(\varphi_* \oplus (\varphi^{-1})^*\big) \circ \e^B \ . $$ Here $B$ is a two-form seen as a map $B \colon TM \to T^*M$, and $\e^{B}\, \colon \IT M \rightarrow \IT M$ is the vector bundle map given by
\begin{align}
\e^{B}\, (X + \alpha) = X + B(X) + \alpha \ .
\end{align}

\begin{remark}[\textbf{\v{S}evera's Classification}\cite{Severa-letters}]
The isomorphism classes of exact Courant algebroids $E \to \IT M $ determined by choices of splitting of the short exact sequence \eqref{eqn:exCou} are in one-to-one correspondence with cohomology classes in ${\sf H}^3( M,\IR).$  
Recall that a closed three-form $H_\sigma$ is defined by Equation \eqref{eqn:severaclass}.
Since the difference between two splittings $\sigma-\sigma'$ defines a two-form $B\in \mathsf{\Omega}^2(M)$ via $(\sigma-\sigma')(X) = \rho^*(\iota_X B)$, the three-form $H_\sigma$ is shifted by $\de B$ under a change of splitting. Hence there is a well-defined cohomology class $[H_\sigma]\in {\sf H}^3(M,\IR)$ associated to the exact sequence \eqref{eqn:exCou} which completely determines the exact Courant algebroid. This is called the \v{S}evera class of the exact Courant algebroid. A general discussion can be found in~\cite[Section 2.2]{Garcia-Fernandez:2020ope}.
\end{remark}

\medskip

\subsubsection{Reduction of Courant Algebroids}~\\[5pt]
 We briefly explain the reduction of Courant algebroids by isotropic subbundles, following \cite{Zambon2008reduction}. We begin with the notion of basic sections for an isotropic subbundle $K$ of a Courant algebroid $E$. Let $K^\perp\supset K$ denote the orthogonal bundle of $K$ with respect to the pairing on $E$.

\begin{definition}
    The space of sections of $K^\perp$ which are \emph{basic with respect to} $K$ is given by
    \begin{align}
        \mathsf{\Gamma}_{\text{bas}}(K^\perp) \coloneqq \set{e \in \mathsf{\Gamma}(K^\perp) \ \vert \ \llbracket \mathsf{\Gamma}(K) ,e \rrbracket \subset \mathsf{\Gamma}(K)} \ .
    \end{align}
\end{definition}
 
If $\mathsf{\Gamma}_{\text{bas}}(K^\perp)$ spans $K^\perp$ pointwise, then $K\subset E$ is involutive and $\rho(K)\subset TM$ is an integrable distribution.

The reduction of Courant algebroids by a foliation is the content of the Bursztyn-Cavalcanti-Gualtieri-Zambon Theorem, see \cite{Bursztyn2007reduction, Zambon2008reduction}.

\begin{theorem}[\!\!\!\textbf{\cite[Theorem~3.7]{Zambon2008reduction}}]
\label{thm:reduction}
    Let $E$ be a Courant algebroid over $M$, and let $K$ be an isotropic subbundle of $E$.
    Suppose that $\rho(K^\perp) = TM$, that $K^\perp$ is  spanned pointwise by basic sections and that the quotient of $M$ by the foliation $\cf F$ integrating  $\rho(K)$ is smooth with surjective submersion $\varpi \colon M \to \cf Q$. Then there is a Courant algebroid $\red E$ over $\cf Q$ fitting into the pullback diagram
    \[
\begin{tikzcd}
K^\perp / K \arrow{rr}{} \arrow[swap]{dd}{} & & \red E \arrow[]{dd}{} \\ & & \\
M \arrow[]{rr}{\varpi} & & \cf Q 
\end{tikzcd}
\]
of vector bundles. The reduced Courant algebroid $\red E$ is exact if $E$ is exact. 
\end{theorem}

Theorem~\ref{thm:reduction} yields a pointwise isomorphism
\begin{align}\label{eqn:bethisom}
    \mathscr{J}_{q, m} \colon \red E{}_q \longrightarrow (K^\perp/K)_m
\end{align}
for any $q\in\cQ$ and $m \in \varpi^{-1}(q)$.
We denote by $\natural \colon K^\perp\to \red E$ the vector bundle morphism covering $\varpi \colon M\to\cQ$ given by the quotients by $K$ and $\cF$.
We denote the anchor map inherited from $\rho \colon E\to TM$ by $\red\rho \colon \red E\to T\cQ$ through the commutative diagram
 \[
\begin{tikzcd}
K^\perp \arrow{rr}{\rho} \arrow[swap]{dd}{\natural\,} & &  TM \arrow[]{dd}{\varpi_*} \\ & & \\
\red E \arrow[]{rr}{\red\rho} & & T\cf Q 
\end{tikzcd}
\]

We shall now characterise splittings of an exact Courant algebroid that descend to the reduced Courant algebroid.

\begin{definition}\label{defn:adaptedsplittings}
Let $E$ be an exact Courant algebroid over $M$ with an isotropic subbundle $K$. A splitting $\sigma \colon  TM \to E$ is \emph{adapted to $K$} if
\begin{enumerate}[label= (\alph{enumi}), labelwidth=0pt]
\item the image of $\sigma$ is isotropic,\\[-3mm]
\item $\sigma(TM)\subset K^\perp \ ,$ \  and \\[-3mm]
\item $\sigma(X)\in \mathsf{\Gamma}_{\mathrm{bas}}(K^\perp)$ for any vector field $X$ on $M$ which is projectable to $\cQ$.
\end{enumerate} 
\end{definition}

If $\sigma$ is an adapted splitting, then
\begin{align}
    \sigma(T\cF) = \sigma\big(\rho(K)\big) \ \subseteq \ K \ ,
\end{align}
by~\cite[Remark~5.2]{Zambon2008reduction}.
The condition of having basic sections spanning $K^\perp$ pointwise for the reduction to be possible also implies the existence of adapted splittings, as shown in \cite[Proposition 5.5]{Zambon2008reduction}. 
In this case a splitting $\sigma:TM\to E$ induces an isotropic splitting $\red\sigma:T\cQ\to\red E$ of the reduced Courant algebroid $\red E$ over $\cQ$ by~\cite[Proposition~5.7]{Zambon2008reduction}.

For an exact Courant algebroid $E$ with a given splitting $\sigma \colon TM\to E$, we will make use of the $\sigma$-adjoint action~\cite[Definition~5.21]{DeFraja:2023fhe}
\begin{align} 
\ad^\sigma:\sfgamma(E)\times\sfgamma(E)\longrightarrow\sfgamma(E)
\end{align}
which is defined by
\begin{align} \label{eqn:sigmatwistact}
    {\sf ad}^\sigma_e\, e' \coloneqq \llbracket e , e' \rrbracket - \rho^*\,{\sigma}^*\llbracket e , {\sigma}(\rho(e')) \rrbracket \ ,
\end{align}
for all $e, \, e' \in \sfgamma(E).$ If $X\in\sfGamma(TM)$, we denote $\mathsf{ad}_X^\sigma := \mathsf{ad}_{\sigma(X)}^\sigma$.
See \cite[Appendix A]{deFraja2025Ricci} for some properties of the $\sigma$-adjoint action ${\sf ad}^\sigma$. A further useful property is

\begin{proposition}\label{prop:adaptedsplittingadaction}
     If $K\subset E$ is an isotropic subbundle such that $K^\perp$ is spanned pointwise by basic sections, and $\sigma$ is a splitting adapted to $K$, then $$\cbrak{\sfgamma(K), \sfgamma(K^\perp)} = \ad^\sigma_{\sfgamma(K)} \sfgamma(K^\perp) \ . $$
\end{proposition}

\begin{proof}
    Let $e \in \sfgammabas(K^\perp)$ and $k\in \sfgamma(K)$. Consider 
    \begin{align}
        \cbrak{k,e} - \ad^\sigma_k\, e = \rho^*\, \sigma^* \cbrak{k, \sigma(\rho( e))} \ .
    \end{align}
    Since
    \begin{align}
        \rho(\cbrak{k, \sigma (\rho (e))}) = [\rho(k), \rho(e)] = \rho(\cbrak{k,e}) \ \in \  \rho(K) \ ,
    \end{align}
    it follows that $\cbrak{k, \sigma( \rho( e))} \in \sfgamma(K)$. Since $(\sigma \circ \rho)|_K = \unit_K$ by \cite[Lemma 2.34]{DeFraja:2023fhe}, it follows that $\cbrak{k, \sigma (\rho( e))} = \sigma \, \rho(\cbrak{k, \sigma (\rho (e))})$. The result now follows from $\sigma^* \circ \sigma = 0$, using that basic sections span $K^\perp$ pointwise.
\end{proof}

\medskip

\subsection{Dirac Structures}~\\[5pt]
We give a brief introduction to Dirac structures.

\begin{definition}
 Let $(E,\llbracket\,\cdot\,,\,\cdot\,\rrbracket,\langle\,\cdot\,,\,\cdot\,\rangle,\rho)$   be a Courant algebroid over a smooth manifold $M$. An \emph{almost Dirac structure} $L$ is a Lagrangian subbundle $L \subset E.$ An almost Dirac structure $L$ is \emph{integrable} or a \emph{Dirac structure} if it is also involutive with respect to the Dorfman bracket. A \emph{Dirac structure $L$ supported on $N$} is an involutive Lagrangian subbundle $L \subset E \rvert_N,$ where $N \subset M$ is a smooth (embedded) submanifold such that $\rho(L) \subseteq TN.$ 
\end{definition}

\begin{remark}
If $L$ is a Dirac structure, the triple $(L, \rho \rvert_L , \llbracket\,\cdot\,,\,\cdot\,\rrbracket_L)$ defines a Lie algebroid over $M$, since
\begin{align}
    \llbracket e_1, e_2\rrbracket + \llbracket e_2, e_1\rrbracket = \rho^*\,\de \, \ip{e_1 , e_2} = 0 \ ,
\end{align}
for any $e_1, e_2 \in \mathsf{\Gamma}(L),$ i.e. the restriction of the Dorfman bracket to $L$ is skew-symmetric and satisfies the Jacobi identity. For a Dirac structure $L$ supported on a submanifold $N$, the restriction of the anchor $\rho$ is defined in order to provide the data for a Lie algebroid $L$ over $N.$
\end{remark}

\begin{example}[\textbf{Foliations}]
Let $M$ be a smooth manifold endowed with a regular foliation $\cF.$ Then the (untwisted) standard Courant algebroid $\IT M$ admits a Dirac structure $$L= T\cF \oplus {\rm Ann}(T\cF)$$ characterised by the involutivity condition
\begin{align} \label{eqn:Diracfoliation}
    \llbracket X + \alpha , Y + \beta \rrbracket = [X,Y] +\iota_X\, \de \beta - \iota_Y\, \de \alpha \ \in \ \mathsf{\Gamma}(L) \ ,
\end{align}
for all $X + \alpha, Y + \beta \in \mathsf{\Gamma}(L),$ where ${\rm Ann}(T\cF)$ is the annihilator of $T\cF$ in $T^*M$. The bracket \eqref{eqn:Diracfoliation} is skew-symmetric and satisfies the Jacobi identity, i.e. it is a Lie bracket. Thus $L$ together with the projection $\rmp_1$ to $TM$ and the restriction of the Dorfman bracket to $L$ defines a Lie algebroid.
\end{example}

\begin{example}[\textbf{Pre-Symplectic Structures}] \label{eg:symplecticDirac}
Let $M$ be an even-dimensional manifold endowed with a two-form $\omega.$ Then the graph
\begin{align}
    L^\omega = \gr(\omega) = \set{X + \iota_X\omega \ \vert \ X \in TM}
\end{align}
is an almost Dirac structure on the standard Courant algebroid $TM \oplus T^*M$. 
To check involutivity, for all $X, Y \in \mathsf{\Gamma}(TM)$ we compute
\begin{align}
  \llbracket X + \iota_X  \omega , Y + \iota_Y \omega \rrbracket &= [X,Y] +\pounds_X \, \iota_Y \omega - \iota_Y\, \de\, \iota_X \omega  \\[4pt]
  &= [X,Y] + \iota_{[X,Y] }\omega + \iota_Y\, \iota_X \, \de \omega \ .
\end{align}
Hence $L^\omega$ is  involutive, i.e.~it is a Dirac structure, if and only if $\de \omega = 0.$ In other words, $L^\omega$ is a Dirac structure if and only if $(M, \omega)$ is a pre-symplectic manifold.     
\end{example}

\medskip

\subsubsection{Dirac Reduction}~\\[5pt] \label{subsect:DiracReduction}
We discuss reduction of Dirac structures following \cite[Proposition~4.1]{Zambon2008reduction}. 

\begin{proposition} \label{prop:reducedKsub-bundle}
 Let $E$ be a Courant algebroid over $M$ endowed with an isotropic subbundle $K$ satisfying the assumptions of Theorem~\ref{thm:reduction}. Let $L$ be a subbundle of $E$ such that $L \cap K^\perp$ has constant rank. If
 \begin{align} \label{eqn:reductionW}
 \llbracket \mathsf{\Gamma}(K) , \mathsf{\Gamma}(L \cap K^\perp) \rrbracket  \ \subset \ \mathsf{\Gamma}(L+K) \ ,  
 \end{align}
 then $L$ descends to a subbundle $$\red L \coloneqq \natural (L \cap K^\perp + K)$$ of the reduced Courant algebroid $\red E$ over $\cQ.$
\end{proposition}

\begin{corollary} \label{cor:Diracreduct}
If $L$ is a Dirac structure for $E$, then $\red L$ is a Dirac structure for $\red E$.
\end{corollary}

\begin{proof}
Consider any $\red e, {\red e}^\prime \in \mathsf{\Gamma}(\red L).$ They can be lifted to $e, e^\prime \in \mathsf{\Gamma}_{\rm bas}(L \cap K^\perp),$ because $L \cap K^\perp$ has constant rank. Then
\begin{align}
    \llbracket e, e^\prime \rrbracket_E \ \in \ \mathsf{\Gamma}_{\rm bas}(L \cap K^\perp) \ ,
\end{align}
since $L$ is involutive and 
\begin{align}
\llbracket\mathsf{\Gamma}_{\rm bas}(K^\perp), \mathsf{\Gamma}_{\rm bas}(K^\perp)  \rrbracket_E \ \subseteq \ \mathsf{\Gamma}_{\rm bas}(K^\perp) \ .
\end{align}
Therefore 
\begin{align} \label{eqn:reducedbracket}
\natural(\llbracket e, e^\prime \rrbracket_E) =: \llbracket \red e, {\red e}^\prime \rrbracket_{\red E} \ \in \ \mathsf{\Gamma}(\red L) \ .
\end{align}
Moreover, $\red L$ is Lagrangian by construction.
\end{proof}

\begin{remark}\label{rmk:reductedKsub-bundle}
If $K \subset L \subset K^\perp,$ then Equation \eqref{eqn:reductionW} becomes
\begin{align}
 \llbracket \mathsf{\Gamma}(K) , \mathsf{\Gamma}(L) \rrbracket   \ \subset \ \mathsf{\Gamma}(W) \ .    
\end{align}
Any subbundle $L$ satisfying the above properties is the pointwise span of basic sections $\mathsf{\Gamma}_{\mathrm{bas}}(L)$. That is, there are elements $\ell_i\in L$ such that $\llbracket\mathsf{\Gamma}(K),\ell_i\rrbracket \subset \mathsf{\Gamma}(K)$ and $L=\mathrm{Span}_\IR\lbrace \ell_i \rbrace$ pointwise. The proof can be found in \cite[Theorem 4.1]{calvo2010deformation}.\footnote{In their notation, $S=K$, $D=D^S=L$, and canonicity means existence of basic sections. Note that we do not assume that $L$ is a Dirac structure, and hence it is not involutive, thus we only have the implication c)$\implies$a) of  \cite[Theorem 4.1]{calvo2010deformation} and not the full equivalence.}
\end{remark}

\begin{example}[\textbf{Pre-Symplectic Reduction}] \label{eg:symplecticred}
Let us start from Example \ref{eg:symplecticDirac}. We also assume that $M$ is foliated by $\cF$ such that $\cQ = M / \cF$ is smooth, with surjective submersion denoted by $\varpi \colon M \to \cQ.$ By setting $$K = T\cF \oplus \set{0} \ , $$ the standard Courant algebroid $TM \oplus T^*M$ reduces to the standard Courant algebroid $T\cQ \oplus T^*\cQ.$ We shall check under which circumstances Proposition~\ref{prop:reducedKsub-bundle} holds for the Dirac structure $L^\omega$, under the additional assumption that $\ker(\omega)=T\cF$.

In this case $$K^\perp = TM \oplus \ann(T\cF) \ ,$$ hence
\begin{align}
    L^\omega\cap K^\perp =\set{X + \iota_X \omega \ \vert \ X \in T\cF^\omega} \ ,
\end{align}
where $T\cF^\omega$ is the orthogonal complement of $T\cF$ with respect to the closed two-form $\omega.$ Thus
\begin{align}
    \rk_\IR(L^\omega\cap K^\perp) = \rk_\IR(T\cF^\omega) \ ,
\end{align}
hence $L^\omega\cap K^\perp$ has constant rank. 

In order to check that the condition \eqref{eqn:reductionW} holds, we compute
\begin{align}
 \llbracket X_\cF , X + \iota_X \omega \rrbracket = [X_\cF, X] + \iota_{[X_\cF, X]}\omega + \iota_X\, \pounds_{X_\cF}\omega   
\end{align}
for all $X_\cF\in\sfGamma(T\cF)$ and $X \in \sfGamma(T\cF^\omega).$ This is an element in $\sfGamma(L^\omega+K)$ if and only if $\pounds_{X_\cF}\omega =0.$ Hence Equation \eqref{eqn:reductionW} holds if and only if $\omega$ is constant along the leaves of $\cF.$ 

The reduced subbundle $\red L$ is given by
\begin{align}
  {\red L} = \natural(L^\omega \cap K^\perp + K) =\set{ {\red X} + \iota_{\red X} {\red \omega} \ \vert \ {\red X} \in T\cQ} \ ,  
\end{align}
since $T\cQ$ and $T\cF^\omega$ are pointwise isomorphic, where $\red \omega$ is induced by the restriction $\omega \rvert_{T\cF^\omega} = \varpi^* {\red \omega}.$ Then $\red L$ is a Dirac structure on $T\cQ \oplus T^*\cQ,$ as shown in Corollary \ref{cor:Diracreduct}.
\end{example}

\medskip

\subsection{Generalised Complex Structures}~\\[5pt]
\label{sub:gencomplex}
We introduce the basic notions of generalised complex geometry, following \cite{Hitchin:2003cxu} and \cite[Section 4]{gualtieri:tesi}. 

\begin{definition}
Let $(E,\llbracket\,\cdot\,,\,\cdot\,\rrbracket,\langle\,\cdot\,,\,\cdot\,\rangle,\rho)$ be an exact Courant algebroid over $M$ with $\rk_\IR(E) = 4 n,$ where $\mathrm{dim}_\IR(M)=2n.$ A \emph{generalised almost complex structure} on $E$ is an automorphism $\cJ \in \mathsf{Aut}(E)$ covering the identity such that $\cJ^2 = - \unit$ and 
\begin{align}
 \ip{\cJ(e), \cJ(e^\prime)} = \ip{e, e^\prime} \ ,   
\end{align}
for all $e, \, e^\prime \in E.$ A generalised almost complex structure $\cJ$ is \emph{integrable} or a \emph{generalised complex structure} if either of its $\pm \, \ii$-eigenbundles $L$ or $\overline{L}$ is involutive with respect to the Dorfman bracket~$\llbracket \, \cdot \, , \, \cdot \, \rrbracket.$ 
\end{definition}

\begin{remark}
This definition is tantamount to a reduction of the structure group $\sfO(2n, 2n)$ of $E$ to ${\sf U}(n,n).$ The integrability condition for one of the eigenbundles implies the integrability of the other one, since they are complex conjugate to each other. Then both eigenbundles $L$ and $\overline{L}$ are Dirac structures, since they are Lagrangian by construction. As discussed in \cite{gualtieri:tesi}, generalised complex structures are in one-to-one correspondence with complex Dirac structures satisfying the condition $$L \cap \overline{L} = \set{0} \ . $$ 
\end{remark}

\begin{example}[\textbf{Complex Manifolds}] \label{ex:complexmanifolds}
Let $M$ be endowed with an almost complex structure $J.$ Then the standard Courant algebroid $\IT M = TM \oplus T^*M$ inherits a generalised almost complex structure 
\begin{equation} \label{eqn:Jcomplex}
 \cJ =
 \begin{pmatrix}
 -J & 0 \\
 0 & J^{\tt t}
 \end{pmatrix}
\end{equation}
where $J^{\tt t} \colon T^*M \to T^*M$ is the transpose of $J.$ The orthogonality of $\cJ$ is equivalent to the condition
$\cJ^{\tt t} = - \cJ .$ 

Notice that $\cJ$ is a generalised complex structure if and only if $J$ is a complex structure on $M.$ This is shown as follows. The $+\,\ii$-eigenbundle $L$ of $\cJ$ defined in Equation \eqref{eqn:Jcomplex} is given by
\begin{align}
 L = T_{0,1} \oplus \mathrm{Ann}(T_{0,1}) = T_{0,1} \oplus T^*_{1,0} \ \subset \ \IT M^{\IC} \ ,   
\end{align}
which is Lagrangian, where $T_{0,1}$ is the $-\,\ii$-eigenbundle of $J$ given by
\begin{align}
    T_{0,1} = \set{X + \ii\, J(X) \ \vert \ X \in TM} \ \subset \ TM^{\IC} \ ,
\end{align}
and $T_{1,0}=\overline{T_{0,1}}$ is the $+\,\ii$-eigenbundle.
If $L$ is involutive, then
\begin{align}\label{eqn:DorfmanJ}
 \llbracket X+ \xi , Y + \eta \rrbracket = [X,Y] + \pounds_X \eta - \iota_Y\, \de \xi \ \in \ \mathsf{\Gamma}(T_{0,1} \oplus T^*_{1,0})   
\end{align}
for any $X + \xi , \, Y + \eta \in \mathsf{\Gamma}(L),$ which implies that $[X,Y] \in \mathsf{\Gamma}(T_{0,1})$ for any $X, \, Y \in \mathsf{\Gamma}(T_{0,1}).$ Conversely, if $J$ is integrable then the Dorfman bracket \eqref{eqn:DorfmanJ} can be written as
\begin{align}
 \llbracket X+ \xi , Y + \eta \rrbracket = [X,Y] + \iota_X \, \overline{\del} \eta -\iota_Y \, \overline{\del} \xi     
\end{align}
where $\de=\partial + \overline{\partial}$ and $\del$ is the Dolbeault differential induced by the integrability of $J.$ It follows that $\iota_X \, \overline{\del} \eta -\iota_Y \, \overline{\del} \xi \in \mathsf{\Gamma}(T^*_{1,0})$ by definition of $\del,$ 
 and hence $L$ is involutive with respect to the Dorfman bracket. 
\end{example}

\begin{example}[\textbf{Symplectic Manifolds}] \label{ex:symplecticmanifolds}
 Let $(M, \omega)$ be a symplectic manifold. Then $\omega$ determines a generalised complex structure on $\IT M$ given by 
 \begin{equation} \label{eqn:omegacomplex}
 \cJ =
 \begin{pmatrix}
 0 & -\omega^{-1} \\
 \omega & 0
 \end{pmatrix}
\end{equation}
whose $+\,\ii$-eigenbundle is 
\begin{align}
L = \set{X - \ii \, \iota_X \omega \ \vert \ X \in TM} \ \subset \ \IT M^{\IC} \ .    
\end{align}
One shows that $L$ is involutive by an identical calculation to that of Example~\ref{eg:symplecticDirac}.
\end{example}

Let $L\subset E^\IC$ be a Lagrangian subbundle. Recall from Remark \ref{rmk:spinorlines} that, at each point $m \in M$, there is a pure spinor line $\Omega_m$ associated to the fibre $L_m$ given by
\begin{align}
    \Omega_m = \set{\psi \in \sfS_m^\IC \ \vert \ \Psi(e) \psi = 0 \ , \ \forall\, e\in L_m} \ .
\end{align} 
This defines a complex line bundle $\Omega$ over $M$, which equivalently defines a generalised almost complex structure $\cJ$ when $L\cap \overline L = \set 0$. We may thus interchangably use $\cJ$, $L$ or $\Omega$ when discussing generalised complex structures. The pure spinor line bundle $\Omega$ defines a filtration of $\midwedge^\bullet T^*M^\IC$, where $\Omega^{n} = \Omega$ and $$\Omega^{n-k} = \midwedge^k \overline{L} \cdot \Omega^{n}$$ for $k=1,\dots,n$; we set $\Omega^k = \set{0}$ for $k<0$ and $k>n$. If $L$ defines a generalised almost complex structure $\cJ$,  this becomes a grading
\begin{align}\label{eqn:GCSgrading}
    \midwedge^\bullet T^*M^\IC = \Omega^{n} \oplus \Omega^{n-1} \oplus \cdots \oplus \Omega^{0}
\end{align}
since then $L \cap \overline{L} = \set 0$.

Let $H\in\sfOmega_{\tt cl}^3(M)$ be a closed three-form. The $H$-twisted integrability of a generalised almost complex structure $\cJ$ on $\IT M$ (or a Lagrangian subbbundle in general) can be formulated in terms of the Dirac generating operator $\de_H$: by \cite[Theorem 3.38]{gualtieri:tesi}, $\cJ$ is integrable if and only if $$\de_H\big(\sfGamma(\Omega^{n})\big) \ \subset \ \sfgamma(\Omega^{n-1}) \ . $$ If $\cJ$ is integrable, then $\de_H(\sfgamma(\Omega^k)) \subset \sfgamma(\Omega^{k-1}) \oplus \sfgamma(\Omega^{k+1})$ for each $k\in\set{0,\dots,n}$, and there is a decomposition $\de_H = \partial_H + \overline{\partial}_H$ where \cite[Theorem 4.23]{gualtieri:tesi}
\begin{align}
    \partial_H \colon \sfgamma(\Omega^k) \longrightarrow \sfgamma(\Omega^{k-1}) \qquad \text{and} \qquad \overline{\partial}_H \colon \sfgamma(\Omega^k) \longrightarrow \sfgamma(\Omega^{k+1}) \ .
\end{align}

There is also the notion of the type of a generalised complex structure $\cJ$ on a real $2n$-dimensional manifold $M$, which returns an integer $k\in \set{0,\dots,n}$ at each point, and defines a lower semicontinuous function on $M$.

\begin{definition}
    The \emph{type} of $\cJ$, denoted ${\rm type}(\cJ)$, is the complex codimension of $\rho(L)$ in $TM^\IC$: $${\rm type}(\cJ) = n - {\rm rk}_\IC\big(\rho(L)\big) \ . $$
\end{definition}

In the case of the standard Courant algebroid, for any $m\in M$ an element $\psi$ of $\Omega_m$ can be written as \cite[Theorem 4.8]{gualtieri:tesi}
\begin{align}
    \psi =  \exp(B+\ii\,\omega) \wedge \Theta
\end{align}
for real two-forms $B$ and $\omega$, and $\Theta = \theta_1 \wedge \cdots \wedge \theta_k$ for some linearly independent complex one-forms $\theta_1 , \ldots , \theta_k$. The condition that $L\cap \overline L = \set 0$ can be formulated as the requirement that $$\omega^{n-k} \wedge \Theta \wedge \overline{\Theta} \neq 0 \ . $$ The integer $k = {\rm type}(\cJ)$ is  the type of the generalised complex structure $\cJ$.

\medskip

\subsubsection{Reduction of Generalised Complex Structures}~\\[5pt]
\label{subsec:gencomplexreduction}
In this subsection we assume that the exact Courant algebroid $E \to M$ satisfies the assumptions of Theorem \ref{thm:reduction}.

\begin{proposition}[\textbf{Generalised Complex Reduction}] \label{prop:reducedcomplexdirac}
Let $E$ be a reducible exact Courant algebroid over $M$ endowed with a generalised almost complex structure $\cJ,$ with $+\,\ii$-eigenbundle $L \subset E^{\IC},$ such that $\cJ(K) \cap K^\perp \subset K$ has constant rank.  Suppose that
\begin{align} \label{eqn:complexdiracred}
\llbracket \mathsf{\Gamma}(K^\IC), \mathsf{\Gamma}(L \cap K^\perp {}^\IC)\rrbracket \ \subset \ \mathsf{\Gamma}(L + K^\IC) \ .
\end{align}
Then the reduced Courant algebroid $\red E$ over $\cQ$ inherits a generalised almost complex structure $\red \cJ$ with $+\,\ii$-eigenbundle $$\red L = \natural(L \cap K^\perp{}^\IC + K^\IC) \ . $$
\end{proposition}

\begin{proof}
By Proposition \ref{prop:reducedKsub-bundle}, the condition \eqref{eqn:complexdiracred} guarantees the existence of the subbundle  $\red L  \subset \red E^{\IC}.$ In order to show that $\red L$ is the $+\,\ii$-eigenbundle of a generalised almost complex structure, we need  $\red L \cap \overline{\red L} = \set{0}.$ This follows from the assumptions that $\cJ(K) \cap K^\perp \subset K$  
 has constant rank by using \cite[Lemma 5.1]{Bursztyn2007reduction}. 
\end{proof}

\medskip

\subsection{Generalised K\"{a}hler Structures}~\\[5pt]
We now introduce the notion of generalised K\"ahler structures on Courant algebroids. First, we define generalised metrics.

\begin{definition}[\textbf{Generalised Metrics}]\label{defn:generalisedmetric}
A \emph{generalised metric} on a Courant algebroid $E$ is an automorphism $\tau\in\mathsf{Aut} (E)$ covering the identity with $\tau^2 = \unit$ which, together with the pairing $\ip{\,\cdot\,,\,\cdot\,}$, defines a positive-definite fibrewise metric on $E$ given by
\begin{align}\label{eqn:generalisedmetrictau}
    \cG (e, e') = \ip{e, \tau(e')} \ ,
\end{align}
for $e, e'\in E$. Equivalently, a generalised metric on $E$ is a maximally positive-definite subbundle $V^+ \subset E$  (with respect to $\ip{\,\cdot \, , \, \cdot\,}$); the subbundles $V^+$ and $V^- \coloneqq (V^+)^\perp$ give an orthogonal decomposition $$E = V^+\oplus V^-$$ into the $\pm1$-eigenbundles
\begin{align}
    V^\pm = \ker(\tau\mp\unit) \ .
\end{align}
\end{definition}

\begin{remark} If $E$ is an exact Courant algebroid, then the restriction of the anchor to $V^+$ is an isomorphism $\rho|_{V^+}:V^+\to TM$. In a splitting $E\simeq\IT M$, any generalised metric $V^+$ can be written as
\begin{align}
    V^+ = \gr(g+b) = \{X+\iota_X g+\iota_Xb \ \vert \ X\in TM\} \ ,
\end{align}
for a Riemannian metric $g$ on $M$ and a Kalb-Ramond field $b\in\sfOmega^2(M)$. The complement $V^-$ is then given by $V^-=\gr(-g+b).$
\end{remark}

\begin{definition}[\textbf{Generalised K\"ahler Structures}]\label{def:GKstructure}
    Let $E\to M$ be an exact Courant algebroid with $\rk_\IR(E)=4n$, where $\dim_\IR(M)=2n$. A \emph{generalised almost K\"ahler structure} on $E$ is a pair of generalised almost complex structures $(\cJ^+, \cJ^-)$ such that
    \begin{enumerate}
        \item $\cJ^+\, \cJ^- = \cJ^-\, \cJ^+ \ , $ \ and \\[-3mm]
        \item $\tau = - \cJ^+\, \cJ^-$ defines a generalised metric.
    \end{enumerate} 
    We sometimes denote the pair $(\cJ^+,\cJ^-)$ by $\cK$. A generalised almost K\"ahler structure $\cK = (\cJ^+,\cJ^-)$ is \emph{integrable} or a \emph{generalised K\"ahler structure} when both $\cJ^\pm$ are generalised complex structures.

    A \emph{generalised Calabi-Yau structure} on $E$ is a generalised K\"ahler structure whose associated pure spinor line bundles are holomorphically trivial, i.e. there are nowhere vanishing sections $\psi^\pm$ of the line bundles $\Omega^\pm$  associated to $\cJ^\pm$ such that $$\de_H \psi^\pm = 0 \ . $$
    The Mukai pairings on these sections are required to be related as
    \begin{align}
        (\!(\psi^+,\overline{\psi^+}\,)\!) = \pm\,(\!(\psi^-,\overline{\psi^-}\,)\!) \ .
    \end{align}
\end{definition}

\begin{example}[\textbf{K\"ahler Manifolds}]
    A classical K\"ahler structure on a manifold $M$ is a triple $(g,J,\omega)$ of a Riemannian metric $g$, a complex structure $J$ and a symplectic form $\omega$ such that $\omega=g\,J$. Taking $\cJ^+ = \cJ_J$ (Example~\ref{ex:complexmanifolds}) and $\cJ^- = \cJ_\omega$ (Example~\ref{ex:symplecticmanifolds}), it follows that $(\cJ^+,\cJ^-)$ is a generalised K\"ahler structure on $(\IT M,0)$. 
    
    If the classical K\"ahler structure is a Calabi-Yau structure, then the associated holomorphic volume form $\Theta$ and K\"ahler form $\omega$ satisfy $$\omega^n = \frac{\ii^n\, n!}{2^n} \ \Theta \wedge \overline \Theta \ , $$ and hence define a generalised Calabi-Yau structure with 
    \begin{align}
        \psi^+ = \Theta \qquad \text{and} \qquad \psi^- = \exp(\,\ii\,\omega) \ .
    \end{align}
    Here $\mathrm{type}(\cJ^+)=n$ whereas $\mathrm{type}(\cJ^-)=0$.
\end{example}

\begin{remark}
If the manifold $M$ is compact, then generalised Calabi-Yau structures are shown in~\cite[Theorem~4.13]{Apostolov:2022} to be actually K\"ahler Calabi-Yau structures. In Example~\ref{ex:harmonic-symplectic} below we will produce a local non-K\"ahler example.
\end{remark}

We shall now discuss the link, shown in \cite{gualtieri:tesi}, between generalised K{\"a}hler structures and bi-Hermitian structures.

\begin{definition}[\textbf{Bi-Hermitian Structures}]
    A \emph{bi-Hermitian structure} on a real $2n$-dimensional manifold $M$ is a triple $(g, J^+, J^-)$ of a Riemannian metric $g$ and two complex structures $J^\pm$, such that both pairs $(g, J^+)$ and $(g, J^-)$ are Hermitian structures on $M$. A \emph{bi-Hermitian structure with $b$-field} on $M$ is a quadruple $(g,b,J^+,J^-)$ of a bi-Hermitian structure $(g,J^+,J^-)$ and a Kalb-Ramond field $b \in \sfOmega^2(M)$, such that the fundamental two-forms $\omega_{\pm} = g\, J^{\pm}$ satisfy
    \begin{align}
        - \de^c_+ \omega_+ = \de^c_- \omega_- = \de b \ ,
    \end{align} 
    where $\de^c_\pm \coloneqq J^\pm{}^{-1} \circ \de \circ J^\pm$. A bi-Hermitian structure with $b$-field $(g,b,J^+,J^-)$ is $H$\emph{-twisted} if the fundamental two-forms satisfy
     \begin{align} \label{eqn:deomegaH}
        - \de^c_+ \omega_+ = \de^c_- \omega_- = H + \de b \ ,
    \end{align} 
    for some $H \in \sfOmega^3_{\tt cl}(M)$. 
    
    All of these bi-Hermitian structures are $\sfU(n) {\times} \sfU(n)$-structures on $M$. In the following we will refer to all of them simply as bi-Hermitian structures for the sake of brevity.
\end{definition}

An equivalent definition of generalised K\"ahler structures is now provided by

\begin{theorem}[\!\!\!\textbf{\cite[Theorem 6.37]{gualtieri:tesi}}]
    Let $M$ be a real $2n$-manifold and take a closed three-form $H\in \sfOmega_{\tt cl}^3(M)$. The bi-Hermitian structure $(g,b,J^+,J^-)$ of a Riemannian metric $g$ on $M$, a Kalb-Ramond field $b\in \sfOmega^2(M)$, and two complex structures $J^\pm$ on $M$ preserving $g$ is equivalent to a generalised almost K\"ahler structure by 
    \begin{align}
        \cJ^\pm = \frac 12\, \begin{pmatrix}
            \unit & 0 \\
            b & \unit
        \end{pmatrix}\,
        \begin{pmatrix}
            J^+ \pm J^- & -(\omega_+^{-1} \mp \omega_-^{-1}) \\
            \omega_+ \mp \omega_- & -(J^+{}^{\tt t} \pm J^-{}^{\tt t})
        \end{pmatrix}\,
        \begin{pmatrix}
            \unit & 0 \\
            -b & \unit
        \end{pmatrix} \ .
    \end{align}
    Then $\cK = (\cJ^+, \cJ^-)$ is a generalised K\"ahler structure on $(\IT M, H)$ if and only if Equation \eqref{eqn:deomegaH} is satisfied.
\end{theorem}

This result shows that there is a one-to-one correspondence between generalised K{\"a}hler structures on $M$ and bi-Hermitian structures on $M$. The first instance in which generalised K{\"a}hler structures appeared, in the form of bi-Hermitian structures, was in \cite{Gates:1984nk} where the action functional of a two-dimensional $\cN=(2,2)$ sigma-model is written using the parametrisation $(g,b)$ of the generalised metric $\tau=-\cJ^+\,\cJ^-$ and the closed three-form $H$. The condition that this action functional is invariant under $\cN=(2,2)$ supersymmetry transformations is equivalent to Equation \eqref{eqn:deomegaH}.

\section{Courant Algebroid and Dirac Relations} 
\label{sect:section4}

Building on the notions of linear relations from Section~\ref{sub:linear_relations}, in this section we briefly review Courant algebroid relations and their uses in setting up a notion of T-duality, following~\cite{Vysoky2020hitchiker,DeFraja:2023fhe}. We then use this to give a suitable notion of relation between Dirac structures, and establish conditions under which Dirac structures are related by T-duality.

\medskip

\subsection{Courant Algebroid Relations}\label{ssec:CArelations}~\\[5pt]
We give a brief introduction to Courant algebroid relations.
For two Courant algebroids $E_1$ and $E_2$ over smooth manifolds $M_1$ and $M_2$, respectively, the product $$\big(E_1\times E_2\,,\,\llbracket\,\cdot\,,\,\cdot\,\rrbracket \,,\,\langle\,\cdot\,,\,\cdot\,\rangle\,,\,\rho\big)$$ is a Courant algebroid over $M_1\times M_2$, with the Courant algebroid structures defined by
\begin{align}
    \rho(e_1, e_2) & \coloneqq \big(\rho_{E_1}(e_1), \rho_{E_2}(e_2)\big) \ , \\[4pt] \ip{(e_1, e_2)\,,\, (e_1', e_2')} & \coloneqq \ip{e_1, e_1'}_{E_1} \circ \rmp_1 + \ip{e_2, e_2'}_{E_2} \circ \rmp_2 \ , \label{eqn:pairingprod} \\[4pt]
    \llbracket(e_1,e_2)\,,\,(e_1',e_2')\rrbracket &\coloneqq \big(\llbracket e_1,e_1'\rrbracket_{E_1}\,,\,\llbracket e_2, e_2'\rrbracket_{E_2}\big) \ ,
\end{align}
where $\rmp_i\colon M_1\times M_2 \to M_i$ are the projection maps for $i=1,2$; for the sake of brevity, we will usually not write the projections explicitly in the following.
Denote by $\overline{E}_i$ the Courant algebroid $(E_i,\llbracket\,\cdot\,,\,\cdot\,\rrbracket_{E_i} ,-\langle\,\cdot\,,\,\cdot\,\rangle_{E_i},\rho_{E_i})$.

\begin{definition}
Let $E_1$ and $E_2$ be Courant algebroids over $M_1$ and $M_2$ respectively.

A \emph{Courant algebroid relation from $E_1$ to $E_2$} is a Dirac structure $R \subseteq E_1 \times \overline{E}_2$ supported on a submanifold $C\subseteq M_1 \times M_2$ such that $\rho(R^\perp) \subset TC,$ denoted $R\colon  E_1 \dashrightarrow E_2.$

If $C$ is the graph of a smooth map $\varphi \colon M_1 \to M_2$, then $R$ is a \emph{Courant algebroid morphism from $E_1$ to $E_2$ over $\varphi$}, denoted $R\colon E_1 \rightarrowtail E_2$.

If $R$ is further the graph of some vector bundle map $\Phi \colon E_1 \to E_2$ covering $\varphi\colon M_1 \to M_2$, then $\Phi$ is a \emph{classical Courant algebroid morphism over $\varphi$}.

By viewing a Courant algebroid relation $R$ as a subbundle of $E_2 \times \overline{E}_1$, we obtain the \emph{transpose relation} $R^\top\colon E_2 \dashrightarrow E_1$, whose support is denoted $C^\top\subseteq M_2\times M_1$, and similarly the \emph{transpose morphism} $R^\top\colon E_2 \rightarrowtail E_1$ when $\varphi$ is a diffeomorphism. 
\end{definition}

We will also need a notion of sections related by a Courant algebroid relation. If $E$ is a Courant algebroid over $M$ and $L\subset E$ is a subbundle supported on a submanifold $C\subset M$, denote by $\sfGamma(E; L)$ the $C^\infty(M)$-submodule of sections of $E$ which take values in $L$ when restricted to $C$.

\begin{definition}
 Let $R \colon E_1 \dashrightarrow E_2$ be a Courant algebroid relation supported on $C\subseteq M_1\times M_2$. Two sections $e_1 \in \mathsf{\Gamma}(E_1)$ and $e_2 \in \mathsf{\Gamma}(E_2)$ are $R$-\emph{related}, denoted $e_1 \sim_R e_2,$ if $(e_1, e_2) \in \mathsf{\Gamma}(E_1 \times \overline{E}_2 ; R).$ 
\end{definition}

Following Definition \ref{def:linearlagrel} we introduce

\begin{definition} \label{def:kerandcoker}
Let $R \colon E_1 \rel E_2$ be a Courant algebroid relation supported on $C$. The \emph{kernel} and \emph{cokernel} of $R$ are given by
    \begin{align}
        {\rm K}_R \coloneqq (E_1 \times 0_{E_2})\big\rvert_C \cap R \qquad \text{and} \qquad
        {\rm CK}_R \coloneqq (0_{E_1} \times E_2) \big\rvert_C \cap R \ ,
    \end{align}
    respectively, where $0_{E_i}$ is the zero section of $E_i$, for $i=1,2$. Note that both ${\rm K}_R$ and ${\rm CK}_R$ are singular subbundles of $R$. We also denote $\ker(R) = \rmp_1({\rm K}_R)$ and ${\rm coker}(R) = \rmp_2({\rm CK}_R)$. 
\end{definition}

Definition \ref{def:kerandcoker} may be formulated more generally for any isotropic subbundle $R$ of $E_1 \times \overline{E}_2$, where $E_1$ and $E_2$ are pseudo-Euclidean vector bundles.

\begin{example}[\textbf{Reduction of Courant Algebroids}] \label{eg:reductionCArel}
By~\cite[Proposition 3.7]{DeFraja:2023fhe}, the reduction of a Courant algebroid in Theorem~\ref{thm:reduction}
gives the Courant algebroid morphism 
\begin{equation}
\begin{tikzcd}
Q(K):E \ar[r, rightarrowtail] & \red E
\end{tikzcd}
\end{equation}
supported on $\gr(\varpi) \subset M \times \cQ,$
where $Q(K)$ is defined pointwise by
    \begin{align}  Q(K)_{(m,\varpi(m))}=\set{(e,\natural(e))\, | \,  e\in K^\perp_m} \ \subset \ E_m\times \red{\overline{E}}{}_{\varpi(m)} \ ,
\end{align}
for any $m \in M.$ Its kernel and cokernel are $${\rm K}_{Q(K)} = (K \times 0_{\red E}) \big\rvert_{\gr(\varpi)} \qquad  \text{and} \qquad {\rm CK}_{Q(K)} = (0_{E_1}\times 0_{E_2}) \big\rvert_{\gr(\varpi)}=0_{Q(K)} \ . $$
\end{example}

A particularly relevant class of Courant algebroid relations is given by

\begin{definition}[\textbf{Generalised Isometries}]\label{defn:generalisedisometry}
Suppose $R\colon E_1 \rel E_2$ is a Courant algebroid relation supported on a submanifold $C\subseteq M_1 \times M_2$. Let $\tau_1$ and $\tau_2$ be generalised metrics for $E_1$ and $E_2$ respectively. Then $R$ is a \emph{generalised isometry between $\tau_1$ and $\tau_2$} if $$(\tau_1\times\tau_2) (R) = R \ . $$ Equivalently, if $V_1^+$ and $ V_2^+$ are the maximally positive-definite subbundles associated to $\tau_1$ and $\tau_2$ respectively, then $R$ is a generalised isometry if
\begin{align}
    R_c = R_c^+ \oplus R^-_c
\end{align}
for every $c\in C$, where $$R_c^\pm = R_c \cap\big(V_1^\pm \times V_2^\pm\big)_c \ . $$
\end{definition}

As proved in \cite[Proposition 4.17]{DeFraja:2023fhe}, the kernel and cokernel of a generalised isometry $R$ are trivial, i.e. $ {\rm K}_R = {\rm CK}_R = {0_R}$.

Let us discuss how Courant algebroid relations compose.

\begin{definition}
Let $R\colon E_1 \rel E_2$ and $R' \colon  E_2 \rel E_3$ be Courant algebroid relations supported on $C\subseteq M_1\times M_2$ and $C'\subseteq M_2\times M_3$ respectively. The \emph{composition} $R' \circ R$ is the subset of $E_1 \times \overline{E}_3$ given by
\begin{align}\label{eqn:reldefinition}
\hspace{-5mm} R'\circ R = \set{(e_1,e_3) \in E_1 \times \overline{E}_3 \, | \, (e_1, e_2) \in R \ , \ (e_2,e_3) \in R' \ \text{ for some } e_2\in E_2} \ .
\end{align}
\end{definition}

The composition is a Courant algebroid relation $R' \circ R:E_1\rel E_3$ supported on $C'\circ C\subseteq M_1\times M_3$ if $R$ and $R'$ {compose cleanly}, see \cite{Vysoky2020hitchiker,DeFraja:2023fhe} for more details. 

\medskip

\subsection{Relating Dirac Structures}~\\[5pt]    
We introduce the notion of relation between Dirac structures descending from Courant algebroid relations.

\begin{definition}[\textbf{Dirac Relations}] \label{def:Diracrel}
Let $R \colon E_1 \dashrightarrow E_2$ be a Courant algebroid relation supported on a submanifold $C \subseteq M_1 \times M_2$. Let $L_1$ and $L_2$ be (almost) Dirac structures for $E_1$ and $E_2,$ respectively. The restriction of $R$ to $L_i$ for $i=1,2$ is given by
\begin{align}
R \big\rvert_{L_i} \coloneqq \set{(e_1, e_2) \in R \ \vert \ e_i \in L_i} \ .   
\end{align}
Then $R$ is a \emph{relation between the} (\emph{almost}) \emph{Dirac structures $L_1$ and $L_2$}, or a (\emph{almost}) \emph{Dirac relation}, if
\begin{align}
R \big\rvert_{L_1} = R\cap (L_1 \times L_2) \big\rvert_C + {\rm CK}_R  \qquad \text{and} \qquad  R \big\rvert_{L_2} = R\cap (L_1 \times L_2) \big\rvert_C + {\rm K}_R  
\end{align}
are regular vector subbundles supported on $C,$ and we write $L_1\sim_R L_2$. The inclusion $\supseteq$ always holds.
\end{definition}

\begin{remark}[\textbf{Dirac vs. Lie Algebroid Relations}] \label{rmk:DiracLieAlgebroid}
     Let $(A_1 , [\, \cdot \, ,\, \cdot \,]_1, \rho_1)$ and $(A_2 , [\, \cdot \, ,\, \cdot \,]_2, \rho_2)$ be Lie algebroids over $M_1$ and $M_2 $, respectively. Then the product vector  bundle $A_1 \times A_2$ over $M_1 \times M_2$ is endowed with a Lie algebroid structure given by the Lie bracket $([\, \cdot \, ,\, \cdot \,]_1,[\, \cdot \, ,\, \cdot \,]_2 )$ and anchor map $(\rho_1, \rho_2).$
     A \emph{Lie algebroid relation} $R$ between the Lie algebroids $A_1$ and $A_2$ is a Lie subalgebroid $R \subseteq A_1 \times A_2$ supported on a submanifold $C \subseteq M_1 \times M_2$. This generalises the notion of Lie algebroid morphism over a smooth map $\varphi:M_1\to M_2$ where $C=\gr(\varphi)$, see e.g.~\cite[Section~2.3]{Meinrenken}. 
    
     Suppose ${\rm CK_R} = {\rm K}_R = {0_R}$ in Definition~\ref{def:Diracrel}, so that $R \rvert_{L_1} = R\rvert_{L_2}$. If $L_1$ and $L_2$ are Dirac structures, then the Dirac relation $L_1 \sim_R L_2$ is a Lie algebroid relation between $L_1$ and $L_2$. Recall that the Lie algebroid structures for $L_1$ and $L_2$ are induced by the Courant algebroid structures on $E_1$ and $E_2$ respectively. We may motivate our definition of Dirac relation as a way to generalise morphisms of Dirac structures as descending from the ``arrows'' between the Courant algebroids they live in, rather than just considering their Lie algebroid structure.
\end{remark}

\begin{example}[\textbf{Foliated Manifolds}]
Let $(M_1, \cF_1)$ and $(M_2, \cF_2)$ be foliated manifolds and consider a foliation-preserving diffeomorphism $\varphi \in \mathsf{Diff}(M_1, M_2),$ i.e. $\varphi_* (T\cF_1) = T\cF_2$. Then $\varphi$ induces the  Courant algebroid isomorphism 
\begin{align}
    \boldsymbol{\varphi} \coloneqq \varphi_* \oplus (\varphi^{-1})^* \colon (\IT M_1 ,0) \longrightarrow (\IT M_2 , 0)
\end{align}
covering $\varphi.$ We show that the induced Courant algebroid relation $\gr(\boldsymbol{\varphi})$ supported on $\gr(\varphi)$ is a Dirac relation between the Dirac structures $L_1$ and $L_2$ induced by the foliations, where
\begin{align}
 L_i \coloneqq T\cF_i \oplus \mathrm{Ann}(T\cF_i) \ ,   
\end{align}
for $i=1,2.$ 

In this case
\begin{align}
 \gr(\boldsymbol{\varphi}) \big\rvert_{L_1} = \set{(X + \alpha , \varphi_*(X) + (\varphi^{-1})^*(\alpha)) \ \vert \ X + \alpha \in T\cF_1 \oplus \mathrm{Ann}(T\cF_1)} \ .   
\end{align}
Since $\varphi_* (T\cF_1) = T\cF_2$ it follows that 
$\gr(\boldsymbol{\varphi}) \rvert_{L_1} \subseteq L_1 \times L_2$ is a subbundle of $L_1 \times L_2$ supported on $\gr(\varphi),$ i.e. $\gr(\boldsymbol{\varphi})$ is a Dirac relation. Furthermore, $\gr(\boldsymbol{\varphi})\rvert_{L_1}$ is a Lie subalgebroid of $L_1 \times L_2$ supported on $\gr(\varphi)$ since
\begin{align}
\boldsymbol{\varphi}(\llbracket X+\alpha , Y+\beta\rrbracket_{L_1}) = \llbracket \boldsymbol{\varphi}(X+\alpha), \boldsymbol{\varphi}(Y+\beta)\rrbracket_{L_2}
\end{align}
for all $X+\alpha, Y +\beta \in \mathsf{\Gamma}(L_1),$ as is easily verified by applying $\boldsymbol{\varphi}$ to the bracket \eqref{eqn:Diracfoliation}.
\end{example}

\begin{example}[\textbf{Dirac Reduction}] \label{eg:Diracreductionrel}
Consider the setting of Example \ref{eg:reductionCArel}, i.e. let $Q(K)$ be the reduction relation of a Courant algebroid $E$ by an isotropic subbundle $K$ with reduced Courant algebroid $\red E$. We also assume to be in the setting of Subsection \ref{subsect:DiracReduction}, i.e. $E$ is endowed with a Dirac structure $L$ that reduces to the Dirac structure $$\red L = \natural(L \cap K^\perp + K) \ . $$ 

We shall show first that $Q(K)$ is a Dirac relation between $L$ and $\red L$ in the sense of Definition \ref{def:Diracrel}. For each $m\in M$ consider
\begin{align}
    {Q(K) \big\rvert_L}_{(m, \varpi(m))} = \set{(e , \natural(e)) \ \vert \ e \in K^\perp_m \cap L_m} \ .
\end{align}
This is a vector subspace of $L_m \times {\red L}_{\varpi(m)} $ because the image of $\natural$ restricted to $K^\perp \cap L + K$ is $\red L,$ as shown in Subsection \ref{subsect:DiracReduction}. For $L$ to be reducible, $ K^\perp\cap L$ has to have constant rank, which implies that $$Q(K) \big\rvert_L \ \subset \ (L \times {\red L}) \big\rvert_{\gr(\varpi)}$$ is a regular vector subbundle supported on $\gr(\varpi).$ Note that $ {\rm CK}_{Q(K)} = 0_{Q(K)}$ by Example~\ref{eg:reductionCArel}.

Similarly
\begin{align}
    {Q(K) \big\rvert_{\red L}}_{(m, \varpi(m))}= \set{(e, \natural(e)) \in Q(K)_{(m, \varpi(m))} \ \vert \ \natural(e) \in {\red L}_{\varpi(m)} } \ ,
\end{align}
where $e \in K^\perp_m \cap L_m + K_m$ by definition of $\natural$ and $\red L$. Again by the regularity of $K^\perp \cap L$ it follows that $$Q(K)\big|_{\red L} \ \subset \ (L \times \red L)\big\rvert_{\gr(\varpi)}+ {\rm K}_{Q(K)}$$ is a regular vector subbundle, thus showing that $L \sim_{Q(K)} \red L$.

Lastly, $Q(K) \rvert_L$ is a ``bracket morphism'' between $L$ and $\red L,$ as shown in Equation \eqref{eqn:reducedbracket}. The anchors are related because $Q(K)$ is a Courant algebroid relation. Thus it represents a Lie algebroid relation between $L$ and $\red L.$ Hence $Q(K)$ is a Dirac relation between $L$ and $\red L.$ 
\end{example}

\medskip

\subsubsection{Morphisms of Manin Pairs}~\\[5pt]
Let us make contact with the existing literature concerning Courant algebroid relations between Dirac structures. We focus primarily on morphisms of Manin pairs for Courant algebroids, formulated in \cite{bursztyn2008courant}.
A \emph{Manin pair}  $(E, L)$ is given by a Courant algebroid $E$ over $M$ endowed with a Dirac structure $L$ supported on $M.$

\begin{definition}
Let $(E_1, L_1)$ and $(E_2, L_2)$ be Manin pairs over $M_1$ and $M_2,$ respectively. Let $\varphi \colon M_1 \to M_2$ be a smooth map.  A \emph{morphism of Manin pairs}    between $(E_1, L_1)$ and $(E_2, L_2)$ is a Dirac structure $D \subset E_1 \times E_2$ supported on $\gr(\varphi) \subset M_1 \times M_2$ such that $D \cap \rmp_1^*(L_1) = \set{0}$ and $D \cap \rmp_1^*(L_1 \oplus \varphi^* E_2)$ projects onto $\varphi^* L_2$ under the natural projection $E_1 \oplus \varphi^* E_2 \to \varphi^* E_2.$
\end{definition}

The latter projection becomes an isomorphism onto $\varphi^*L_2$ when restricted to $D \cap \rmp_1^*(L_1 \oplus \varphi^* E_2).$

\begin{proposition}
 A morphism $D$ of Manin pairs $(E_1, L_1)$ and $(E_2, L_2)$ is a Dirac relation in the sense of Definition \ref{def:Diracrel}.    
\end{proposition}

\begin{proof}
    If $D\cap \rmp_1^*(L_1 \oplus \varphi^* E_2)$ projects onto $\varphi^* L_2$ then
    \begin{align}
        D\cap(L_1 \times E_2) = D\big|_{L_1} = D\cap (L_1 \times L_2) \ .
    \end{align}
    It follows that ${\rm CK}_D = 0_D$. Since the projection $D\cap \rmp_1^*(L_1 \oplus \varphi^* E_2)\to \varphi^*L_2$ is onto, it follows that for every $e_2 \in L_2$ there is some $e_1\in E_1$ such that $(e_1, e_2) \in D$, and hence $$D\big|_{L_2} = D\cap(L_1 \times L_2) + {\rm K}_D \ , $$ as required.
\end{proof}

A particularly relevant example of morphism of Manin pairs arises from 

\begin{definition}
    Let $M$ and $N$ be smooth manifolds with a smooth map $f \colon M \to N$. For any $m \in M,$ elements $X+ \alpha \in \IT_m M$ and $Y+ \beta \in \IT_{f(m)} N$  are $f$-\emph{related} if $Y = (f_*)_m(X)$ and $\alpha = f^*_m(\beta).$ We write $(X+ \alpha) \sim_f (Y+ \beta).$
\end{definition}

\begin{lemma}\label{lemma:stienon}
    Let $(\IT M, f^*H)$ and $(\IT N, H)$ be twisted standard Courant algebroids, where $H\in\sfOmega_{\tt cl}^3(N)$ and $f \colon M \to N$ is a smooth map. If $X_i + \alpha_i \in \sfGamma(\IT M)$ is $f$-related to $Y_i+ \beta_i \in \sfGamma(\IT N),$ for $i=1,2,$ then
    \begin{align}
        \ip{X_1+ \alpha_1, X_2+ \alpha_2}_{\IT M} = \ip{Y_1+ \beta_1, Y_2+ \beta_2}_{\IT N} \ ,
    \end{align}
    and
    \begin{align}
        \llbracket X_1+ \alpha_1, X_2+ \alpha_2\rrbracket_{f^*H} = \llbracket Y_1+ \beta_1, Y_2+ \beta_2\rrbracket_{H} \ .
    \end{align}
\end{lemma}

\begin{proof}
    See \cite{Stienon2008}.
\end{proof}

An immediate consequence of Lemma \ref{lemma:stienon} is thus

\begin{proposition}
    The subbundle 
    \begin{align}
        R_f \coloneqq \set{\left( X+ \alpha, Y+ \beta \right) \in \IT M \times \IT N \ \vert \ (X+ \alpha) \sim_f (Y+ \beta)} 
    \end{align}
    is a a Courant algebroid relation $R_f:(\IT M, f^*H)\rel (\IT N, H)$ supported on $\gr(f)$. 
\end{proposition}

As in \cite[Example 2.11]{bursztyn2008courant}, by taking any Dirac structures $L \subset \IT M$ and $L'\subset\IT N$, it follows that that $R_f$ induces a Dirac relation if and only if 
\begin{align}
    \set{f_*(X)+ \alpha  \ | \ X\in TM \ ,\ \alpha \in T^*N \ , \ X+ f^*(\alpha) \in L} &= L' \ , \\[4pt]
     \rmp_1^*\ker(\de f) \cap \rmp_1^*(TM) \cap L &= \set{0}  \ ,
\end{align}
that is, $f$ is a \emph{strong Dirac map} (see \cite{bursztyn2008courant} and references therein for details). 

\medskip

\subsection{The T-Duality Relation}~\\[5pt] \label{subsect:Tdualityrel}
 Let $E$ be an exact Courant algebroid over $M$ endowed with isotropic subbundles $K_1$ and $K_2$ such that
$\rk_\IR (K_1) = \rk_\IR(K_2)$.    
Assume that $K_1^\perp$ and $K_2^\perp$ are pointwise spanned by basic sections, and denote by $\cF_i$ the foliation of $M$ induced by the integrable distribution $\rho(K_i).$ Suppose that the leaf spaces $\cQ_i=M / \cF_i$ have a smooth structure, hence there are unique surjective submersions $\varpi_i \colon M \to \cQ_i$ for $i=1,2,$ which we summarise in the bisubmersion diagram
\begin{equation}
    \begin{tikzcd} [row sep = 1cm]
  &\arrow[swap]{dl}{\varpi_1} M \arrow{dr}{\varpi_2}&   \\
\cQ_1 & & \cQ_2 
 \end{tikzcd}
\end{equation}

We consider the setting in which the diagram
\begin{equation}
\begin{tikzcd}[row sep = 1cm] \label{cd:T-dualitycd}
  &\arrow[tail,swap]{dl}{Q(K_1)} E \arrow[tail]{dr}{Q(K_2)}&   \\
\red E{}_1 \arrow[dashed]{rr}{R}& & \red E{}_2 
 \end{tikzcd}
 \end{equation}
defines a  Courant algebroid relation 
\begin{equation}
\begin{tikzcd}
 R =  Q(K_2) \circ Q(K_1)^\top \ \colon \ \red E{}_1 \ar[r,dashed] & \red E{}_2 
 \end{tikzcd}
\end{equation} 
supported on the submanifold 
\begin{align}
    {\red C}=\set{(\varpi_1(m),\varpi_2(m)) \, | \, m \in M} \ \subset \ \cQ_1 \times \cQ_2 \ .
\end{align}

It was shown in \cite[Theorem 5.8]{DeFraja:2023fhe} that if $\red C$ is a smooth manifold and $T\cF_1 \cap T\cF_2$ has constant rank, then $R$ is a Courant algebroid relation if and only if $K_1 \cap K_2$ has constant rank. Then $R$ is a Dirac structure  in $\red E{}_1 \times \overline{\red E}{}_2$ supported on $\red C\,$, which is given explicitly by
\begin{align}\label{eqn:T-dualRelation}
    R = \set{(\natural_{1}(e), \natural_{2}(e))\,|\, e \in K_1^\perp \cap K_2^\perp} \ ,
\end{align}
where $\natural_i \colon  K_i^\perp \to \bigr( K_i^\perp /K_i \bigl) / \cF_i = \red E{}_i$ denotes the quotient map  for $i=1,2$.
It was shown in \cite[Lemma~5.7]{DeFraja:2023fhe}  that $\red C$ is a smooth manifold if $\cQ_1$ and $\cQ_2$ are fibred over the same manifold $\cB,$ in which case $\red C$ is the fibred product $\red C = \cQ_1 \times_\cB \cQ_2.$

In this instance $R$ is said to be a \emph{T-duality relation} between the reduced Courant algebroids $\red E{}_1$ and $\red E{}_2$~\cite{DeFraja:2023fhe}. It describes a topological T-duality between the manifolds $\cQ_1$ and $\cQ_2$ (generically endowed with NS--NS fluxes). Geometric T-duality is defined in \cite{DeFraja:2023fhe} by introducing generalised metrics $ V_1^+$ and $ V_2^+$, or equivalently $\tau_1$ and $\tau_2$, into this setting.

\begin{definition} \label{def:geomTdual}
Two Courant algebroids endowed with generalised metrics $(\red E{}_1,  V_1^+)$ and $(\red E{}_2,  V_2^+)$ fitting into Diagram \eqref{cd:T-dualitycd} are \emph{geometrically T-dual} if $R$ is a generalised isometry: $(\tau_1 \times \tau_2)(R) = R$.
\end{definition}

The problem of transporting a generalised metric $ V_1^+ \subset \red E{}_1$ to $\red E{}_2$ such that $R$ becomes a generalised isometry is addressed in \cite{DeFraja:2023fhe}. 
For this, introduce the subbundle $$D_1 \coloneqq T {\red\cF}{}_1 = \varpi_{1*}(T \cF_2)$$ of $T\cQ_1$, which is induced by the subbundle $K_2 \subset E$ and called the \emph{distribution of T-duality directions}.
It plays the role of the bundle of isometries over $\cQ_1$.

\begin{definition} \label{def:D1invariance}
    Let $\red \sigma_1 \colon T\cQ_1 \to \red E{}_1 $ be the splitting of the exact Courant algebroid $\red E{}_1$ over $\cQ_1$ induced by an isotropic  splitting $\sigma_1 \colon TM \to E$ of $E$ which is adapted to $K_1\subset E$.   
    A generalised metric $V_1^+$ on $\red E{}_1$ is \emph{invariant with respect to $D_1$}, or simply \emph{$D_1$-invariant}, if there are generators $\set{\red X{}_1,\dots, \red X{}_{r_1}}\subset \sfgamma(T\cQ_1)$ of the $C^\infty(\cQ_1)$-submodule $\sfgamma(D_1)$ such that
    \begin{align}\label{eqn:genmetricinvariance}
        \ad^{\red \sigma_1}_{\red X{}_k} (\red w{}_1 ) \ \in \ \mathsf{\Gamma}(V_1^+) \ ,
    \end{align}
    for $k=1,\dots,r_1$ and $\red w{}_1\in \mathsf{\Gamma}(V_1^+)$. The Lie subalgebra $$\frk{k}_{\tau_1} \coloneqq \text{Span}_{\IR}\set{\red X{}_1,\dots, \red X{}_{r_1}}$$ of $\sfgamma(T\cQ_1)$ is the \emph{isometry algebra of $D_1$} or the \emph{Lie algebra of Killing vector fields for $V_1^+$}.
\end{definition}

\begin{remark}
    For a twisted standard Courant algebroid $(\IT \cQ_1, \red H{}_1)$ with any $ \red X\in \sfGamma(T\cQ_1)$, $\red e = \red Y+\red \eta \in \sfgamma(\IT \cQ_1)$ and $ \red \omega \in \sfOmega^\bullet(\cQ_1)$, one computes
    \begin{align}
        \pounds_{\red X} \big( (\red Y+\red \eta) \cdot \red \omega\big) &= \pounds_{\red X} (\iota_{\red Y} \omega) + \pounds_{\red X}(\red \eta \wedge \red \omega\big)\\[4pt]
        &= \iota_{\red Y}\, \pounds_{\red X}\red  \omega + \iota_{[\red X,\red Y]} \red \omega + (\pounds_{\red X}\red  \eta) \wedge \red \omega + \red \eta \wedge \pounds_{\red X}\red \omega\\[4pt]
        &=(\red Y+\red \eta)\cdot \pounds_{\red X}\red  \omega + ([\red X,\red Y] + \pounds_{\red X}\red  \eta) \cdot \red \omega \ ,
    \end{align}
which implies
\begin{align}\label{eqn:spinorLeibnizrule}
        \pounds_{\red X} \big( (\red Y+\red \eta) \cdot \red \omega\big) = \ad_{\red X}^{\red \sigma_1}(\red Y+\red \eta) \cdot \red \omega + (\red Y+\red \eta)\cdot \pounds_{\red X}\red  \omega \ .
    \end{align}
\end{remark}

\begin{definition} \label{def:compatiblesplit}
    Let $V_1^+$ be a $D_1$-invariant generalised metric on $\red E{}_1$ with isometry algebra $\frk{k}_{\tau_1}$. Adapted splittings $\sigma_i:TM\to E$ with respect to $K_i$ for $i=1,2$ are \emph{compatible} if they satisfy
    \begin{align} \label{eqn:compatiblesplits}
        \llbracket \sigma_2  (X) , e \rrbracket = \ad^{\sigma_1}_X\, e \ ,
    \end{align}
    for every $X \in  \frk{k}_2\subset\sfGamma(TM)$ and $e\in \mathsf{\Gamma}(E)$, where
    \begin{align}
       \frk{k}_2 \coloneqq \varpi_{1*}^{-1}(\frk{k}_{\tau_1}) \cap \mathsf{\Gamma}\big(\rho(K_2)\big) \ .
    \end{align}
\end{definition}

We can now state the main result proven in \cite{DeFraja:2023fhe}. It provides global conditions under which a generalised metric $V_1^+$ on $\red E{}_1$ can be transported to a unique generalised metric $V_2^+$ on $\red E{}_2$ such that the T-duality relation $R$ is a generalised isometry.

\begin{theorem}\label{thm:maingeneral}
    Let $E$ be an exact Courant algebroid over $M$ with isotropic subbundles $K_1$ and $K_2$ of the same rank such that $K_1^\perp$ and $K_2^\perp$ are pointwise spanned by basic sections, giving T-duality related exact Courant algebroids $\red E{}_1$ and $\red E{}_2$ with T-duality relation $R$. Assume further that the involutivity condition
    \begin{align} \label{eqn:bisubmersioncd}
\big[\mathsf{\Gamma}\bigl(\ker(\varpi_{1*})\bigr) , \mathsf{\Gamma}\big(\ker(\varpi_{2*})\big)\big] \ \subseteq  \ \mathsf{\Gamma}\bigl(\ker(\varpi_{1*})\bigr) + \mathsf{\Gamma}\bigl(\ker(\varpi_{2*})\bigr)
\end{align}
    for bisubmersions holds. Let  $\sigma_1$ and $\sigma_2$ be compatible adapted splittings.  Let $\red E{}_1$ be endowed with a $D_1$-invariant generalised metric $V_1^+$. 
    Then the following statements are equivalent:
    \begin{enumerate}[label = (\roman{enumi})]
        \item\label{item:main1} $K_2^\perp \cap K_1 \subseteq K_2 \ . $ \\[-3mm]
        \item\label{item:main2} $K_2 \cap K_1^\perp \subseteq K_1 \ . $ \\[-3mm]
        \item\label{item:main3} There exists a unique generalised metric $V_2^+$ on $\red E{}_2$ such that $R$ is a generalised isometry between $V_1^+$ and $V_2^+$, i.e. $(\red E{}_1, V_1^+)$ and $(\red E{}_2, V_2^+)$ are geometrically T-dual.
    \end{enumerate}
\end{theorem}

The condition \eqref{eqn:bisubmersioncd} implies that there are unique foliations $\red \cF{}_1$ of $\cQ_1$, whose tangent distribution is $D_1$, and $\red \cF{}_2$ of $\cQ_2$ which are Hausdorff-Morita equivalent, see \cite{Garmendia2019}.

\begin{remark} \label{rmk:tdualsplitiso}
 Consider the setting of Theorem~\ref{thm:maingeneral} with the isomorphism 
 \begin{align}
 \Psi_1 = \sigma_1 \oplus \tfrac{1}{2}\,\rho^* \colon TM \oplus T^*M \longrightarrow E \ .
 \end{align}
 Recall that the difference between two isotropic splittings $\sigma_1$ and $\sigma_2$ is given by $$(\sigma_1 - \sigma_2)(X)= \rho^*(\iota_X B) \ , $$ for some $B \in \sfOmega^2(M)$ and for all $X \in \sfgamma(TM)$. Since $\sigma_i$ is adapted to $K_i$, it then follows that
 \begin{align}
 \Psi_1^{-1}(K_1) = T\cF_1 \qquad \text{and} \qquad \Psi_1^{-1}(K_2)= \e^{-B}\,(T\cF_2) \ . 
 \end{align}
Thus if $V_1^+$ is a $D_1$-invariant generalised metric on $\red E{}_1$, then a T-dual generalised metric $V_2^+$ on $\red E{}_2$ exists if $\pounds_X B = 0$, for all $X \in \frk{k}_2$, and $T\cF_1 \cap T\cF_2 \subseteq \ker(B)$~\cite[Theorem 5.44]{DeFraja:2023fhe}.
\end{remark}

The classic example of the construction of Theorem \ref{thm:maingeneral} is given by the Cavalcanti-Gualtieri formulation of T-duality~\cite{cavalcanti2011generalized}, see \cite[Section 6.1]{DeFraja:2023fhe}. 

Since we have Dirac relations at our disposal, we can now discuss T-duality for Dirac structures. For this, as in Definition \ref{def:D1invariance}, we need 

\begin{definition} \label{def:invDirac}
    Let $\red \sigma_1 \colon T\cQ_1 \to \red E{}_1 $ be the splitting of the exact Courant algebroid $\red E{}_1$ over $\cQ_1$ induced by an isotropic  splitting $\sigma_1 \colon TM \to E$ of $E$ which is adapted to $K_1\subset E$.   
    A Dirac structure $\red L{}_1$ on $\red E{}_1$ is \emph{invariant with respect to $D_1$}, or simply \emph{$D_1$-invariant}, if there are generators $\set{\red X{}_1,\dots, \red X{}_{r_1}}\subset \sfgamma(T\cQ_1)$ of the $C^\infty(\cQ_1)$-submodule $\sfgamma(D_1)$ such that
    \begin{align}
        \ad^{\red \sigma_1}_{\red X{}_k} (\red e{}_1 ) \ \in \ \mathsf{\Gamma}(\red L{}_1) \ ,
    \end{align}
    for $k=1,\dots,r_1$ and $\red e{}_1\in \mathsf{\Gamma}(\red L{}_1)$.  
\end{definition}

Then we arrive at a result for Dirac structures analogous to Theorem \ref{thm:maingeneral} given by

\begin{proposition} \label{prop:TdualDirac}
    Under the assumptions of Theorem \ref{thm:maingeneral}, let $\red L{}_1$ be a $D_1$-invariant Dirac structure on $\red E{}_1$. Then there exists a unique Dirac structure $\red L{}_2$ on $\red E{}_2$ such that the T-duality relation $R$ is a Dirac relation.
\end{proposition}

\begin{proof}
    The construction of the T-dual Dirac structure $\red L{}_2$ on $\red E{}_2$ follows the same steps as in the construction of the T-dual generalised metric $V_2^+$ of Theorem \ref{thm:maingeneral}, see the proof of \cite[Theorem~5.27]{DeFraja:2023fhe} for details. 
\end{proof}

In this setting the T-duality relation $R$ restricts to a Lie algebroid relation as in Remark \ref{rmk:DiracLieAlgebroid}.

\medskip

\section{Spinors and Courant Algebroid Relations} \label{sect:section5}

We shall now show how to extend the notion of Courant algebroid relations to the spinor modules between related Courant algebroids, using the ideas of Section~\ref{ssec:linearspinorrel}. When applied to the instance of geometric T-duality relations, we demonstrate how our approach offers a novel alternative perspective to the Fourier-Mukai transform as well as on the T-duality invariance of Type~II supergravity.

\medskip

\subsection{Clifford Algebra Relations and Dirac Generating Operators}~\\[5pt]
\label{sub:CliffordDGO}
Let $E_1 \to M_1$ and $E_2 \to M_2$ be split signature pseudo-Euclidean vector bundles. For $i=1,2$, let $\sfS_{E_i}$ be the spinor module over the Clifford algebra bundle $\Cl(E_i)$.

Let $R$ be an isotropic subbundle of $E_1 \times \overline{E}_2$ supported on $C \subseteq M_1 \times M_2$ such that $\rho(R^\perp) \subseteq TC$. As in Definition \ref{def:RCliffrel2}, we can define the space $R^\times$ as
\begin{align}\label{eq:Rtimes}
    R^\times \coloneqq (R/ {\rm K}_R)/ \underline{\rm CK}_R \ , 
\end{align}
where $\underline{{\rm CK}}_R \coloneqq (0_{E_1} \times E_2) \cap (R / {\rm K}_R)$
while the kernel and cokernel of $R$ are defined as in Definition~\ref{def:kerandcoker}.

\begin{definition} \label{def:Rcliffrelation}
    A vector subbundle ${\sf S}_R \subset \bigl( \sfS_{E_1}\times\sfS_{E_2}\bigr) \big \rvert_C$ is an \emph{$R$-Clifford relation of parity $i \in \IZ_2$} if it is a spinor module of the typical fibre of $\Cl^i (R)$, where the Clifford action is defined as in \eqref{eqn:productcliffordaction} with the fibres of the bundle $\big(\Cl(E_1)\times \Cl(E_2)\big) \big\rvert_C$ acting on the fibres of $ \big(\sfS_{E_1}\times\sfS_{E_2} \big) \big\rvert_C $.

    If there exists an action of $\Cl(R^\times)$ on $\bigl( \sfS_{E_1}\times\sfS_{E_2}\bigr) \big \rvert_C$ compatible with $\Cl^i(R)$ and a subbundle $\sfS_R \subset \bigl( \sfS_{E_1}\times\sfS_{E_2}\bigr) \big \rvert_C$ such that $\sfS_R$ is a spinor module for $\Cl^i(R^\times)$, then $\sfS_R$ is a \emph{weak $R$-Clifford relation of parity $i \in \IZ_2$}.
\end{definition}

Proposition \ref{prop:Uisomorphism} and Corollary \ref{cor:URinjective} have their smoothly varying counterparts in this case, given respectively by

\begin{proposition}\label{prop:Uisomorphismsmooth}
    Let $\rmp_i$ be the projection from $E_1 \times E_2$ to the $i$-th factor for $i=1,2$, and denote by the same symbols the projections from $\sfS_{E_1}\times\sfS_{E_2}$. Let $\sfS_R$ be an $R$-Clifford relation of parity $j$, with $\rmp_i(R) \neq \set{0}$ and $\rmp_i(\sfS_R) \neq \set{0}$ for $i=1,2$. Then for each $c \in C$
    \begin{align}
        (\sfS_R)_c \cap \big(\set{0} \times \sfS_{E_2} \big)_c = (\sfS_R)_c \cap \big(\sfS_{E_1} \times \set{0} \big)_c = \set{(0,0)_c} \ .
\end{align}
\end{proposition}

\begin{corollary}\label{cor:URinjectivesmooth}
    Let $\sfS_R$ be an $R$-Clifford relation with $\rmp_i(R) \neq \set{0}$ and $\rmp_i(\sfS_R) \neq \set{0}$ for $i=1,2$. If $(\nu, \mu)_c, (\nu, \mu')_c\in (\sfS_R)_c$ for every $c\in C$, then $\mu = \mu'$.
\end{corollary}

\begin{example}\label{ex:pullbackspinorrel}
    Let $(M,\cF)$ be a foliated manifold with smooth leaf space $\cQ$. Let 
\begin{equation}
\begin{tikzcd}
Q(\cF): \IT M \ar[r, rightarrowtail] & \IT \cQ
\end{tikzcd}
\end{equation}
be the reduction relation of Example~\ref{eg:reductionCArel} with $K=T\cF\oplus\set{0}$ (cf. Example~\ref{eg:symplecticred}), covering the surjective submersion $\varpi \colon M \to \cQ$. Consider the subbundle
    \begin{align}
        \sfS_{Q(\cF)} = \set{(\varpi^* \red\alpha, \red\alpha) \ \vert \ \red\alpha \in \midwedge^\bullet T^*\cQ}
    \end{align}
    supported on $\gr(\varpi)$.
    Recalling that $\mathrm{CK}_{Q(\cF)}=0_{Q(\cF)}$, for each $ \red X + \red\eta\in \sfgamma(\IT \cQ)$ there is a section $ X + \eta \in \sfgamma(\IT M)$ such that $( X+\eta,\red X + \red\eta) \in \sfGamma(\IT M \times \overline {\IT\cQ}\,; Q(\cF))$, where $\eta = \varpi^*\red\eta$ and $ X$ is any lift of $\red X$, i.e. $\varpi_*  X = \red X$. Then on $\gr(\varpi)$
    \begin{align}
        ( X+\eta, \red X + \red\eta) \cdot(\varpi^* \red\alpha , \red\alpha) &= \big(\iota_{ X} (\varpi^*  \red\alpha) + \eta \wedge \varpi^* \red\alpha,\iota_{\red X}\red\alpha +\red\eta \wedge \red\alpha\big) \\[4pt]
        &=\big(\varpi^*(\iota_{\red X} \red\alpha +\red\eta\wedge \red\alpha),\iota_{\red X} \red\alpha +\red\eta\wedge \red\alpha\big) \ \in \ \sfS_{Q(\cF)} \ .
    \end{align}
    
    Extending to the whole Clifford bundle $\Cl(Q(\cF))$, it follows that $\sfS_{Q(\cF)}$ is invariant under the action of $\Cl(Q(\cF))$, providing it a spin representation. To show that it is an irreducible representation, suppose that $\sfS' \subseteq \sfS_{Q(\cF)}$ is invariant. Then $\sfS'$ has the form
    \begin{align}
        \sfS' = \set{(\varpi^* \red\alpha, \red\alpha) \ \vert \ \red\alpha \in S}
    \end{align}
where $S \subseteq \midwedge^\bullet T^*\cQ$ is a vector subbundle. By invariance of $\sfS'$ and projection to the second factor, it follows that $S$ is invariant under the action of $\Cl(\IT \cQ)$. But $\sfOmega^\bullet(\cQ)$ is a spinor module for $\Cl(\IT \cQ)$, and hence $S = \set{0}$ or $S=\midwedge^\bullet T^*\cQ$. Thus $\sfS_{Q(\cF)}$ is a $Q(\cF)$-Clifford relation of parity $+$.
\end{example}

\begin{example}\label{ex:Qfibreintegration}
    Let $(M,\cF)$ be a foliated manifold with compact oriented fibres and smooth leaf space $\cQ$. Since the leaves of $\cF$ are compact and oriented, there is a pushforward map
    \begin{align}
        \varpi_! \colon \sfOmega^\bullet(M) \longrightarrow \sfOmega^\bullet(\cQ) \ , \qquad
        \xi\longmapsto \int_{\cF}\, \xi
    \end{align}
    given by the fibre integration of forms over the leaves of  $\cF$, inducing a vector bundle morphism $\midwedge^\bullet T^*M \to \midwedge^\bullet T^*\cQ$ denoted by the same symbol. Consider the graph $\gr(\varpi_!)$ as a subspace of $(\midwedge^\bullet T^*M \times \midwedge^\bullet T^* \cQ)|_{\gr(\varpi)}$. This is not a spinor module for $\Cl(Q(\cF))$. Indeed take any $\red\alpha \in \midwedge^\bullet T^* \cQ$ and note that $\varpi_! (\varpi^*\red\alpha) = 0$, thus $(\varpi^*\red\alpha,0) \in \gr(\varpi_!)$. By Proposition \ref{prop:Uisomorphismsmooth}, it follows that $\gr(\varpi_!)$ is not a spinor module for $\Cl(Q(\cF))$.

    Consider instead a fibre volume form $\mu \in \det(T^*\cF)$ and set $$\sfS_\mu \coloneqq \gr\big(\varpi_!\big|_{A_\mu}\big) \ , $$ where $A_\mu $ is the subspace of forms $ \varpi^* \red\alpha \wedge\mu $ for $ \red\alpha\in \midwedge^\bullet T^*\cQ$, supported on $\gr(\varpi)$. This is not an invariant subspace for the action of $\Cl(Q(\cF))$, since if $X \in T\cF = {\rm ker}(Q(\cF))$ then $$(X,0) \cdot \big(\varpi^*\red\alpha \wedge \mu, \varpi_!(\varpi^* \red\alpha \wedge \mu)\big) = (\varpi^* \red\alpha \wedge \iota_X \mu, 0) \ \notin \ \sfS_\mu \ . $$ 
    
    However, we may define an action of $\Cl(Q(\cF)^\times)$ on \smash{$\bigl( \midwedge^\bullet T^*M \times \midwedge^\bullet T^*\cQ\bigr) \big \rvert_{\gr(\varpi)}$}. Note that $Q(\cF) ^\times $ is the set of pairs $ (X+\xi, \red X + \red \xi)$ where $X \in TM / T\cF$ and $\xi \in {\rm Ann}(T\cF)$, while $\red X$ and $\red \xi$ are canonically associated to $X$ and $\xi $ respectively. By considering $D = {\rm Ann}(\mu) \simeq TM / T\cF$, we then obtain $X_D \in TM$ associated to $X$, and hence we define
    \begin{align}\label{eqn:Qweakaction}
        (X + \xi, \red X+ \red \xi) \cdot\big(\varpi^* \red\alpha \wedge \mu, \varpi_!(\varpi^* \red\alpha \wedge \mu)\big) \coloneqq (X_D + \xi, \red X+ \red \xi)\cdot\big( \varpi^* \red\alpha \wedge \mu,  \varpi_!(\varpi^* \red\alpha \wedge \mu)\big)
    \end{align}
    for the usual action of $\Cl^j(Q(\cF))$ on \smash{$\bigl( \midwedge^\bullet T^*M \times \midwedge^\bullet T^*\cQ\bigr) \big \rvert_{\gr(\varpi)}$}, where $j = {\rm dim}_\IR(\cF) \mod 2$. Since $\varpi_* X_D = \red X$ and $\xi = \varpi^*\red\xi$, we then obtain
    \begin{align}
       & (X + \xi, \red X+ \red \xi) \cdot\big(\varpi^* \red\alpha \wedge \mu, \varpi_!(\varpi^* \red\alpha \wedge \mu)\big) \\[4pt]
       & \hspace{2cm} = \big(\varpi^*(\iota_{\red X}\red\alpha+\red\xi\wedge \red\alpha) \wedge \mu, (-1)^j\,(\red X + \red \xi)\cdot \varpi_!( \varpi^* \red\alpha \wedge \mu)\big) \\[4pt]
        & \hspace{4cm} =\big(\varpi^*(\iota_{\red X}\red\alpha+\red\xi\wedge \red\alpha) \wedge \mu, \varpi_!(\varpi^*(\iota_{\red X}\red\alpha + \red\xi \wedge \red\alpha) \wedge \mu)\big) \ \in \  \sfS_\mu \ ,
    \end{align}
    where we use Appendix \ref{appendix} to pass the action through the fibre integration in the second factor. Thus $\sfS_\mu$ is a weak $Q(\cF)$-Clifford relation of parity ${\rm dim}_\IR(\cF) \mod 2$. 
    
    Note that the subbundle $\sfS_{Q(\cF)}$ from Example~\ref{ex:pullbackspinorrel} is also a spinor module for $\Cl(Q(\cF)^\times)$, and so is isomorphic to $\sfS_\mu$, with the isomorphism given explicitly by
    \begin{align}
        (\varpi^*\red\alpha , \red\alpha) \big\rvert_{\gr(\varpi)} \longmapsto \Big( \varpi^* \red\alpha \wedge \mu\,,\, \big(\text{\footnotesize{$\int_\cF$}}\, \mu\big)\ \red\alpha \Big) \bigg\rvert_{\gr(\varpi)}\ .
    \end{align}
\end{example}

The smooth version of the canonical Clifford relation, allowing us to relate spinor lines as in Definition \ref{def:canonicalRclifford}, is given by

\begin{definition}
    If $\sfS_R$ is a weak $R$-Clifford relation of parity $i$ for $R \subseteq E_1 \times \overline{E}_2$, the \emph{canonical $R$-Clifford relation of parity $i$} is the vector bundle
    \begin{align}
    \cS_R = \sfS_R \otimes |\det \sfS_R^* |^{\frac{1}{2^n}}
\end{align}
over $C \subseteq M_1 \times M_2$, where $n = \frac12\,\rk_\IR(E_1)$. As a spinor module, the typical fibre of $\Cl^i(R)$ acts trivially on the second component of $\cS_R$.

If $\Omega_1 \subset \sfS_{E_1}$ and $\Omega_2 \subset \sfS_{E_2}$ are pure spinor line bundles, then $\Omega_1$ and $\Omega_2$ are \emph{$R$-related}, denoted as $\Omega_1 \sim_R \Omega_2$, if
\begin{align}
    (\Omega_1 \times \Omega_2)\big\rvert_C\otimes |\det \sfS_R^* |^{\frac{1}{2^n}} \ \subset \ \cS_R
\end{align}
is a vector subbundle.
\end{definition}

Again the smoothly varying version of Proposition \ref{prop:LreliffUrel} follows as

\begin{proposition}\label{prop:LreliffUrelsmooth}
    Let $R \colon E_1 \rel E_2$ be a Courant algebroid relation supported on $C$, with spinor modules $\sfS_i$ for $\Cl(E_i)$ for $i=1,2$ and an $R$-Clifford relation $\sfS_R \subset (\sfS_1 \times \sfS_2) \rvert_C$. Let $L_1 \subset E_1$ and $L_2 \subset E_2$ be almost Dirac structures, and let $\Omega_1 \subset \sfS_1$ and $\Omega_2 \subset \sfS_2$ be the corresponding pure spinor line bundles. If $\Omega_1 \sim_R \Omega_2$, then $L_1 \sim_R L_2$. 
\end{proposition}

\begin{remark}\label{rem:smoothULequiv}
    The smoothly varying cases of Remarks \ref{rem:fullrankcase} and \ref{rem:weakcase} also follow. The condition that $\rmp_1(R) \supseteq L_1$ or $\rmp_2(R) \supseteq L_2$ is satisfied for instance when $R$ is a generalised isometry, whose kernel and cokernel are necessarily trivial. It is also satisfied when $R$ is the reduction relation $Q(K)$, since $\rmp_2(Q(K)) = \red E$.
\end{remark}

\begin{remark}
    Let us discuss Example \ref{ex:Qfibreintegration} in relation to the work of Drummond~\cite{Drummond:2011eq}, where the reduction of pure spinors is discussed. Drummond considers a Courant algebroid reduction by an isotropic subbundle $K\subset\IT M$. A pure spinor $\alpha$ gives a representative of an almost Dirac structure $L$. The pushforward of $\alpha$ is non-zero if and only if $L\cap K=\set{0}$. However, we know from Proposition~\ref{prop:reducedKsub-bundle} that  to any  almost Dirac structure $L$, such that $L\cap K^\perp$ has constant rank, there is an associated reduced almost Dirac structure $\red L$. 
    
    In the case $L \cap K \neq \set{0}$, Drummond resolves this issue by considering a perturbation $L_D$ of $L$ by a subbundle $D\subset \IT M$. The pure spinor associated to $L_D$ is then given by $\mu_D \cdot \alpha$, for some nowhere vanishing $\mu_D \in \det(D)$. The pure spinor $\varpi_!(\mu_D \cdot \alpha)$ then generates the almost Dirac structure $\red L$. 
    
    For instance, in the case that $\alpha = \varpi^* \red \alpha$ for some $\red \alpha \in \sfOmega^\text{top}(\cQ)$, the corresponding  Dirac structure is $$L = T\cF \oplus {\rm Ann}(T\cF) \ . $$ Then $D = T^*\cF$ and so $\alpha_D = \mu \wedge \alpha$ for some $\mu \in \det(T^*\cF)$. Thus the pushforward of $\alpha_D$ by $\varpi_!$ is non-zero. 

    Applying this construction to all $\red\alpha \in \sfOmega^\bullet(\cQ)$, the almost Dirac structure $L$ changes but we may still use the same $D=T^*\cF$. One then obtains the subspace $\sfS_D \subset \sfOmega^\bullet(M) \times \sfOmega^\bullet(\cQ)$ of pairs
    \begin{align}
        \sfS_D = \set{\Big(\alpha_D\,,\, \big(\text{\footnotesize{$\int_\cF$}}\, \mu\big)\ \red\alpha\Big) \ \Big\vert \ \red\alpha \in \sfOmega^\bullet(\cQ)}
    \end{align}
    and hence $\sfS_D = \sfS_\mu$, where $\sfS_\mu$ is as in Example \ref{ex:Qfibreintegration} and the Clifford action given by Equation \eqref{eqn:Qweakaction}. For an almost Dirac structure $L\subset \IT M$, $\sfS_D$ is an $L$-weak $Q(\cF)$-Clifford relation (see Remark \ref{rem:weakcase}) if and only if $\Omega\subset \rmp_1(\sfS_D)$ for the pure spinor line bundle $\Omega$ associated to $L$.

    We highlight that, while Drummond perturbs the almost Dirac structure $L$ to $L_D$ to obtain a reduced pure spinor with line bundle $\red\Omega$ such that \smash{$\Omega \sim_{Q(\cF)} \red\Omega$}, this can also be achieved  by considering instead an isomorphic spinor module for $\Cl(Q(\cF)^\times)$. For instance, the case $L\cap K=\set{0}$ requires the spinor module $\sfS_D$, while the case $K \subset L$ can be considered using the spinor module $\sfS_{Q(\cF)}$ as in Example~\ref{ex:pullbackspinorrel}.
    \end{remark}

\begin{definition}\label{def:invariantforms}
    Let $\cF$ be a foliation of $M$. 
    \begin{itemize}
        \item[(a)] 
    A form $\omega\in \sfOmega^\bullet(M)$ is  \emph{invariant along $\cF$} or \emph{$T\cF$-invariant} if there is a Lie subalgebra $\frf \subset \sfgamma(T\cF)$ spanning $T\cF$ pointwise such that $\pounds_X \omega = 0$ for each $X \in \frf$. \\[-3mm]
    \item [(b)] A pure spinor line bundle $\Omega \subset \midwedge^\bullet T^*M$ is \emph{invariant along $\cF$} or \emph{$T\cF$-invariant} if there is a Lie subalgebra $\frf \subset \sfgamma(T\cF)$ spanning $T\cF$ pointwise such that $\pounds_X \sfgamma(\Omega) \subset \sfgamma(\Omega)$ for each $X \in \frf$. 
    \end{itemize}
\end{definition}

If $\omega$ is invariant along $\cF$ then the line bundle  $\Omega^\omega$ generated by $\omega$ is invariant along $\cF$.

\begin{example}
    Let $\red \alpha \in \sfOmega^\bullet(\cQ)$ be a nowhere vanishing form on $\cQ = M/\cF$, and consider the pure spinor line bundle  $\Omega \subset \midwedge^\bullet T^*M$ generated by $\alpha = \varpi^*\red \alpha$. Then $\Omega$ is invariant along $\cF$, since for $X\in T\cF$ and $f\in C^\infty(M)$,
    \begin{equation}
    \begin{aligned}
        \pounds_X (f\,\varpi^* \red \alpha) &= \iota_X\, \de (f\,\varpi^* \red \alpha) + \de (f\, \iota_X \varpi^*\red\alpha)\\[4pt]
        &= \iota_X(\de f \wedge \varpi^* \red \alpha) + f\,\iota_X \varpi^*(\de \red \alpha) +0 \\[4pt]
        &= (\pounds_Xf)\, \varpi^* \red \alpha \ \in \ \Omega \ .
    \end{aligned}
    \end{equation}
    
    For the $Q(\cF)$-Clifford relation $\sfS_{Q(\cF)}$ of Example \ref{ex:pullbackspinorrel}, denote by $\red \Omega$ the pure spinor line bundle generated by $\red \alpha$. For each $u \in (\Omega\times \red \Omega) \rvert_{\gr(\varpi)}$ and $c=(m,\varpi(m))$ with $m\in M$,
    we find
    \begin{align}
        u_c = (f\, \alpha, \red g\, \red \alpha)_c = \big(f(m)\, \alpha_m, \red g(\varpi(m))\,\red \alpha_{\varpi(m)}\big) = \big(f(m)\, \alpha_m, \tfrac{\red g(\varpi(m))}{f(m)}\,f(m)\,\red \alpha_{\varpi(m)}\big) \ \in \ \sfS_{Q(\cF)c}^{\kappa(c)}
    \end{align}
    where $f\in C^\infty(M)$ and $\red g \in C^\infty(\cQ)$ are nowhere vanishing functions, while $\kappa\in C^\infty(\gr(\varpi))$ is given by $$\kappa(c) = \frac{\varpi^*\red g(m)}{f(m)} \ . $$ Thus
    \begin{equation}
    \begin{aligned}
        (\Omega \times \red \Omega)\big\rvert_{\gr(\varpi)} \otimes |\det\sfS_{Q(\cF)}^*|^{\frac1{2^n}} &= (\Omega \times \red \Omega)\big\rvert_{\gr(\varpi)} \otimes |\det(\sfS^\kappa_{Q(\cF)})^*|^{\frac1{2^n}} \\
        & \hspace{1cm} \subset \ \sfS_{Q(\cF)}^\kappa \otimes |\det(\sfS^\kappa_{Q(\cF)})^*|^{\frac1{2^n}} = \cS_{Q(\cF)} \ ,
    \end{aligned}
    \end{equation}
where $n=\dim_\IR M$. Hence $\Omega \sim_{Q(\cF)} \red \Omega$.
\end{example}

\medskip

\subsubsection{Relating Dirac Generating Operators}~\\[5pt]
Let $R \subset E_1 \times \overline{E}_2$ be an isotropic subbundle supported on $C\subset M_1\times M_2$ such that $\rho(R^\perp)\subseteq TC.$ Let $\sfD_i$ be a Dirac generating operator (DGO) on the spinor module $\sfS_{E_i}$  over $\Cl(E_i)$, for $i=1,2$. Suppose that $\sfS_R \subset (\sfS_{E_1} \times \sfS_{E_2}) \rvert_C$ is an $R$-Clifford relation of parity $j$.

\begin{definition}\label{def:DGOrelation}
    $\sfS_R$ is a \emph{$(\sfD_1,\sfD_2)$-DGO relation} if 
    $\sfS_R$ is invariant under $\sfD_1 \times  (-1)^j \, \sfD_2$. 
\end{definition}

This notion yields a Courant algebroid relation for structures defined  as in Theorem \ref{thm:CourantDirac}  through

\begin{proposition}\label{prop:DGOimpliesCArel}
  Let $R \subset E_1 \times \overline{E}_2$ be a Lagrangian subbundle supported on $C\subset M_1 \times M_2$ with trivial kernel and cokernel. Suppose that $\Cl(R)$ has a faithful irreducible spinor representation $\sfS$, and that $\sfS$ is a $(\sfD_1, \sfD_2)$-DGO relation. Then $R$ is a Courant algebroid relation.
\end{proposition}

\begin{proof}
 According to Theorem \ref{thm:CourantDirac}, 
\begin{align}
    \cbrak{r,r'} \cdot \sfs \ \in \ \sfS \ ,
\end{align}
for every $s \in \sfS$ and $r = (e_1, e_2), r' = (e_1',e_2') \in \sfgamma(R)$. Let $c \in C$.
Since $\ker(R)=\set{0}$, there is some $^R\cbrak{e_1,e_1'}_1$ such that $(\cbrak{e_1,e_1'}_1,^R\cbrak{e_1,e_1'}_1) \in R_c$. Then $(\cbrak{e_1,e_1'}_1,^R\cbrak{e_1,e_1'}_1) \cdot \sfs \in \sfS_c$, and it follows that
\begin{align}
   \big(0\cdot \sfs_1, (^R\cbrak{e_1, e_1'}_1 - \cbrak{e_2,e_2'}_2 )\cdot \sfs_2\big) \ \in \ \sfS_c \ ,
\end{align}
 for every $\sfs = (\sfs_1, \sfs_2) \in \sfS_c$.
By Proposition \ref{prop:Uisomorphismsmooth}, the second factor must vanish.
Since $\sfS$ is a faithful representation, $^R\cbrak{e_1, e_1'}_1 = \cbrak{e_2,e_2'}_2$. Since this happens for all $c\in C$, it follows that $R$ is involutive and thus defines a Courant algebroid relation.
\end{proof}

\medskip

\subsection{Integrability of Almost Dirac Structures}~\\[5pt]
Let $L$ be an almost Dirac structure of a twisted standard Courant algebroid $(\IT M, H)$, and let $\Omega$ be the corresponding pure spinor line bundle. By \cite[Proposition 3.44]{gualtieri:tesi}, $L$ is a Dirac structure (i.e.~it is integrable) if and only if the Dirac generating operator $\de_H$ satisfies $\de_H(\Omega) \subset L^* \cdot \Omega\simeq L^*\otimes\Omega$, where $L^*\simeq\IT M/L$. That is, for any nowhere vanishing section $\omega \in \sfGamma(\Omega)$, there is a section $X+\xi$ of $\IT M$ such that 
\begin{align}
    \de_H \omega = (X+\xi)\cdot\omega \ .
\end{align}

Suppose $R \colon (\IT M_1, H_1) \rel (\IT M_2, H_2)$ is a Courant algebroid relation with $R$-Clifford relation $\sfS_R$ of parity $j$, and assume that $\sfS_R$ is a $(\de_{H_1}, \de_{H_2})$-DGO relation. Suppose that $L_i$ is an almost Dirac structure on $\IT M_i$ with pure spinor line bundle $\Omega_i$, for $i=1,2$, and that $\Omega_1 \sim_R \Omega_2$. Then $L_1 \sim_R L_2$ by Proposition \ref{prop:LreliffUrelsmooth}. 

If $L_1$ is integrable, then for every $\omega_1 \in \Omega_1$ there is a section $X_1+\xi_1\in \sfgamma(\IT M_1)$ such that 
\begin{align}
    \de_{H_1} \omega_1 = (X_1+\xi_1) \cdot \omega_1 \ .
\end{align}
If $\rmp_1(R) \supseteq L_1$, then there is a section $X_2 + \xi_2 \in \sfGamma(\IT M_2)$  such that $X_1+ \xi_1 \sim_R X_2 + \xi_2$. Let $\omega_2 \in \Omega_2$ with $(\omega_1, \omega_2) \in \sfS_R$. Since $\de_{H_1}\omega_1 \sim_R(-1)^j\,\de_{H_2} \omega_2$ and $(X_1+\xi_1)\cdot \omega_1 \sim_R (-1)^j\,(X_2+\xi_2)\cdot \omega_2$, by Corollary~\ref{cor:URinjectivesmooth} it follows that
\begin{align}
    \de_{H_2} \omega_2 = (X_2 + \xi_2) \cdot \omega_2
\end{align}
and hence $L_2$ is integrable. Thus we have shown

\begin{proposition}\label{prop:integrabilityrelation}
    Let $R \colon (\IT M_1, H_1) \rel (\IT M_2, H_2)$ be a Courant algebroid relation with $R$-Clifford relation $\sfS_R$ that is a $(\de_{H_1}, \de_{H_2})$-DGO relation. Take an almost Dirac structure $L_i$ on $\IT M_i$ with associated pure spinor line bundle $\Omega_i$ for $i=1,2$. Suppose that $\Omega_1 \sim_R \Omega_2$ and that $\rmp_i(R) \supseteq L_i$. Then $L_1$ is integrable if and only if $L_2$ is integrable.
\end{proposition}

\medskip

\subsection{Spinor Relations for Geometric T-Duality}~\\[5pt] \label{ssec:SpinorrelationTduality}
In the setting of geometric T-duality with $K_1\cap K_2 = \set{0}$, we may transport pure spinor line bundles to the T-dual manifold through

\begin{theorem}\label{thm:relatedspinorline}
    Let $R \colon \mathbb{T} \mathcal{Q}_1 \dashrightarrow \mathbb{T} \mathcal{Q}_2$ be a geometric T-duality relation supported on $\red C \subset \cQ_1 \times \cQ_2$. If $\Omega_1 \subset \midwedge^\bullet T^*\cQ_1$ is a $D_1$-invariant pure spinor line bundle, then there is a pure spinor line bundle $\Omega_2 \subset \midwedge^\bullet T^* \cQ_2$ such that $\Omega_1\sim_{R} \Omega_2$.
\end{theorem}

\begin{proof}
    Let $\psi_1 \in \sfgamma(\Omega_1)$ be a section of $\Omega_1$. Then by Equation \eqref{eqn:spinorLeibnizrule},
    \begin{align}
        0=\pounds_{\red X} (\red e \cdot \psi_1) = \mathsf{ad}^{\red\sigma_1}_{\red X}(\red e) \cdot \psi_1 + \red e \cdot \pounds_{\red X} \psi_1 \ ,
    \end{align}
for $\red e \in \sfgamma(\red N{}_{\psi_1})$ and $\red X \in \sfgamma(T\cQ_1)$.  It follows that $\Omega_1$ is $D_1$-invariant if and only if the null Lagrangian subbundle $\red N{}_{\psi_1}$ is $D_1$-invariant, i.e. $\mathsf{ad}^{\red\sigma_1}_{\red X}\, \sfgamma(\red N{}_{\psi_1}) \subset \sfgamma(\red N{}_{\psi_1})$ for each $\red X \in \sfgamma(D_1)$. According to \cite[Lemma 3.32]{deFraja2025Ricci}, we may thus take a spanning set of $D_1$-invariant sections for $\red N{}_{\psi_1}$ which can be carried to $\IT\cQ_2$, and define $\red N{}_{\psi_2}$ as the span of these sections. It follows that $\red N{}_{\psi_2}$ is Lagrangian and $\red N{}_{\psi_1}\sim_R \red N{}_{\psi_2}$. There is a pure spinor line bundle $\Omega_2$ associated with $\red N{}_{\psi_2}$. 
    By Remark~\ref{rem:smoothULequiv}, at each point $c = (c_1, c_2)\in \red C$ it follows that $\Omega_1{}_{c_1} \sim_R \Omega_2{}_{c_2}$. Since this is true for all $c\in \red C$, it follows that $\Omega_1\sim_R \Omega_2$.
\end{proof}

\medskip

\subsubsection{Twisted Differentials and the Fourier-Mukai Transform}\label{ssec:twisteddifferential}~\\[5pt]
We shall now show a converse to Proposition \ref{prop:DGOimpliesCArel}. By this we mean to show how, in this setting, the geometric T-duality relation gives rise to a DGO relation. To do this, we will restrict to the case that $\cQ_1$ and $\cQ_2$ are both fibre bundles over a manifold $\cB$, with bundle projections $\pi_i:\cQ_i\to\cB$, and 
\begin{align}
    M = \cQ_1\times_\cB\cQ_2 \ \subset \ \cQ_1\times\cQ_2
\end{align}
is the correspondence space with the canonical projections $\varpi_i=\rmp_i:M\to\cQ_i$ for $i=1,2$. Then the support of the T-duality relation 
\begin{equation}
\begin{tikzcd} 
R:(\IT\cQ_1,\red H{}_1) \ar[r,dashed] & (\IT\cQ_2,\red H{}_2)
\end{tikzcd}
\end{equation}
is $\red C=M$.
We will construct a spinor module $\sfS_R$ for $\Cl(R)$, and subsequently construct the canonical Dirac generating operator for this spinor module (cf. Theorem~\ref{thm:DGOconverse}), following the constructions of~\cite{Cortes:2019roa}. We will then show that this Dirac generating operator is in fact given by the pair of twisted differentials $(\de_{\red H{}_1}, \de_{\red H{}_2})$, extended from $\sfOmega^\bullet(\cQ_1)\times\sfOmega^\bullet(\cQ_2)$ to $\sfS_R$ as graded derivations.

The construction of $\sfS_R$ is by means of  a Courant algebroid $R_\cB$ over $\cB$ whose typical fibre is isomorphic to a fibre of $R$. In the following we denote $\red E{}_i=\IT\cQ_i$ and abbreviate subscripts ${}_{\red E{}_i}$ simply by ${}_i$, for $i=1,2.$

\begin{proposition}\label{prop:regCAonB}
    There is a transitive Courant algebroid $R_\cB$ over $\cB$ which is fibrewise isomorphic (as a vector bundle) to $R$. 
\end{proposition}

\begin{proof}
    $\cQ_1$ and $\cQ_2$ are both fibre bundles over the base manifold $\cB$, fitting the commutative diagram
    \begin{equation} \label{cd:correspondencespace}
        \begin{tikzcd}[row sep = 1cm]
             & M \arrow[swap]{dl}{\varpi_1} \arrow{dr}{\varpi_2} & \\
            \cQ_1 \arrow[swap]{dr}{\pi_1} &  & \cQ_2 \arrow{dl}{\pi_2} \\
             & \cB & 
        \end{tikzcd}
    \end{equation} 
    Then $\sfGamma(\red E{}_1 \times \overline{\red E}_2; R)$ is a locally-free finitely-generated $C^\infty(\cB)$-module, where the $C^\infty(\cB)$-module structure is given by pullback $\pi_1^*\times\pi_2^*:C^\infty(\cB)\to C^\infty(\cQ_1\times\cQ_2)$. Hence by the smooth version of the Serre-Swan Theorem (as formulated in \cite[Chapter 11]{Nestruev2020}) there exists a vector bundle $R_\cB$ over $\cB$ whose space of sections coincides with $\sfGamma(\red E{}_1 \times \overline{\red E}_2; R)$. 
    
    The bundle $R_\cB\to\cB$ inherits the structure of a Courant algebroid in the following way. Let $c = (c_1,c_2)\in \red C$, which has a corresponding point $b=\pi_1(\varpi_1(c_1))=\pi_2(\varpi_2(c_2))\in \cB$. We denote a point of the fibre $R_\cB{}_b$ by the same symbol as the corresponding point $(\red e_1,\red e_2) \in R_c$, where the context will be clear. By the construction of $R$, this is independent of the point $c \in \red C$ corresponding to $b\in\cB$. 
    
    Then the anchor $\rho_{R_\cB}:R_\cB\to T\cB$ is given by
    \begin{align}
        \rho_{R_\cB}(\red e_1,\red e_2)_b \coloneqq \pi_1{}_* \,\red \rho{}_1(\red e_1)_{c_1}=\pi_2{}_*\,\red \rho{}_2(\red e_1)_{c_2} \ ,
    \end{align}
    which is well-defined since $$\rho(\red e_1,\red e_2) = \big(\red\rho{}_1(\red e_1),\red\rho{}_2(\red e_2)\big) \ \in \ T\red C $$  projects to the same vector in $T\cB$ under pushforward by $\pi_1{}_* \circ \varpi_{1*}$ and $\pi_2{}_* \circ \varpi_{2*}$. Similarly, the pairing is given by
    \begin{align}
        \big\langle (\red e_1,\red e_2) \,,\, (\red e'_1,\red e'_2)\big\rangle_{R_\cB}(b) := \langle\red e_1,\red e_1'\rangle_1(c_1) =  \langle\red e_2 , \red e_2'\rangle_2(c_2) \ ,
    \end{align}
    while the Dorfman bracket is
    \begin{align}
    \cbrak{(\red e_1,\red e_2)\,,\,(\red e'_1,\red e'_2)}_{R_\cB{}_b} := \big( \cbrak{\red e_1,\red e_1'}_{1 c_1} \,,\, \cbrak{\red e_2,\red e_2'}_{2 c_2}\big) \ ,
    \end{align}
    where we used the usual Courant algebroid structures on $R\subset \red E{}_1\times\overline{\red E}_2$.
\end{proof}

By a result of Chen-Sti\'enon-Xu~\cite[Lemma 1.2]{Chen2013}, a regular Courant algebroid $E$ admits  isomorphisms
\begin{align}
    E \simeq \Im(\rho_E) \oplus \ker(\rho_E) \simeq \Im(\rho_E)\oplus \ccG \oplus \ker (\rho_E)^\perp \ ,
\end{align}
such that $\Im(\rho_E)$ is isotropic in $E$ and $\ccG\simeq\ker(\rho_E)/\ker(\rho_E)^\perp$ is orthogonal to $\Im(\rho_E)$ in $E$. In this splitting, the restrictions of the Dorfman bracket and pairing of $E$ make $\ccG$ into a bundle of quadratic Lie algebras.

For the transitive Courant algebroid $R_\cB$ constructed in Proposition~\ref{prop:regCAonB}, this provides the splitting
\begin{align}\label{eqn:splittingforRB}
    R_\cB \simeq T\cB \oplus \ccG \oplus T^*\cB \ .
\end{align}
The choice $$\ccG\simeq\ker(\rho_{R_\cB})\,\big/\,\ker(\rho_{R_\cB})^\perp \ \subset \ \Im(\rho_{R_\cB})^\perp=T\cB^\perp$$  is canonically associated to $R_\cB$, and the projection $\rmp_\ccG$ to $\ccG$ is calculated from the relation $R\subset\IT\cQ_1\times\overline{\IT\cQ_2}$: for $g\in \ccG$ we choose a lift $\widehat g=(\dot g,\ddot g)\in\ker(\rho_{R_\cB})\subset R$. Note that $\ker(\rho_{R_\cB})\subset T\red\cF{}_1\times T\red\cF{}_2$ in this lift to $R$. The choice of splitting \eqref{eqn:splittingforRB} is then equivalent to the unique lifts of $X\in T\cB$ to $\dot X \in T\cQ_1$ and $\ddot X \in T\cQ_2$ such that $\widehat X\coloneqq(\dot X, \ddot X) \in R$, while for $\omega\in T^*\cB$ we set $\widehat\omega \coloneqq (\pi_1^*\omega, \pi_2^* \omega)\in R$. 

The Courant algebroid structures in the splitting \eqref{eqn:splittingforRB} are then given by
\begin{align}
    \ip{\widehat X + \widehat g + \widehat\omega, \widehat Y + \widehat g' + \widehat\chi}_{R_\cB} &= \tfrac 12\,(\iota_X \omega + \iota_Y \chi) + \ip{\widehat g,\widehat g'}_{R_\cB} \  , \\[4pt]
    \rho_{R_\cB}(\widehat X + \widehat g + \widehat\omega) &= X \ , \\[4pt]
    \cbrak{\widehat X + \widehat g + \widehat\omega, \widehat Y + \widehat g' + \widehat \chi}_{R_\cB} &=\big(\cbrak{\dot X + \dot g + \pi_1^* \omega, \dot Y + \dot g' + \pi_1^* \chi}_1\,,\,\cbrak{\ddot X + \ddot g + \pi_2^* \omega, \ddot Y + \ddot g' + \pi_2^* \chi}_2\big) \ .
\end{align}
One easily checks that these operations are well-defined.

By results of Alekseev-Xu \cite[Theorem 4.1]{Alekseev2001}, if $E$ is a regular Courant algebroid then every spinor bundle $\sfS$ over $\Cl(E)$  locally admits a Dirac generating operator which can be expressed by means of a Courant algebroid connection, as we discuss below. On the weighted spinor bundle there is a canonical Dirac generating operator \cite[Theorem 52]{Cortes:2019roa}. For illustration, let us first discuss how this works in the case of principal circle bundles.

\begin{example}[\textbf{Circle Bundles}]\label{ex:DGOcircle}
Let us consider the T-duality relation $R$ when $\cQ_1$ and $\cQ_2$ are circle bundles over $\cB$, see~\cite[Section~6]{DeFraja:2023fhe} and~\cite[Example~2.42]{deFraja2025Ricci}. The restriction of the Courant algebroid $R_\cB$ to any open subset $U \subset \cB$ is given by
\begin{align}
    \sfGamma(U,R_\cB) = \set{(X,X)  |  X\in \sfGamma(TU)} \, \oplus \, \set{
    (\xi,\xi)  |  \xi\in \sfGamma(T^*U)} \, \oplus \, \text{Span}_{C^\infty(U)}\set{(\partial_{\theta_1} , \theta_2) \,,\, (\theta_1, \partial_{\theta_2})}
\end{align}
for a choice of connections $\theta_1$ and $\theta_2$ on $\cQ_1\to \cB$ and $\cQ_2 \to \cB$ respectively, pulled back to $\cB$ and restricted to $U$. The vector fields $\partial_{\theta_i}$ are dual to $\theta_i$ for $i=1,2.$

Thus 
\begin{align}
    \sfGamma\big(U,\ker (\rho_{R_\cB})\big) & = \set{(\xi,\xi)  |  \xi\in \sfGamma(T^*U)} \, \oplus \, \text{Span}_{C^\infty(U)}\set{(\partial_{\theta_1}, \theta_2) \,,\,(\theta_1, \partial_{\theta_2})} \ ,\\[4pt]
    \sfGamma\big(U,\ker (\rho_{R_\cB} )^\perp  \big) & = \set{(\xi, \xi)  |  \xi\in \sfGamma(T^*U)} \ .
\end{align}
It follows that the restriction of the bundle $\ccG\simeq\ker(\rho_{R_\cB})/\ker(\rho_{R_\cB})^\perp$ to $U$ is given by $$\sfGamma(U,\ccG) = \text{Span}_{C^\infty(U)}\set{(\partial_{\theta_1}, \theta_2) \,,\,(\theta_1, \partial_{\theta_2})}$$ and 
\begin{align}
    R_\cB \big\rvert_U \simeq TU \oplus T^*U \oplus \ccG \big\rvert_U \ .
\end{align}
In the following we drop the explicit restriction to $U\subset\cB$ from the notation for brevity.

In order to define the bracket of two sections of $\Im(\rho_{R_\cB}) = T\cB$, we note that for circle bundles (and more generally principal bundles) a choice of lift of $X \in T\cB$ to $(\dot X, \ddot X) \in R$ is equivalent to the choice of connections $(\theta_1, \theta_2$), and the lifts are thus  the horizontal lifts $(X^{\mathrm{h}_1}, X^{\mathrm{h}_2})\in T\cQ_1\times T\cQ_2$ with respect to these connections. The requirement that $(X^{\mathrm{h}_1}, X^{\mathrm{h}_2}) \in R$ is the condition $$\red H{}_1- \red H{}_2 = \de (\theta_1 \wedge \theta_2) \ , $$ as discussed in \cite{cavalcanti2011generalized}. 

The twisting three-forms are given by~\cite{cavalcanti2011generalized} $$\red H{}_i = \pi_i^*h +c_{3-i} \wedge \theta_i \ , $$ for $i=1,2$, where $c_j=\de\theta_j$ represents the Chern class of $\cQ_j\to\cB$ for $j=1,2$ and $h \in \sfOmega^3(\cB)$.
The bracket of two elements of $\Im(\rho_{R_\cB})=T\cB$ is then
\begin{align}
    \begin{aligned}\label{eqn:bracketonR}
    \llbracket (X,X)\,,\,(Y,Y) \rrbracket_{R_\cB} &= \big([X,Y]\,,\,[X,Y]\big) + \rmp_{T^*\cB}\big( \red H{}_1(X^{\mathrm{h}_1}, Y^{\mathrm{h}_1},\,\cdot\,)\,,\,\red H{}_2(X^{{\mathrm{h}_2}},Y^{\mathrm{h}_2},\,\cdot\,)\big) \\
    & \qquad \, + \rmp_\ccG\big([X^{\mathrm{h}_1},Y^{\mathrm{h}_1}] + \red H{}_1(X^{\mathrm{h}_1}, Y^{\mathrm{h}_1},\,\cdot\,)\,,\, [X^{\mathrm{h}_2}, Y^{\mathrm{h}_2}]+ \red H{}_2(X^{\mathrm{h}_2},Y^{\mathrm{h}_2},\,\cdot\,)\big)\\[4pt]
    &=\big([X,Y]\,,\,[X,Y]\big) + \big(h(X,Y,\,\cdot\,)\,,\, h(X,Y,\,\cdot\,)) \\
    & \qquad \, + c_1(X, Y)\,( \partial_{\theta_1}, \theta_2) + c_2(X, Y)\,(\theta_1, \partial_{\theta_2}) \ .
\end{aligned}
\end{align}

By \cite[Section 1.2]{Alekseev:2007pure}, $(1,\theta_2)$ is a pure spinor for $\Cl(\ccG)$ with Lagrangian subbundle $L\subset \ccG$ given by \smash{$\sfGamma(L) = \text{Span}_{C^\infty(\cB)}\set{(\partial_{\theta_1}, \theta_2)}$}, and thus defines a spinor bundle $ \sfS_\ccG=\Cl(\ccG)/(\Cl(\ccG)\cdot L)$. The subbundle $L^*\subset \ccG$ transverse to $L$ given by $$\sfGamma(L^*) = \text{Span}_{C^\infty(\cB)}\set{(\theta_1, \partial_{\theta_2})}$$ provides an isomorphism $$\sfGamma(\sfS_\ccG) \simeq \sfGamma(\midwedge^\bullet L^*) \cdot (1,\theta_2) = \textrm{Span}_{C^\infty(\cB)}\set{(1,\theta_2)\,,\,(\theta_1, 1)} \ . $$ 
The spinor bundle for $\Cl(R_{\cB}) = \Cl(\IT\cB)\otimes\Cl(\ccG)$ is then
\begin{align}
    \sfS_{R} = \midwedge^\bullet T^*\cB \otimes \sfS_\ccG \ . 
\end{align}
It has $-$ parity.

We define the map $\overline{W}:\sfGamma(\sfS_{R})\to\sfGamma(\sfS_{R})$ by
\begin{align}\label{eqn:wbarincirclecase}
    \overline{W}(\omega \otimes \sfs) = \omega \wedge(c_2, c_2) \otimes (\theta_1,\partial_{\theta_2}) \cdot \sfs - \omega \wedge(c_1,c_1)\otimes (\partial_{\theta_1},\theta_2) \cdot \sfs \ ,
\end{align}
for $\omega \in \Omega^\bullet(\cB)$ and $\sfs \in \sfGamma(\sfS_\ccG)$.
As in \cite[Theorem 61]{Cortes:2019roa}, one can then define the canonical Dirac generating operator on $\sfS_{R}$ as
\begin{align}
    \sfD(\omega \otimes \sfs) = (\de \omega) \otimes \sfs - \big((h,h) \wedge \omega\big) \otimes \sfs +(-1)^{|\omega|+1}\,\overline{W}(\omega \otimes \sfs) \ .
\end{align}
After pullback by $\pi_1^* \times \pi_2^*$, this reduces to the Dirac generating operator $\de_{\red H{}_1}\times(-\de_{\red H{}_2})$ on $\sfS_{R}$ and  provides a $(\de_{\red H{}_1},\de_{\red H{}_2})$-DGO relation.
\end{example}

We now extend Example~\ref{ex:DGOcircle} to the general case where $\cQ_1$ and $\cQ_2$ are fibred over the same manifold $\cB$ with surjective submersions $\pi_i (\cQ_i)= \cB$.
By Proposition~\ref{prop:regCAonB}, there is a transitive Courant algebroid $R_\cB$ over $\cB$ that is fibrewise isomorphic to $R\subset \IT\cQ_1\times\overline{\IT\cQ_2}$. 
Given a splitting $R_\cB \simeq \IT \cB \oplus \ccG$, define $$\ccG_1^\star = \rmp_{T^*\cQ_1}\,\rmp_1(\ccG) \ \subset \ T^*\cQ_1 \ . $$  Consider the spinor $$u = (\lambda, 1)$$ of the Clifford algebra bundle $\Cl(\ccG)$, where $\lambda\in \det \ccG_1^\star$ (this can always be done locally, provided that the typical fibre of $D_1 = T\red\cF{}_1$ is orientable). This is a pure spinor, with Lagrangian subbundle $L\subset \ccG$. Choose its dual $L^*$ such that $$\red\rho{}_2(\rmp_2(L^*)) = \set{0} \ . $$ 

This then yields a spinor bundle $$\sfS_\ccG = \midwedge^\bullet L^* \cdot u \ \subset \ \left(\midwedge^\bullet T^* \cQ_1 \times \midwedge^\bullet T^* \cQ_2\right)\big\rvert_{\red C} \ . $$ The bundle
\begin{align}\label{eqn:Tdualityrelspinrep}
    \sfS_R = \midwedge^\bullet T^*\cB \otimes \sfS_\ccG 
\end{align}
over $\cB$ is a spinor representation for the Clifford bundle $\Cl(R_\cB)$, with action
\begin{align}
    (\widehat X + \widehat g + \widehat\chi) \cdot (\omega \otimes \sfs) = (X+\chi) \cdot \omega \otimes  \sfs + (-1)^{|\omega|}\,\omega \otimes g\cdot \sfs
\end{align}
for $\widehat X + \widehat g + \widehat\chi \in R_\cB$, $\omega \in \midwedge^\bullet T^*\cB$, and $\sfs \in \sfS_\ccG$. These actions are performed using the Clifford action \eqref{eqn:productcliffordaction} of $\Cl(R_\cB)$ on $(\midwedge^\bullet T^*\cQ_1 \times \midwedge^\bullet T^* \cQ_2) \rvert_{\red C}$ by pulling $\sfS_R$ back by $\pi_1^*\times \pi_2^*$, where we recall that $\red C \simeq \cQ_1 \times_\cB \cQ_2$. By construction, $\sfS_R$ is also a spinor representation for $\Cl(R)$, with the same action. We leave the pullback by $\pi_1^*\times \pi_2^*$ implicit in the following.

\begin{proposition}
    $\sfS_R$ is an $R$-Clifford relation of parity $j= \dim_\IR (\red\cF{}_1) \mod 2$.
\end{proposition}

\begin{proof}
    By \cite[Remark 57]{Cortes:2019roa} it follows that $\sfS_R$ is an irreducible spinor representation for $\Cl^j(R_\cB)$, where $j=\frac12\,\rk_\IR(\ccG)\mod 2$, and hence for $\Cl^j(R)$. In the present situation $$\rk_\IR(\ccG)=2\,\rk_\IR(L)=2\dim_\IR(\red\cF{}_1) \ , $$
    and the result follows.
\end{proof}

In accordance with \cite{Cortes:2021aab}, the canonical $R$-Clifford relation we consider is $$\cS_R = \midwedge^\bullet T^*\cB \otimes \sfS_\ccG \otimes |\det\sfS_\ccG^*\,|^{\frac1{\tt g}} \ , $$ where $\tt g$ is the rank of $\sfS_\ccG$, since the bundle of forms $\midwedge^\bullet T^*\cB$ is canonically defined over $\red C$.

\begin{example}[\textbf{The Fourier-Mukai Transform}]\label{ex:FourierMukai}
Consider the setting of T-dual torus bundles described in \cite{ bouwknegt2004tduality,cavalcanti2011generalized}. Let $\cQ_1$ and $\cQ_2$ be principal $\sfT^k$-bundles over $\cB$, where $\sfT^k$ is a $k$-dimensional torus, endowed with closed $\sfT^k$-invariant three-forms $\red H{}_1$ and $\red H{}_2$ respectively. Then $\cQ_1\times_\cB\cQ_2$ is a principal $\sfT^{2k}$-bundle over $\cB$. Choose a $\sfT^{2k}$-invariant two-form \smash{$B\in\sfOmega^2_{\sfT^{2k}}(\cQ_1\times_\cB\cQ_2)$} whose restriction to $\ker(\varpi_1{}_*)\otimes\ker(\varpi_2{}_*)$ is non-degenerate and which satisfies 
\begin{align}
   \de B = \varpi_1^*\,\red H{}_1-\varpi_2^*\, \red H{}_2 \ .
\end{align}
Let
\begin{align}\label{eq:ccRiso}
    \ccR:\sfGamma_{\sfT^k}(\IT\cQ_1)\longrightarrow\sfGamma_{\sfT^k}(\IT\cQ_2)
\end{align}
be an isomorphism of $C^\infty(\cB)$-modules, where the $C^\infty(\cB)$-module structure on $\sfGamma_{\sfT^k}(\IT\cQ_i)$ is given by pullback $\pi_i^*:C^\infty(\cB)\to C^\infty(\cQ_i)$ for $i=1,2$.

There is an isomorphism of twisted cohomology $$\varrho \colon \sfH^\bullet(\cQ_1, \red H{}_1) \longrightarrow \sfH^\bullet(\cQ_2, \red H{}_2) \ , $$ which at the level of differential forms is given by 
\begin{align}\label{eqn:FMT}
    \varrho(\omega) \coloneqq \int_{\sfT^k}\, \e^{B}\wedge \varpi_1^*\,\omega  
\end{align}
for $\omega\in\sfOmega_{\sfT^k}^\bullet(\cQ_1)$, where the integration is along the fibres of the projection $\varpi_2:\cQ_1 \times_{\cB} \cQ_2 \to \cQ_2$. The map $\varrho$ satisfies the crucial properties\footnote{See Appendix~\ref{appendix} for the origin of the grading factors $(-1)^k$ in \eqref{eqn:FMTproperties}.} 
\begin{align}\label{eqn:FMTproperties}
    \de_{\red H{}_2} \varrho( \omega) = (-1)^k\, \varrho(\de_{\red H{}_1} \omega) \qquad \text{and} \qquad \mathscr{R}(e) \cdot\varrho( \omega) = (-1)^k\,\varrho( e\cdot \omega) \ ,
\end{align}
for $\omega\in\sfOmega_{\sfT^k}^\bullet(\cQ_1)$, $e\in\sfGamma_{\sfT^k}(\IT\cQ_1)$ and a unique choice of \eqref{eq:ccRiso}~\cite{cavalcanti2011generalized}.
This implies that $\varrho$ is an isomorphism between the twisted differential complexes of torus-invariant forms $(\sfOmega_{\sfT^k}^\bullet(\cQ_1), \de_{\red H{}_1})$ and $(\sfOmega_{\sfT^k}^\bullet(\cQ_2), \de_{\red H{}_2})$, as well as an isomorphism of irreducible Clifford modules. This is a smooth version of the Fourier-Mukai transform.

By \cite[Lemma 6.5]{DeFraja:2023fhe}, there is a T-duality relation $R \colon (\IT \cQ_1,\red H{}_1) \rel (\IT \cQ_2,\red H{}_2)$ supported on $\red C=\cQ_1 \times_\cB \cQ_2$. When restricted to $\sfT^k$-invariant sections it gives rise to the $C^\infty(\cB)$-module isomorphism \eqref{eq:ccRiso} by~\cite[Proposition~6.13]{DeFraja:2023fhe}; with a slight abuse of notation, we continue to use the same symbol $R$ for this restriction. Because of the second property in \eqref{eqn:FMTproperties}, it follows that the graph of the Fourier-Mukai transform $$\gr(\varrho) \ \subset \ \big(\sfOmega_{\sfT^k}^\bullet(\cQ_1) \times \sfOmega_{\sfT^k}^\bullet(\cQ_2)\big)\big|_{\red C} $$ is invariant under the action of $\Cl(R)$ defined as follows: if $c=(c_1,c_2) \in \red C$,  $(e_1, e_2) \in R_c$ and $\omega\in (\midwedge^\bullet T^*\cQ_1/\sfT^k)_{c_1}$, then
\begin{align}
    (e_1, e_2) \cdot \big(\omega, \varrho(\omega)\big) := \big(e_1 \cdot \omega , (-1)^k\, e_2 \cdot \varrho(\omega)\big) = \big(e_1 \cdot \omega, \varrho(e_2 \cdot \omega)\big) \ \in \ \gr(\varrho)_c \ ,
\end{align}
where $k$ is the dimension of the fibres $\sfT^k$ and hence defines the parity of a representation of $\Cl(R)$.
Thus $\gr(\varrho)$ is a spinor representation of $\Cl^{(-1)^k}(R)$.

We show that $\gr(\varrho)$ is an $R$-Clifford relation of parity $k\mod 2$. Suppose that there is a $\Cl(R)$-invariant subbundle $\sfS' \subset \gr(\varrho)$. It follows that $\sfS' = \gr(\varrho|_{\sfOmega_1})$ for some subspace $\sfOmega_1 \subset \sfOmega_{\sfT^k}^\bullet(\cQ_1)$. Then $\sfOmega_1$ is invariant under $\Cl(\IT \cQ_1/\sfT^k)$ since $\sfS'$ is invariant under $\Cl(R)$. By the irreducibility of $\sfOmega_{\sfT^k}^\bullet(\cQ_1)$ as a $\Cl(\IT \cQ_1/\sfT^k)$-module, it follows that either $\sfOmega_1 = \set{0}$ or $\sfOmega_1 = \sfOmega_{\sfT^k}^\bullet(\cQ_1)$, and hence $\gr(\varrho)$ is an irreducible representation of \smash{$\Cl^{(-1)^k}(R)$}. 

Since an irreducible Clifford module is unique up to isomorphism, it follows that the irreducible representation $\sfS_R$ constructed in \eqref{eqn:Tdualityrelspinrep} is isomorphic to $\gr(\varrho)$. Thus the $R$-Clifford relation $\sfS_R$ recovers the Fourier-Mukai transform.
\end{example}

In order to write a Dirac generating operator on $\sfS_R$, following \cite[Section 1.3]{Chen2013} we introduce the maps 
\begin{align}
    W &\colon \sfgamma(T\cB) \times \sfgamma(T\cB) \longrightarrow \sfgamma(\ccG) \ , \qquad \qquad \hspace{0.9cm} \ W(X,Y)= \rmp_\ccG \cbrak{\widehat X,\widehat Y}_{R_\cB} \ , \\[4pt]
    \cH &\colon \sfgamma(T\cB) \times \sfgamma(T\cB) \times \sfgamma(T\cB) \longrightarrow C^\infty(\cB) \ , \quad \cH(X,Y,Z) = \ip{\rmp_{T^*\cB} \cbrak{\widehat X,\widehat Y}_{R_\cB},\widehat Z}_{R_\cB} \ , \\[4pt]
    \nabla & \colon \sfgamma(T\cB) \times \sfgamma(\ccG) \longrightarrow \sfgamma(\ccG) \ , \qquad \qquad \quad \hspace{1.65cm} \nabla_X g = \rmp_\ccG\cbrak{\widehat X,\widehat g}_{R_\cB} \ ,
\end{align}
for all $X,Y,Z \in \sfGamma(T\cB)$ and $g \in \sfgamma(\ccG)$. We write $\widehat\cH = (\dot h, \ddot h)\in\sfGamma(\midwedge^3 R)$ with $h \in \sfOmega^3(\cB)$.

Let $\nabla^{\sfS_\ccG}:\sfGamma(T\cB)\times\sfGamma(\sfS_\ccG)\to\sfGamma(\sfS_\ccG)$ be a $T\cB$-connection on $\sfS_\ccG$ compatible with $\nabla$, that is
\begin{align}
    \nabla^{\sfS_\ccG}_X(g\cdot \sfs) = \nabla_X g \cdot \sfs+ g \cdot \nabla^{\sfS_\ccG}_X \sfs \ ,
\end{align}
for $X\in\sfGamma(TM)$, $g\in\sfGamma(\ccG)$ and $\sfs\in\sfGamma(\sfS_\ccG).$
Let $C\in\sfGamma(\midwedge^3 \ccG^*)$ be the three-form on $\ccG$ given by $$C(g,g',g'')=\ip{\cbrak{\widehat g,\widehat g'}_{R_\cB},\widehat g''}_{R_\cB} \ , $$ for $g,g',g''\in\sfGamma(\ccG)$, viewed as a section of $\Cl(\ccG)$; then $C$ acts by Clifford multiplication on $\sfS_\ccG$. Finally, define the map $\overline{W}:\sfGamma(\sfS_R)\to\sfGamma(\sfS_R)$ as
\begin{align}
    \overline{W}(\omega \otimes \sfs) = \frac 12\, \sum_{i,j,\mu} \, \ip{\reallywidehat{W(X_i,X_j)}\,,\, \widehat g_\mu}_{R_\cB} \  \alpha_i\wedge  \alpha_j \wedge \omega \otimes  \tilde g_\mu \cdot \sfs \ ,
\end{align}
for $\omega\in\sfOmega^\bullet(\cB)$ and $\sfs\in\sfGamma(\sfS_\ccG)$, where $(X_i)$ is a (local) basis for $T\cB$ with $ (\alpha_i)$ its canonical dual basis for $T^*\cB$, i.e. $\alpha_i(X_j) = \iota_{X_j}\alpha_i = \delta_{ij}$, while $(g_\mu)$ is a basis for $\ccG$ with $ (\tilde g_\mu)$ its dual basis for $\ccG^*\simeq\ccG$ with respect to the pairing on $\ccG$ induced by $\ip{\,\cdot\,,\,\cdot\,}_{R_\cB}$, i.e.~$ \tilde g_\mu( g_\nu) = \ip{\widehat g_\mu, \widehat g_\nu}_{R_\cB} = \delta_{\mu \nu}$.

With these ingredients, the canonical Dirac generating operator $\sfD$ on $\sfS_R$ is then given by \cite[Theorem 61]{Cortes:2019roa}
\begin{align}
\begin{aligned}\label{eqn:canonicalDGO}
    \sfD(\omega\otimes \sfs) &=
    (\de \omega) \otimes \sfs + \nabla^{\sfS_\ccG}(\sfs) \wedge \omega - (\cH \wedge \omega) \otimes \sfs \\
    & \qquad \, + (-1)^{|\omega|+1}\, \overline{W}(\omega \otimes \sfs) + \tfrac 14\,(-1)^{|\omega|+1}\, \omega \otimes C\cdot\sfs \ .
\end{aligned}
\end{align}
Our main result concerning a converse to Proposition \ref{prop:DGOimpliesCArel}, which can be formulated in the realm of geometric T-duality, is now formulated as

\begin{theorem}\label{prop:canonicalDGOforR}
Let $R\colon (\IT \cQ_1, \red H{}_1) \rel (\IT \cQ_2, \red H{}_2)$ be a geometric T-duality relation, and let $R_\cB$ be the transitive Courant algebroid constructed in Proposition \ref{prop:regCAonB}. Let $\sfS_R$ be the $R$-Clifford relation of  parity $j=\dim_\IR(\red\cF{}_1)\mod 2$ defined by \eqref{eqn:Tdualityrelspinrep} with respect to the splitting of $R_\cB$ given in \eqref{eqn:splittingforRB}.
Let $\sfD$ be the canonical Dirac generating operator on $\sfS_R$ given by \eqref{eqn:canonicalDGO}. Choose $\lambda\in \det(\ccG^\star_1)$ so that the pure spinor line bundle $\Omega^\lambda$ generated by $\lambda$ is $D_1$-invariant, i.e. $\pounds_{\red X} \lambda \in \Omega^\lambda$ for all $\red X\in \frk{k}_{\tau_1}$. 

Then $$\sfD = \de_{\red H{}_1}\times(-1)^j\, \de_{\red H{}_2} \ . $$  In particular, $\sfS_R$ is a $(\de_{\red H{}_1}, \de_{\red H{}_2})$-DGO relation.
\end{theorem}

\begin{proof}
    Let $\widehat{\omega}=(\pi_1^*\omega, \pi_2^*\omega)$ and $\widehat{\mathsf{s}} = (\dot\ell,\ddot\ell\,) \cdot u=(\iota_{\dot\ell}\lambda,\ddot\ell\,)$, where $u = (\lambda,1)$  is $D_1$-invariant. In the following calculations, we omit the parity $j$ of the $R$-Clifford relation for brevity. 
    We consider each term in Equation~\eqref{eqn:canonicalDGO} individually: The first term simply becomes
    \begin{align}
        (\de\omega)\otimes\sfs = (\de\,\pi_1^*\omega\otimes\iota_{\dot\ell}\lambda,\de\,\pi_2^*\omega\otimes\ddot\ell\,) \ .
    \end{align}
    For the second term, take a $T\cB$-connection $\nabla^{\sfS_\ccG}$ on $\sfS_\ccG$ for which $u=(\lambda,1)$ is parallel:
    \begin{align}
        \nabla^{\sfS_\ccG}u = \nabla^{\sfS_\ccG}(\lambda,1) = (0,0) \ ,
    \end{align}
    and extend to the whole of $\sfS_\ccG$ by compatibility with $\nabla$.

    Let $\widehat X_i = (\dot X_i, \ddot X_i)$ and $\widehat \alpha_i=(\pi_1^*\alpha_i,\pi_2^*\alpha_i)$, with $i=1,\ldots,\dim_\IR\cB$. Let $n = \dim_\IR \red\cF{}_1 = {\rm rk}_\IR(L)$. Choose the basis $(g_\mu)$ for $\ccG=L\oplus L^*$ such that $(g_\mu)_{1\leqslant \mu\leqslant n}\subset L$ and $(g_\mu)_{n< \mu\leqslant 2n}\subset L^*$; then the dual basis $(\tilde g_\mu)$ is seen as a subset of $\ccG=L^* \oplus L$. Let $\widehat g_\mu=(\dot g_\mu,\ddot g_\mu)$ and \smash{$\widehat {\tilde g}_\mu=(\dot {\tilde g}_\mu,\ddot{\tilde  g}_\mu)$}. Let $\ell\in L^*$ such that \smash{$\widehat\ell = (\dot\ell,\ddot\ell) \in \frk{k}_{\tau_1}$}, and set $\sfs = \ell\cdot u$.
    Then
    \begin{equation}
    \begin{aligned}
        \nabla\ell &= \sum_i\, \alpha_i \wedge \nabla_{X_i}\ell \\[4pt]
        &= \sum_{i}\, \widehat \alpha_i  \wedge \rmp_\ccG\big(\cbrak{\dot X_i, \dot\ell\,}_1, \cbrak{\ddot X_i, \ddot\ell\, }_2\big)\\[4pt]
        &= \sum_i\,\widehat \alpha_i \wedge\Big( \sum_{\mu \leqslant n}\,  \big(\dot {\tilde g}_\mu\, (\iota_{[\dot X_i, \dot\ell\,]} \dot g_\mu),\, \ddot {\tilde g}_\mu \, (\iota_{\ddot g_\mu} \, \pounds_{\ddot X_i} \ddot\ell\,)\big) + \sum_{\mu >n}\,  \big( (\iota_{\dot g_\mu}\,\iota_{\dot X_i}\, \iota_{\dot\ell}\, \red H{}_1)\,\dot {\tilde g}_\mu,\, 0\big)\Big) \ .
    \end{aligned}
    \end{equation}

    Hence
    \begin{equation}
    \begin{aligned}
        \nabla^{\sfS_\ccG} \sfs &= \sum_i\, \alpha_i \wedge\nabla^{\sfS_\ccG}_{X_i}\sfs \\[4pt]
        &= \sum_i\,\alpha_i \wedge \nabla_{X_i}\ell \cdot u \\[4pt]
        &= \sum_i\,\widehat \alpha_i \wedge\Big(\sum_{\mu \leqslant n}\,  \big(\iota_{[\dot X_i, \dot\ell\,]} \dot g_\mu\, \iota_{\dot g_\mu}\lambda,\, \ddot {\tilde g}_\mu \, (\iota_{\ddot g_\mu}\, \pounds_{\ddot X_i} \ddot\ell\,)\big) + \sum_{\mu >n}\, \big((\iota_{\dot g_\mu}\,\iota_{\dot X_i}\, \iota_{\dot\ell}\, \red H{}_1)\, \dot {\tilde g}_\mu\wedge \lambda\, ,\, 0\big)\Big) \\[4pt]
        &= \sum_{\mu \leqslant n,i}\, \widehat \alpha_i \wedge \big(\iota_{[\dot X_i, \dot\ell\,]} \lambda, \, \ddot {\tilde g}_\mu \, (\iota_{\ddot g_\mu}\,\pounds_{\ddot X_i} \ddot\ell\,)\big) \\[4pt]
        &= \sum_{\mu \leqslant n,i}\, \widehat \alpha_i \wedge \big(\iota_{\dot X_i}(\pounds_{\dot\ell} \lambda), \, \ddot {\tilde g}_\mu \, (\iota_{\ddot g_\mu}\pounds_{\ddot X_i} \ddot\ell\,)\big) \ .
    \end{aligned}
    \end{equation}
The term involving $\red H{}_1$ here vanishes since $\lambda \in \det \ccG_1^\star$ and hence $\dot {\tilde g}_\mu \wedge \lambda = 0$ for all $\mu>n$. Since $\Omega^\lambda$ is $D_1$-invariant, it follows that $\iota_{\dot X_i} (\pounds_{\dot\ell} \lambda) = 0$ and so $\nabla^{\sfS_\ccG} \sfs = 0$ for $\sfs = \ell\cdot u$, because $R$ is a geometric T-duality relation. Thus the second term in \eqref{eqn:canonicalDGO} vanishes.

The third term is
\begin{equation}
\begin{aligned}
    & -(\cH \wedge \omega) \otimes \sfs \\[4pt]
    & \hspace{1cm} = -\frac 1{3!}\,\sum_{i,j,k}\,\big((\red H{}_1)_{ijk}\, \pi_1^*\alpha_i \wedge \pi_1^*\alpha_j \wedge \pi_1^*\alpha_k \wedge \pi_1^*\omega,\, (\red H{}_2)_{ijk}\, \pi_2^*\alpha_i \wedge \pi_2^*\alpha_j \wedge \pi_2^*\alpha_k \wedge \pi_2^*\omega \big) \otimes \widehat\sfs\\[4pt]
    & \hspace{4cm}=-(\dot h \wedge \pi_1^*\omega\otimes\iota_{\dot\ell}\lambda, \, \ddot h \wedge \pi_2^*\omega\otimes\ddot\ell\,) \ .
\end{aligned}
\end{equation}

The fourth term is calculated as follows. Projecting with $\rmp_2$ we get
\begin{align}
\begin{split}
    \rmp_2\overline{W}(\omega \otimes \sfs) &= \frac12\,\rmp_2\,\sum_{i,j,\mu} \, \ip{\reallywidehat{W(X_i,X_j)}\,,\, (\dot g_\mu, \ddot g_\mu)}_{R_\cB} \ (\pi_1^*\alpha_i,\pi_2^*\alpha_i)\wedge (\pi_1^*\alpha_j, \pi_2^*\alpha_j) \wedge( \pi_1^*\omega,\pi_2^*\omega) \\[-10pt]
    &\hspace{8cm}\otimes (\dot {\tilde g}_\mu, \ddot{\tilde g}_\mu) \cdot (\dot\ell, \ddot\ell\,) \cdot (\lambda,1)\\[4pt]
    &=\frac 12 \, \sum_{i,j,\mu}\, \ip{\llbracket \ddot X_i,\ddot X_j\rrbracket_2\,,\,\ddot g_\mu}_2 \ \pi_2^*\alpha_i\wedge \pi_2^*\alpha_j \wedge \pi_2^*\omega \otimes \ddot{\tilde g}_\mu \cdot \ddot\ell\\[4pt]
    &= \frac 12 \, \sum_{\mu \leqslant n,i,j}\, (\red H{}_2)_{ij\mu}\, \pi_2^*\alpha_i \wedge \pi_2^*\alpha_j \wedge \pi_2^*\omega  \otimes \ddot{\tilde g}_\mu \wedge \ddot\ell \\
    & \qquad \, + \frac 12 \, \sum_{\mu >n,i,j}\, \ddot g_\mu([\ddot X_i,\ddot X_j]) \, \pi_2^*\alpha_i \wedge \pi_2^*\alpha_j \wedge \pi_2^*\omega \otimes \iota_{\ddot{\tilde g}_\mu} \ddot\ell \\[4pt]
    &= \frac 12 \, \sum_{\mu \leqslant n,i,j}\, (\red H{}_2)_{ij\mu}\, \pi_2^*\alpha_i \wedge \pi_2^*\alpha_j \wedge \pi_2^*\omega  \otimes \ddot{\tilde g}_\mu \wedge \ddot\ell \\
    & \qquad \, + \frac 12 \, \sum_{i,j}\, \pi_2^*\alpha_i \wedge \pi_2^*\alpha_j \wedge \pi_2^*\omega \otimes \ddot\ell([\ddot X_i, \ddot X_j])\\[4pt]
    &= \frac 12 \, \sum_{\mu \leqslant n,i,j}\, (\red H{}_2)_{ij\mu} \, \pi_2^*\alpha_i \wedge \pi_2^*\alpha_j \wedge \pi_2^*\omega  \otimes \ddot{\tilde g}_\mu \wedge \ddot\ell \\
    & \qquad \, - \frac 12 \, \sum_{i,j}\, \pi_2^*\alpha_i \wedge \pi_2^*\alpha_j \wedge \pi_2^*\omega \otimes \de\ddot\ell(\ddot X_i,\ddot X_j) \ .
\end{split}
\end{align}

Similarly, projecting with $\rmp_1$ yields
\begin{align}
    \rmp_1 \overline{W}(\omega \otimes \sfs) 
    &= \frac 12 \, \sum_{i,j,\mu}\, \ip{\llbracket \dot X_i, \dot X_j \rrbracket_1\,,\, \dot g_\mu}_1 \ \pi_1^*\alpha_i \wedge \pi_1^*\alpha_j \wedge \pi_1^*\omega \otimes \dot{\tilde g}_\mu \cdot \iota_{\dot\ell} \lambda \\[4pt]
    &= \frac 12 \, \sum_{\mu>n,i,j}\, (\red H{}_1)_{ij\mu}\, \pi_1^*\alpha_i \wedge \pi_1^*\alpha_j \wedge \pi_1^*\omega \otimes \dot{\tilde g}_\mu \wedge \iota_{\dot\ell} \lambda \\
    & \qquad \, + \frac 12 \, \sum_{\mu \leqslant n,i,j}\, \dot g_\mu([\dot X_i,\dot X_j])\, \pi_1^*\alpha_i \wedge \pi_1^*\alpha_j \wedge \pi_1^*\omega \otimes \dot{\tilde g}_\mu \cdot \iota_{\dot\ell}\lambda \\[4pt]
    &= \frac 12 \, \sum_{\mu>n,i,j}\, (\red H{}_1)_{ij\mu}\, \pi_1^*\alpha_i \wedge \pi_1^*\alpha_j \wedge \pi_1^*\omega \otimes \dot{\tilde g}_\mu \wedge \iota_{\dot\ell} \lambda \\
    & \qquad \, - \frac 12 \, \sum_{i,j}\, \pi_1^*\alpha_i \wedge \pi_1^*\alpha_j \wedge \pi_1^*\omega \otimes \de(\iota_{\dot\ell}\lambda)(\dot X_i,\dot X_j) \ . 
\end{align}

It follows that the third, fourth and fifth terms combine to give
\begin{equation}
\begin{aligned}
    &-(\cH \wedge \omega) \otimes \sfs + (-1)^{|\omega|+1}\, \overline{W}(\omega \otimes \sfs) +\tfrac 14\, (-1)^{|\omega|+1} \, \omega \otimes C\cdot\sfs  \\[4pt]
    & \hspace{2cm} = \big(-\red H{}_1\wedge \pi_1^*\omega \otimes \iota_{\dot\ell}\lambda +(-1)^{|\omega|}\, \pi_1^*\omega \otimes \de\, \iota_{\dot\ell}\lambda\,,\, -\red H{}_2 \wedge \pi_2^*\omega \otimes \ddot\ell + (-1)^{|\omega|}\,\pi_2^*\omega \otimes \de \ddot\ell\,\big) \ .
\end{aligned}
\end{equation}
Hence
\begin{equation}
\begin{aligned}
    \sfD(\omega \otimes \sfs) &= \big(\de\, \pi_1^*\omega \otimes \iota_{\dot\ell}\lambda + (-1)^{|\omega|}\, \pi_1^*\omega \otimes \de\, \iota_{\dot\ell}\lambda - \red H{}_1 \wedge \pi_1^*\omega \otimes \iota_{\dot\ell}\lambda\ ,\\ 
    &\hspace{6cm}\de\, \pi_2^*\omega \otimes \ddot\ell + (-1)^{|\omega|}\, \pi_2^*\omega\otimes \de \ddot\ell -\red H{}_2 \wedge \pi_2^*\omega \otimes \ddot\ell\,\big)\\[4pt]
    &= (\de_{\red H{}_1} \times \de_{\red H{}_2})( \widehat\omega \otimes \widehat\sfs\,) \ .
\end{aligned}
\end{equation}
Since $\frk{k}_{\tau_1}$ spans $\rmp_1(L^*)$ pointwise, this can be extended to the whole of $L^*$, which by the Leibniz rule can then be extended to the whole of $\sfS_\ccG$ and hence $\sfS_R$.
\end{proof}

\begin{example}[\textbf{Torus Bundles}]
    Continuing from Example \ref{ex:FourierMukai}, the fact that $\sfS_R$ is a $(\de_{\red H{}_1}, \de_{\red H{}_2})$-DGO relation is equivalent to the fact that \smash{$\de_{\red H{}_2} \varrho(\omega) = (-1)^k\,\varrho(\de_{\red H{}_1} \omega)$} for all \smash{$\omega\in\sfOmega_{\sfT^k}^\bullet(\cQ_1)$}. The existence of $\lambda$ satisfying the assumptions of Theorem \ref{prop:canonicalDGOforR} is guaranteed here: if $\theta_1$ is a connection on the principal $\sfT^k$-bundle $\cQ_1 \to \cB$ and $(t_a)$ is a basis of the Lie algebra $\frk{t}^k$ of $\sfT^k$, then $\theta_1 = \sum_a\, \theta_1^a \otimes t_a$ where $\theta_1^a\in\sfOmega^1_{\sfT^k}(\cQ_1)$ for $a=1,\dots,k$. Then $\lambda = \theta_1^1 \wedge \cdots \wedge \theta_1^k$ is $D_1$-invariant.
\end{example}

\begin{example}[\textbf{Poisson-Lie T-Duality}]
    Let $\frg$ be a Lie algebra with a non-degenerate pairing, $\sfG$ a connected Lie group integrating $\frg$, and let $\sfH_1, \sfH_2\subset \sfG$ be subgroups whose Lie subalgebras $\frh_1, \frh_2\subset \frg$ are Lagrangian. There is then the T-duality relation $R \colon E_1 \rel E_2$ between exact Courant algebroids $E_1$ over $\sfG/\sfH_1$ and $E_2$ over $\sfG/ \sfH_2$~\cite{Vysoky2020hitchiker,deFraja2025Ricci}. If $R\colon E_1 \rel E_2$ is a generalised isometry, then the existence of an invariant volume form $\lambda$ on $\sfG/ \sfH_1$ is equivalent to $\frh_1$ being a unimodular Lie algebra. We may thus construct the $R$-Clifford relation $\sfS_{R}$. The unimodularity condition coincides with the existence of an invariant divergence on $\sfG/\sfH_1$ \cite[Example 2.5]{Severa:2018pag}.
\end{example}

\medskip

\subsection{T-Duality Relations and Type~II Supergravity}~\\[5pt]
\label{sub:TDualitySUGRA}
The bosonic field content of Type II supergravity on a manifold $M$ is given by a quadruple $(g,H, \phi, F)$ of a Riemannian metric $g$ on $M$, a closed three-form $H$ (the Kalb-Ramond flux), a smooth function $\phi\in C^\infty(M)$ (the dilaton), and an inhomogenous differential form $F\in \sfOmega^\bullet(M)$  of even or odd degree which is closed under the twisted differential $\de_H$ (the collection of Ramond-Ramond fluxes). The form $F$ is also required to satisfy an additional self-duality condition, which we discuss below. 

The equations of motion for Type II supergravity can then be formulated on the twisted standard Courant algebroid $(\IT M, H)$ equipped with the generalised metric $V_g^+$ corresponding to $g$, the divergence $\div_\phi$ corresponding to $\phi$, and the spinor bundle $\sfOmega^\bullet(M)$. They are given by
\begin{align}
    \gric_{g,\div_\phi}(u_+, v_-) &= \frac{\ii}{8\, \nu}\, (\!( u_+ \cdot F, v_- \cdot F)\!) \ , \quad \forall\, u_+ \in V_g^+ \ , \ v_-\in V_g^- \ , \label{eqn:typeII1}\\[4pt]
    \mathrm{GR}_{g,\div_\phi}&=0\ . \label{eqn:typeII2}
\end{align}
Here $(\!(\, \cdot \, , \, \cdot \,)\!)$ is the pairing on spinors discussed in Subsection~\ref{subsub:spinor_pairing} below, $\gric_{g,\div_\phi}$ is the generalised Ricci tensor associated to the pair $(V^+_g, \div_\phi)$ as in \cite[Definition 3.1]{Severa:2018pag}, and $\mathrm{GR}_{g,\div_\phi}$ is the generalised scalar curvature as in \cite[Definition 4.9]{Streets:2024rfo}.
Equations \eqref{eqn:typeII1} and \eqref{eqn:typeII2} are the equations of motion of the Type IIA theory if $\nu=-1$ and $F$ is of even degree, whereas when $\nu = \ii$ and $F$ is of odd degree they are the equations of motion of the Type IIB theory.

\medskip

\subsubsection{Pairing on Spinors}~\\[5pt]
\label{subsub:spinor_pairing}
Recalling Remark~\ref{rem:Mukai}, we define a pairing on the space $$\widetilde\cS= \sfS \otimes |\det \sfS^*|^{\frac12}$$ by
\begin{equation}
\begin{aligned}
    (\!(\,\cdot \, , \, \cdot \,)\!) &\colon \widetilde \cS \otimes \widetilde \cS \longrightarrow \IC\ , \qquad (\!(\sfs \otimes \mu , \sfs' \otimes \mu)\!) = \ip{(\vartheta\, \sfs \wedge \sfs')^{\rm top}, \mu^2} \ ,
\end{aligned}
\end{equation}
where here $\ip{\,\cdot\, , \,\cdot\,}$ denotes the contraction between $\det \sfS$ and $\det \sfS^*$, and $\vartheta$ is the operator defined by $$\vartheta\, \sfs = \ii^{\,|\sfs|}\,\sfs \ . $$

In the setting of Subsection~\ref{ssec:twisteddifferential} with the splitting \eqref{eqn:splittingforRB} of $R$, recall  that $\sfS_R = \midwedge^\bullet T^*\cB \otimes \sfS_\ccG$, where $\sfS_\ccG = \midwedge^\bullet L^* \cdot( \lambda, 1)$. Let
\begin{align}\label{eqn:anothercanonicalSR}
    \widetilde{\cS}_R  = \midwedge^\bullet T^*\cB \otimes |\det T\cB\,|^{\frac12} \otimes \midwedge^\bullet L^* \cdot (\lambda, 1) \otimes \left\vert\big(\det L^* \cdot (\lambda, 1)\big)^* \cdot (\lambda',1)\right\vert^{\frac12}\ ,
\end{align}
regarded as a subbundle $$\widetilde \cS_R \ \subset \ \big[\big((\midwedge^\bullet T^* \cQ_1) \otimes |\det T \cQ_1|^{\frac12}\big) \times \big((\midwedge^\bullet T^* \cQ_2) \otimes |\det T \cQ_2|^{\frac12}\big)\big]\big|_{\red C} \ , $$ 
where $\lambda'\in \det T\red \cF{}_1$ satisfies $\iota_{\lambda'} \, \lambda = 1$.
Since $L^*$ is chosen so that $\red\rho{}_2(\rmp_2(L^*)) = \set{0}$, it follows that $(\det L^* \cdot (\lambda ,1 ))^*$ is generated by $(1, {}^R \lambda)$ for ${}^R \lambda \in \det T\red \cF{}_2$ satisfying $(\lambda, {}^R\lambda)\in \Cl(R)$. We then write $$\widetilde \sfs_1 \approx_R \widetilde \sfs_2 \qquad \text{if} \quad (\widetilde \sfs_1 , \widetilde \sfs_2)  \in  \widetilde \cS_R \ . $$

\begin{proposition} \label{prop:pairingspinor}
    Let $R\colon\IT \cQ_1 \rel \IT \cQ_2$ be a T-duality relation with $R$-Clifford relation $\sfS_R$ given by Equation \eqref{eqn:Tdualityrelspinrep} and $\widetilde \cS_R$ as in Equation \eqref{eqn:anothercanonicalSR}. 
    If $\widetilde \sfs_1 \approx_R \widetilde \sfs_2$ and $\widetilde\sfs_1^{\,\prime} \approx_{R} \widetilde\sfs_2^{\,\prime}$, we write  $$\widetilde\sfs_i=\sfs_i\otimes\mu_i \qquad \text{and} \qquad \widetilde\sfs_i^{\,\prime}=\sfs_i'\otimes\mu_i$$ for $i=1,2$, where $(\sfs_1,\sfs_2),(\sfs_1',\sfs_2')\in\sfS_R$ and $\mu_i\in|\det T\cQ_i|^{\frac12}$. Then
    \begin{align}
        (\!(\widetilde\sfs_1, \widetilde\sfs_1^{\,\prime})\!)_1 = \ii^{\,n}\,(-1)^{n^2 + n\, |\sfs_1'|}\, (\!(\widetilde\sfs_2, \widetilde\sfs_2^{\,\prime})\!)_2
    \end{align}
    in $\IC$, where $n=\dim_\IR\red\cF{}_1$ and $(\!(\, \cdot \, , \, \cdot \,)\!)_i$ for $i=1,2$ are the pairings on spinors on each related Courant algebroid.
\end{proposition}

\begin{proof}
    Let $c\in \red C$. We can write $$\sfs_1 = \pi_1^*\omega \otimes(\ell \cdot \lambda) \quad ,  \quad \sfs_2 = \pi_2^*\omega \otimes {}^R \ell \quad , \quad \sfs_1' = \pi_1^*\omega' \otimes (\ell' \cdot \lambda) \quad , \quad \sfs_2' = \pi_2^*\omega' \otimes {}^R\ell' \ , $$ where $\omega,\omega'\in\midwedge^\bullet T^*\cB$ and $\ell,\ell',{}^R\ell,{}^R\ell'\in\midwedge^\bullet L^*$. We can also write $\mu_2 \in |\det T \cQ_2|^{\frac12}$ as $$\mu_2 = \ddot \mu_{\cB} \otimes \sqrt{{}^R \lambda}  \ , $$ for $\mu_\cB \in |\det T \cB\,|^{\frac12}$ and $\sqrt{{}^R \lambda} \in \rmp_2\big(|\det L^*\cdot (\lambda, 1)|^{\frac12}\big)$ such that $\big(\sqrt{{}^R \lambda}\,\big)^2 = {}^R\lambda$. Similarly, we write \smash{$\mu_1\in|\det T\cQ_1|^{\frac12}$} as $$\mu_1=\dot\mu_\cB\otimes\sqrt{\lambda'} \ . $$ Assume that $|\omega| + |\omega'| = \dim_\IR \cB$ and $|\ell|+|\ell'| = |{}^R \ell|+|{}^R\ell'| = \dim_\IR \red\cF{}_1 =: n$, so that the pairings are non-zero. 
    
    Then
    \begin{equation}
    \begin{aligned}
        (\!( \widetilde\sfs_2,\widetilde\sfs_2^{\,\prime} )\!)_2 &= (\!( \sfs_2 \otimes \ddot\mu_\cB\otimes\sqrt{{}^R\lambda}\,,\, \sfs_2' \otimes \ddot\mu_\cB\otimes\sqrt{{}^R\lambda} \,)\!)_2 \\[4pt]
        &= \ip{\ii^{\,|\omega| + |\ell|} \, (\pi_2^*\omega \otimes {}^R \ell) \wedge (\pi_2^* \omega' \otimes {}^R \ell') \,,\, \ddot\mu_\cB^2\otimes{}^R\lambda}\\[4pt]
        &=\ii^{\,|\omega| + |\ell|} \, (-1)^{|\ell|\,|\omega'|}\,\ip{\omega \wedge \omega',  \mu_\cB^2} \, \ip{{}^R \ell \wedge {}^R \ell', {}^R \lambda} \ .
    \end{aligned}
    \end{equation}
    Similarly
    \begin{equation}
        \begin{aligned}
      (\!( \widetilde\sfs_1,\widetilde\sfs_1^{\,\prime} )\!)_1 &=  (\!( \sfs_1 \otimes \dot \mu_\cB \otimes\sqrt{\lambda'}\,,\, \sfs_1' \otimes \dot \mu_\cB \otimes \sqrt{\lambda'}\, )\!)_1 \\[4pt] & = \ii^{\,|\omega| + n - |\ell|}\,\ip{(\pi_1^*\omega \otimes (\ell \cdot \lambda)) \wedge (\pi_1^*\omega' \otimes (\ell' \cdot \lambda)) \,,\, \dot \mu_\cB^2 \otimes \lambda'\,} \\[4pt]
            &= \ii^{\,|\omega| + n - |\ell|} \, (-1)^{(n-|\ell|)\,|\omega'|} \, \ip{(\pi_1^*\omega \wedge \pi_1^*\omega'\,) \otimes ((\ell \cdot \lambda) \wedge (\ell' \cdot \lambda))\,,\, \dot \mu_\cB^2 \otimes \lambda '\,}\\[4pt]
            &= \ii^{\,|\omega| + n - |\ell|} \, (-1)^{(n-|\ell|)\,|\omega'|} \, (-1)^{|\ell|\,(n-|\ell|)}\,\ip{\omega \wedge \omega' , \mu_\cB^2}\, \ip{(\ell \cdot \ell' \cdot \lambda)\,\lambda , \lambda '\,}\ .
        \end{aligned}
    \end{equation}
    
    By construction $$\ell \cdot \ell' \cdot \lambda \ \sim_R \ (-1)^{\frac{n\,(n+1)}2} \, \iota_{{}^R \lambda}\big({}^R \ell \wedge {}^R \ell'\,\big)$$ in $\Cl(\ccG_c)$. Since both sides of this $R$-relation are real numbers, and $(r_1,r_2) \in \Cl(\ccG_c)$ for real numbers $r_1,r_2\in \IR$  if and only if $r_1=r_2$, it follows that
    \begin{align}
        (\!(\widetilde \sfs_1, \widetilde\sfs_1^{\,\prime} )\!)_1 = a \, (\!(\widetilde \sfs_2, \widetilde\sfs_2^{\,\prime} )\!)_2
    \end{align}
    where
    \begin{equation}
    \begin{aligned}
        a &= \ii^{\,|\omega|+n-|\ell| - |\omega| - |\ell|} \, (-1)^{ n\,|\omega'| - |\ell|\,|\omega'| + n\,|\ell| - |\ell|^2 - |\ell|\,|\omega'|}\\[4pt]
        &=\ii^{\,n} \, (-1)^{|\ell| + n\,(|\sfs_1'| + |\ell| -n) + n\, |\ell| - |\ell|^2} =\ii^{\,n}\,(-1)^{n^2 + n\,|\sfs_1'|} \ ,
    \end{aligned}
    \end{equation}
    and the result follows.
\end{proof}

\medskip

\subsubsection{Self-Dual Spinors and T-Duality}~\\[5pt]
Let $R:(\IT\cQ_1,V_1^+)\rel(\IT\cQ_2,V_2^+)$ be a geometric T-duality relation with $R$-Clifford relation $\sfS_R$ of parity $j = n\mod 2$, where $n=\dim_\IR\red\cF{}_1=\dim_\IR\red\cF{}_2$. Let $d=\dim_\IR\cQ_1=\dim_\IR\cQ_2$.
Choose a basis $\set{r_1,\ldots,r_d}$ for $R^+$ such that $\set{e_1,\dots,e_d}:=\set{\rmp_1(r_1),\dots,\rmp_1(r_d)}$ is an orthonormal basis for $V_1^+$. Then $\set{f_1,\dots,f_d} := \set{\rmp_2(r_1),\dots, \rmp_2(r_d)}$ is an orthonormal basis for $V_2^+$. Define $$\cR_{V_1^+} = 2^{\frac d2} \ e_1 \cdot \,\cdots\,\cdot e_d \qquad  \text{and} \qquad \cR_{V_2^+} = 2^{\frac d2} \ f_1\cdot \,\cdots\, \cdot f_d \ .$$ Then
\begin{align}
    \big(\cR_{V_1^+}\cdot u_1\,,\, (-1)^{n\,d}\ \cR_{V_2^+} \cdot u_2\big)  \ \in \ \sfS_R \ ,
\end{align}
for all $(u_1,u_2) \in \sfS_R$.

Thus if $u_1$ satisfies the self-duality equation $$\cR_{V_1^+} \cdot u_1 = u_1 \ , $$ then
\begin{align}
    \big(\cR_{V_1^+}\cdot u_1\,,\, (-1)^{n\,d}\ \cR_{V_2^+} \cdot u_2\big) = \big(u_1\,,\, (-1)^{n\,d}\ \cR_{V_2^+}\cdot u_2\big) \ \in \sfS_R
    \end{align}
    implies the (anti-)self-duality equation
    \begin{align}
    (-1)^{n\,d}\ \cR_{V_2^+}\cdot u_2 = u_2 \ ,
\end{align}
by Corollary~\ref{cor:URinjectivesmooth}.
In Type II supergravity, $d=10$ and hence $u_2$ also satisfies the self-duality equation.

In the context of generalised geometry, we replace the space of spinors $\sfOmega^\bullet(\cQ_i)$ with the space $$\widetilde{\sfOmega}^\bullet(\cQ_i) = \sfOmega^\bullet(\cQ_i) \otimes \sfGamma\big(|\det T\cQ_i|^{\frac12}\big)$$ in Equations~\eqref{eqn:typeII1} and~\eqref{eqn:typeII2} with labels $i=A,B$, giving the \emph{modified Type II supergravity equations}, see~\cite[Section 7.1]{Severa:2018pag}.

\begin{proposition}\label{prop:TdualitytypeII}
    Let $(g_A,\red H{}_A,\phi_A, F_A)$ and $(g_B, \red H{}_B, \phi_B, F_B)$ be Type~II bosonic fields on T-dual principal circle bundles $\cQ_A\to\cB$ and $\cQ_B \to \cB$ over the same base $\cB$, respectively.  Let
    \begin{equation}
\begin{tikzcd}
 R  \colon (\IT\cQ_A, \red H{}_A, V^+_{g_A}, \div_{\phi_A}) \ar[r,dashed] &  (\IT \cQ_B, \red H{}_B, V^+_{g_B}, \div_{\phi_B})
 \end{tikzcd}
\end{equation} 
    be a geometric T-duality relation with $R$-Clifford relation $\sfS_R$ such that $$F_A \approx_R F_B \ . $$
    Then $(g_A,\red H{}_A,\phi_A, F_A)$ satisfies the equations of motion~\eqref{eqn:typeII1}--\eqref{eqn:typeII2} for $\nu = -1$ if and only if $(g_B,\red H{}_B,\phi_B, F_B)$ satisfies the equations of motion~\eqref{eqn:typeII1}--\eqref{eqn:typeII2} for $\nu = \ii$.
\end{proposition}

\begin{proof}
    In the Type IIA theory, $F_A$ is of even degree and hence $|u_A{}_\pm \cdot F_A| = 1 \mod 2$ for all $u_A{}_\pm\in V_{g_A}^\pm$. Thus by Proposition~\ref{prop:pairingspinor} with $n=1$ we get
    \begin{align}
        (\!(u_A{}_+ \cdot F_A, v_A{}_- \cdot F_A)\!) = \ii \ (\!(u_B{}_+ \cdot F_B, v_B{}_- \cdot F_B)\!) \ .
    \end{align}
    Since $\sfS_R$ is degree shifting by $\pm 1$, $F_B$ is of odd degree.
    
    We know from \cite[Theorem 4.14]{deFraja2025Ricci} that the generalised Ricci tensors and generalised scalar curvatures are preserved by the T-duality relation $R$. 
    Thus if $(g_A,\red H{}_A,\phi_A, F_A)$ satisfies the Type IIA supergravity equations, then
    \begin{equation}
    \begin{aligned}
        \gric_{g_B,\div_{\phi_B}}(u_B{}_+,v_B{}_-) &= \gric_{g_A,\div_{\phi_A}}(u_A{}_+,v_A{}_-) \\[4pt]
        &= -\frac{\ii}{8}\,(\!(u_A{}_+ \cdot F_A, v_A{}_- \cdot F_A)\!) \\[4pt]
        &= \frac{1}{8}\, (\!(u_B{}_+ \cdot F_B, v_B{}_- \cdot F_B)\!) \ ,
    \end{aligned}
    \end{equation}
    and $\mathrm{GR}_{g_B,\div_{\phi_B}} = \mathrm{GR}_{g_A,\div_{\phi_A}}=0$, hence $(g_B, \red H{}_B, \phi_B, F_B)$ is a solution of the Type IIB supergravity equations.
    
    Conversely, if the Type IIB equations of motion are satisfied by $(g_B, \red H{}_B, \phi_B, F_B)$, then since $|u_B{}_\pm \cdot F_B| = 0 \mod 2$, we find
    \begin{equation}
    \begin{aligned}
        \gric_{g_A,\div_{\phi_A}}(u_A{}_+,v_A{}_-) &= \gric_{g_B,\div_{\phi_B}}(u_B{}_+,v_B{}_-) \\[4pt]
        &= \frac{1}{8}\,(\!(u_B{}_+ \cdot F_B, v_B{}_- \cdot F_B)\!) \\[4pt]
        &= \frac{1}{8\,\ii}\, (\!(u_A{}_+ \cdot F_A, v_A{}_- \cdot F_A)\!) \ ,
    \end{aligned}
    \end{equation}
    as required. 
    
    We have already seen that the self-duality equation is preserved. 
    Since $\sfS_R$ is a $(\de_{\red H{}_A}, \de_{\red H{}_B})$-DGO relation, it follows that $\de_{\red H{}_A} F_A = 0$ if and only if $\de_{\red H{}_B} F_B = 0$.
\end{proof}

\medskip

\section{Generalised Complex Relations}\label{sect:section6}

In this section we consider how generalised complex structures behave under Courant algebroid relations. We will see that a generalised complex relation has an induced grading, and we will write an explicit formula for the type change of generalised complex structures under geometric T-duality. 

\medskip

\subsection{Relating Generalised Complex Structures}~\\[5pt]
Suppose that $E_1$ and $E_2$ are Courant algebroids respectively endowed with generalised almost complex structures $\cJ_1$ and $\cJ_2.$ 
The product Courant algebroid $E_1 \times E_2$ is endowed with the generalised almost complex structure $\cJ_1 \times \cJ_2.$
As before, when working with generalised complex structures one introduces the complexified Courant algebroids in order to discuss the $\pm\, \ii$-eigenbundles, but we shall often suppress the complexification in the notation. All previous results concerning Dirac relations and spinor relations hold in the complex case too.

\begin{definition}\label{def:generalisedcomplexrel}
Let $R \colon E_1 \dashrightarrow E_2$ be a Courant algebroid relation supported on $C \subseteq M_1 \times M_2$. Let $\cJ_1$ and $\cJ_2$ be generalised almost complex structures on $E_1$ and $E_2$ respectively. 
Then $R$ is a \emph{generalised almost complex relation} if $$(\cJ_1\times\cJ_2)(R) = R \ . $$ If $\cJ_1$ and $\cJ_2$ are generalised complex structures, then $R$ is a \emph{generalised complex relation}.
\end{definition}

A first instance of this definition appeared in~\cite{Bailey:2023} with $C$ being the graph of a smooth map. 

\begin{example} \label{ex:gencomplexrelgraph}
Let $R= \gr(\Phi) \colon E_1 \rightarrowtail E_2$ be a classical Courant algebroid morphism supported on $\gr(\varphi) \subseteq M_1 \times M_2,$ for some vector bundle morphism $\Phi \colon E_1 \to E_2$ covering a smooth map $\varphi \colon M_1 \to M_2.$ Suppose that $E_1$ and $E_2$ are endowed with the generalised complex structures $\cJ_1$ and $\cJ_2,$ respectively. An element $(e_1, e_2) \in E_1 \times E_2$ is also in $R$ if and only if 
$(e_1, e_2 ) = (e_1, \Phi(e_1)) \in R$. Thus by Definition \ref{def:generalisedcomplexrel}
\begin{align}
    (\cJ_1\times\cJ_2)\big(e_1 , \Phi(e_1)\big) = \big(\cJ_1 (e_1), \cJ_2(\Phi(e_1))\big)
\end{align}
is again in $R$ if and only if
\begin{align}
    \cJ_2 \circ \Phi = \Phi \circ \cJ_1 \ ,
\end{align}
as expected for a well-defined notion of morphism of generalised complex structures. This condition is equivalent to 
\begin{align}
\Phi(L_1) \ \subseteq \ L_2 \ ,
\end{align}
where $L_1$ and $L_2$ are the $+\ii$-eigenbundles of $\cJ_1$ and $\cJ_2,$ respectively. If $\Phi$ is an isomorphism, this last condition becomes $\Phi(L_1)=L_2.$
\end{example}

\begin{example}\label{eg:symplectictypeiso}
    In the setting of Example~\ref{ex:gencomplexrelgraph}, let $\cJ_1$ be given by a symplectic structure $\omega$ on $M_1$ (cf. Example~\ref{ex:symplecticmanifolds}). Then $$L_1 = \set{X - \ii\, \iota_X \omega\ \vert\ X \in TM_1} \ . $$ With $\Phi = \e^B \circ \big( \varphi_*\oplus(\varphi^{-1})^*\big)$, it follows that
    \begin{equation}
    \begin{aligned}
        \Phi(L_1) &= \set{\varphi_* X + B( \varphi_* X) - (\varphi^{-1})^*(\ii\, \iota_X \omega)\ \vert\ X \in TM_1}\\[4pt]
        &= \set{Y + B(Y) - \ii\,\big(\iota_Y\, (\varphi^{-1})^* \omega) \ \vert\ Y\in TM_2} \ .
    \end{aligned}
    \end{equation}
    Thus $\cJ_2$ is a $B$-twist of the symplectic structure $(\varphi^{-1})^*\omega$ on $M_2$.
\end{example}

\begin{remark}\label{rem:kernelspreserved}
    If $R \colon (E_1, \cJ_1) \rel (E_2, \cJ_2)$ is a generalised complex relation and $(k,0) \in {\rm K}_R$, then $(\cJ_1(k),0) \in {\rm K}_R$ and hence $\cJ_1$ preserves $\ker(R)$. Similarly $\cJ_2$ preserves ${\rm coker}(R)$.
\end{remark}

\begin{definition} \label{def:splitGCR}
    A Courant algebroid relation $R \colon E_1 \rel E_2$ is a \emph{generalised almost complex relation} if at each point $c\in C$ there is a splitting
    \begin{align}\label{eqn:splitGCR}
        R^\IC_c = R^{+\ii}_c \oplus R^{-\ii}_c \ , 
    \end{align}
    where $$R^{+\ii}_c = R^\IC_c\cap (L_1 \times L_2)_c \qquad \text{and} \qquad R^{-\ii}_c = R^\IC_c \cap (\overline{L}_1 \times \overline{L}_2)_c \ . $$ Again we drop the adjective \enquote{almost} when $\cJ_1$ and $\cJ_2$ are generalised complex structures.
\end{definition}

\begin{remark}\label{rmk:ranksofRi}
    Since $R$ is a real vector bundle (so that $\overline{R^\IC} = R^\IC$), it follows that $\overline{R^{\pm\ii}} \subset R^{\mp\ii}$ and hence $R^{\mp\ii} \subset \overline{R^{\pm\ii}}$. Thus $\overline{R^{\pm\ii}} = R^{\mp\ii}$ and  hence $$\rk_\IC(R^{\pm\ii}) = \tfrac 12\,\rk_\IC(R^\IC) \ . $$
    In particular, the rank of a generalised almost complex relation $R$ is necessarily even.
\end{remark}

Similarly to \cite[Proposition 4.15]{DeFraja:2023fhe}, which establishes equivalence of two definitions of generalised isometries (cf. Definition~\ref{defn:generalisedisometry}), we can show

\begin{proposition} \label{prop:equivGCRdefs}
     Definitions \ref{def:generalisedcomplexrel} and \ref{def:splitGCR} are equivalent.
\end{proposition}

\begin{proof}
    Let $\cJ=\cJ_1\times\cJ_2$. If $\cJ(R^\IC_c) = R^\IC_c$ and $r\in R^\IC_c$, then $r^{\pm\ii}:=\frac12\,(r\mp\ii\,\cJ(r))\in R_c^{\pm\ii}$ with $r=r^{+\ii}+r^{-\ii}$, and hence $R_c^\IC \subset R^{+\ii}_c \oplus R^{-\ii}_c$. The converse easily follows. 
\end{proof}

From Proposition \ref{prop:integrabilityrelation} we immediately get

\begin{proposition}
    Let $R\colon(E_1, \cJ_1) \rel (E_2, \cJ_2)$ be a generalised almost complex relation, with $\rmp_2(R)\supset L_2$ (resp. $\rmp_1(R)\supset L_1$). If $\cJ_1$ (resp. $\cJ_2$) is a generalised complex structure (i.e. $L_1$ (resp. $L_2$) is integrable), then $\cJ_2$ (resp. $\cJ_1$) is a generalised complex structure, and hence $R$ is a generalised complex relation.
\end{proposition}

\begin{remark}\label{rem:3rddefinition}
    If $R \colon (E_1, L_1) \rel (E_2, L_2)$ is a generalised almost complex relation, then one automatically has $L_1 \sim_R L_2$ and $\overline{L}_1 \sim_R \overline L_2$. The converse is also true, because $\overline{R^\IC\rvert_{L_1}} = R^\IC\rvert_{\overline L_1}$ and hence
    \begin{align}
        R^\IC\big\rvert_{\overline L_1} = R^\IC\cap (\overline L_1 \times \overline{L}_2) + {\rm CK}_R
    \end{align}
    since $R$ is real, and noting that ${\rm CK}_R$ is also real. Thus there is a third equivalent definition: $R$ is a generalised almost complex relation if and only if $L_1 \sim_R L_2$ for the complex almost Dirac structures generated by $\cJ_1$ and $\cJ_2$. Therefore, by Proposition \ref{prop:LreliffUrelsmooth} we obtain
\end{remark}

\begin{proposition}\label{prop:UreliffJrel}
    Let $R\colon E_1 \rel E_2$ be a Courant algebroid relation, and let $\cJ_i$ be  generalised almost complex structures on $E_i$ with associated pure spinor line bundles $\Omega_i$ for $i=1,2$. If  $\Omega_1\sim_{R} \Omega_2$, then $R$ is a generalised almost complex relation.
\end{proposition}

Let $(E_1, L_1)$ and $ (E_2, L_2)$ be Courant algebroids equipped with complex Lagrangian structures, and let $R \colon E_1 \rel E_2$ be  a Courant algebroid relation. Recalling the kernel and cokernel subbundles ${\rm K}_R$ and ${\rm CK}_R$ of $R$ from Definition~\ref{def:kerandcoker}, we define the subbundle
\begin{align}
    R^{+\ii}_c{}^\times = \big[\big(R^\IC_c\cap (L_1 \times L_2)_c + {\rm K}^\IC_{R}{}_c + {\rm CK}^\IC_{R}{}_c\big)\,\big/\,{\rm K}^\IC_{R}{}_c\big]\,\big/\,{\rm CK}^\IC_{R}{}_c
\end{align}
and similarly for $R^{-\ii}_c{}^\times$.

\begin{definition}
    A Courant algebroid relation $R \colon E_1 \rel E_2$ is a \emph{weak generalised almost complex relation} if at each point $c\in C$ there is a splitting
    \begin{align}
        R^\IC_c{}^\times = R^{+\ii}_c{}^\times \oplus R^{-\ii}_c{}^\times \ .
    \end{align}
\end{definition}

Similarly to Proposition~\ref{prop:strong=weak}, it follows that if $R$ is a generalised almost complex relation then it is necessarily a weak generalised almost complex relation.

The grading induced by a generalised complex structure from~\eqref{eqn:GCSgrading} similarly induces a grading of $R$-Clifford relations according to

\begin{proposition}
    Let $R \colon (E_1, \cJ_1) \rel (E_2, \cJ_2)$ be a generalised almost complex relation with $R$-Clifford relation $\sfS_R$. There is a grading
    \begin{align}\label{eqn:SRgrading}
        \cS_R^\IC = \cS_R^n \oplus \cS_R^{n-1} \oplus \cdots \oplus \cS^0_R
    \end{align}
    where $$\cS^k_R = \big((\Omega^k_1 \times \Omega_2^k) \otimes|\det \sfS_R^* |^{\frac{1}{2^n}}\big) \cap \cS^\IC_R \ , $$ with $\Omega_i^k$ as in~\eqref{eqn:GCSgrading} for $k\in\set{0,\ldots,n}$ and $i=1,2$.
\end{proposition}

\begin{proof}
    Since $L_1 \sim_R L_2$, it follows that $\overline{L}_1 \sim_R \overline{L}_2$. Define $\cS_R^n = \big((\Omega_1 \times \Omega_2)\otimes |\det \sfS_R^* |^{\frac{1}{2^n}}\big)\cap\cS_R^\IC$ and 
    \begin{align}
        \cS_R^{n-k} = \midwedge^k R^{-\ii} \cdot \cS_R^n \ ,
    \end{align}
    for $k=1,\dots,n$.
    Since $\cS_R$ is invariant under the Clifford action of $R$, it follows that $\cS_R^{n-k} \subset \cS_R^\IC$. The projections $\rmp_i(\cS_R^{n-k}) = \Omega_i^{n-k}$ ensure that the sum \eqref{eqn:SRgrading} is direct, and since $R$ is a Dirac relation between $\overline L_1$ and $\overline L_2$ this ensures equality.
\end{proof}

Let $$\de_{H_i} = \partial_{H_i} + \overline{\partial}_{H_i}$$ be the decompositions of the Dirac generating operators associated with  generalised complex structures $\cJ_i$. 

\begin{proposition}\label{prop:gradingforGCR}
    Let $R \colon E_1 \rel E_2$ be a generalised complex relation, with $R$-Clifford relation $\sfS_R$ of parity $j$ such that $\sfS_R$ is a $(\de_{H_1}, \de_{H_2})$-DGO relation. Then
    \begin{align}
        \big(\partial_{H_1}\times (-1)^j\, \partial_{H_2}\big)(\cS^\IC_R) \ \subset \ \cS^\IC_R  \qquad \text{and} \qquad \big(\overline{\partial}_{H_1}\times (-1)^j\, \overline \partial_{H_2}\big)(\cS^\IC_R) \ \subset \ \cS^\IC_R \ .
    \end{align}
\end{proposition}

\begin{proof}
    Let $(\alpha_1, \alpha_2) \in \cS_R^k$ for $k\in\set{0,\dots,n}$, where we suppress the determinant part from the notation. Then $r = (\de_{H_1} \alpha_1, (-1)^j\,\de_{H_2} \alpha_2) \in \cS^\IC_R$, and so there is a decomposition $r = r_0 + \cdots + r_n$ from \eqref{eqn:SRgrading}. Since  $\de_{H_i}(\sfGamma(\Omega_i^k)) \subset \sfGamma(\Omega_i^{k-1}) \oplus \sfGamma(\Omega_i^{k+1})$, it follows that $r = r_{k-1} + r_{k+1}$ and hence
    \begin{align}
        \partial_{H_1}\times (-1)^j\, \partial_{H_2} \colon \cS^k_R \longrightarrow \cS^{k-1}_R \qquad \text{and} \qquad \overline \partial_{H_1} \times (-1)^j\, \overline \partial_{H_2} \colon \cS^k_R \longrightarrow \cS^{k+1}_R \ ,
    \end{align}
    thus the result follows.
\end{proof}

\medskip

\subsection{Generalised Complex Reduction}~\\[5pt]
    Consider the reduction of a generalised almost complex structure $\cJ$ to a generalised almost complex structure $\red \cJ$ from Subsection~\ref{subsec:gencomplexreduction}.
    Then by Remark \ref{rem:kernelspreserved}, the reduction morphism $Q(K)$ of Example~\ref{eg:reductionCArel} is a generalised almost complex relation from $\cJ$ to $\red \cJ$ if and only if $\cJ(K) = K$. This can be formulated in terms of the Lagrangian subbundle $L\subset E^\IC$ corresponding to $\cJ$, where 
    \begin{align}
        Q(K)^\IC = Q(K)^{+\ii} \oplus Q(K)^{-\ii} \ \iff \ (K^\perp)^\IC = \big(L\cap (K^\perp)^\IC\big) \oplus \big(\overline{L}\cap (K^\perp)^\IC\big) \ .
    \end{align}
    The final condition is equivalent to $\cJ(K^\perp) = K^\perp$, and hence $\cJ(K)=K$.

    We shall explore relations coming from generalised complex reduction in their full generality, i.e.~when the generalised complex reduction may satisfy $\cJ(K) \neq K$.
    
    \begin{proposition}
        Consider the setting of Proposition \ref{prop:reducedcomplexdirac} for the generalised almost complex reduction of $(E, \cJ)$. Then the reduction morphism $Q(K)$ is a weak generalised almost complex relation. 
    \end{proposition}
    
    \begin{proof}
    Set $$L_K = \big(L\cap (K^\perp)^\IC + K^\IC\big)\,\big/\,K^\IC \qquad \text{and} \qquad \overline L_K = \big(\overline L \cap (K^\perp)^\IC +K^\IC\big)\,\big/\,K^\IC \ . $$ Then    
    \begin{align}
        Q(K)_{(m,\varpi(m))}^{+\ii\,\times} = \set{\big(\ell, \mathscr{J}^{-1}_{\varpi(m),m}(\ell)\big) \  \vert \ \ell \in L_K{}_m}
    \end{align}
    and similarly for the complex conjugate. Thus $Q(K)$ is a weak generalised almost complex relation if and only if $$(K^\perp)^\IC\,\big/\,K^\IC = L_K \oplus \overline L_K \ . $$ 
    
    Since $L$ is Lagrangian in $E^\IC$ it follows that $L_K$ is Lagrangian in $(K^\perp)^\IC / K^\IC$. Similarly $\overline L_K$ is Lagrangian in $(K^\perp)^\IC / K^\IC$. The condition that their intersection is zero is given by
    \begin{align}\label{eqn:BCGconditiononLs}
        \big(L\cap (K^\perp)^\IC + K^\IC\big) \cap \big(\overline L \cap (K^\perp)^\IC +K^\IC\big) \ \subset \ K^\IC \ .
    \end{align}
    In \cite[Lemma 5.1]{Bursztyn2007reduction} it is shown that Equation \eqref{eqn:BCGconditiononLs} is equivalent to the condition $\cJ(K) \cap K^\perp \subset K$, which is one of the requirements of Proposition \ref{prop:reducedcomplexdirac}. This condition appears in \cite[Proposition~6.1]{Zambon2008reduction}, generalising \cite[Theorem~5.2]{Bursztyn2007reduction}.
    \end{proof}

    Consider a generalised almost complex structure $\cJ$ on $E$. Recall from Section~\ref{sub:gencomplex} that the associated pure spinor is given by $$\psi = \exp(\alpha) \wedge \Theta \ , $$
    for a complex two-form $\alpha$ and a decomposable complex $k$-form $\Theta$, where $k={\rm type}(\cJ)$. Let us examine how the reduction of spinors proceeds in the examples of symplectic (Example~\ref{ex:symplecticmanifolds}) and complex (Example~\ref{ex:complexmanifolds}) generalised complex structures.
    
\begin{example}\label{ex:symplectic_red_spinors}
    Suppose $$\psi = \exp(\alpha)$$ and $\alpha$ satisfies $\iota_Y\, \iota_Z \alpha = 0$ for all $Y,Z\in K^\IC = T\cF^\IC$, but there is some $X \in T\cF^\IC$ such that $\iota_X\alpha \neq 0$ and $\iota_X \alpha$ is real. Then
    \begin{align}
        (X - \iota_X \alpha) \cdot \exp(\alpha) = 0
    \end{align}
    so that $X- \iota_X \alpha  \in L$. Since $\iota_Y\, \iota_X \alpha = 0$ for each $Y \in T\cF^\IC$, it follows that $X- \iota_X \alpha  \in L\cap (K^\perp)^\IC$. Thus $\iota_X \alpha \in L\cap (K^\perp)^\IC +K^\IC$ and $\overline{\iota_X\alpha}\in \overline{L} \cap (K^\perp)^\IC + K^\IC$. Since $\iota_X \alpha$ is real, this implies $$\iota_X \alpha \ \in \ \big(L\cap (K^\perp)^\IC +K^\IC\big)\cap\big(\overline L \cap (K^\perp)^\IC +K^\IC\big) \ , $$ but $\iota_X\alpha$ does not live in $K^\IC$. Hence such a two-form $\alpha$ cannot reduce to a generalised almost complex structure on $\red E$.
    
    Assume now that the foliation $\cF$ has compact orientable fibres, and suppose that $$\psi = \exp(\ii\,\omega)$$ for a real two-form $\omega.$
    By the above arguments, there is no $X\in T\cF^\IC$ such that $\iota_X \omega \in \ann(T\cF^\IC) $ and $\iota_X\omega$ is non-zero. 
    It follows that $\omega$ decomposes as $$\omega = \omega_\cQ + \omega_\cF \ , $$ where $\omega_\cQ \in \midwedge^2 \ann(T\cF)$ and $\iota_X\, \iota_Y \omega_\cF\neq 0$ for every pointwise linearly independent $X, Y \in T\cF^\IC$. 
    Take $$\mu = \frac{\omega_{\cF}^k}{k!} \ \in \ \det(T^*\cF) \ , $$ where $k = \frac 12\dim_\IR \cF$, and $\sfS_\mu$ as in Example~\ref{ex:Qfibreintegration}. 

     If $\psi$ is $T\cF^\IC$-invariant in the sense of Definition~\ref{def:invariantforms}, then $\omega_\cQ$ is a basic two-form and $\int_\cF\, \exp(\ii\,\omega_\cF)$ is a basic function. Thus $\Omega^\psi \subset A^\IC_\mu$, and 
     $$\Big(\exp(\ii\,\omega)\,,\, \big (\text{\footnotesize{$\int_\cF$}}\, \exp(\ii\,\omega_{\cF})\big) \ \exp(\ii\,\omega_{\cQ})\Big) \ \in \ \sfS_\mu^\IC \ . $$ Hence $\Omega^\psi \sim_{Q(\cF)} \Omega^{\red \psi}$, where $$\red \psi = \exp(\ii\,\red \omega_\cQ)$$ and $\varpi^* \red \omega_\cQ = \omega_\cQ$.
\end{example}

\begin{example}
    Let $$\psi = \Theta = \theta_1 \wedge \cdots \wedge \theta_n \ , $$ for some linearly independent complex one-forms $\theta_i$ such that $\theta_1 , \ldots , \theta_n,\overline \theta_1,\ldots, \overline \theta_n$ are also linearly independent. 
    Suppose that there is some $X \in TM^\IC$ with $X\notin T\cF^\IC$ and some $i\in \set{1,\ldots,n}$ such that $\iota_X \theta_i =0$, as well as some $a\in \IC$ such that $X+a \, \overline{X} =: Y\in K^\IC$.\footnote{Locally, such a one-form $\theta_i$ might look like $\de q^i +\ii\,\de s^i$, where $q^i$ is a coordinate on the base $\cQ$ and $s^i$ is a coordinate on the fibre $\cF$. Then $X = \frac\partial{\partial q^i} +\ii\,\frac\partial{\partial s^i}$ and $a=-1$.} 
    It then follows that $$X-Y \ \in \ \big(L\cap (K^\perp)^\IC + K^\IC\big)\cap\big( \overline L \cap (K^\perp)^\IC +K^\IC\big) \ . $$ 
    Since $X \notin T\cF^\IC$ it follows that $X-Y \notin T\cF^\IC = K^\IC$, and so in order for such an $n$-form $\Theta$ to reduce to a generalised almost complex structure on $\red E$, such an $X$ cannot exist. 

    By the arguments of Example~\ref{ex:symplectic_red_spinors} it follows that we can write $$\Theta = \theta_1 \wedge \cdots \wedge \theta_k \wedge \theta_{k+1} \wedge\cdots \wedge \theta_n$$ where $\theta_i\in \ann(T\cF^\IC)$ for $i=k+1,\ldots,n$ and $k = \frac 12 \dim_\IR \cF$. 
    Suppose that $\psi$ is invariant along $T\cF^\IC$, and take $$\mu = \theta_{1}\wedge\cdots\wedge\theta_k \ \in \ \det(T^*\cF^\IC) \ . $$ Recall that $\dim_\IR \cQ = 2n-2k$ and consider 
    \begin{align}
        \sfS_\mu' = \set{\Big(\varpi^*\red\beta\wedge\mu \,,\, \big({\text{\footnotesize{$\int_{\cF}$}}} \, \mu \wedge \overline \mu\big) \ \red\beta\Big) \ \Big\vert \ \red \beta \in \midwedge^\bullet T^*\cQ^\IC} \ .
    \end{align}
    It follows that $\Omega^\psi \sim_{Q(\cF)} \Omega^{\red \psi}$ for $$\red \psi = \big({\text{\footnotesize{$\int_\cF$}}} \, \mu \wedge \overline \mu\big)\ \red \Theta \ , $$ where $\varpi^* \red \Theta = \theta_{k+1}\wedge \cdots\wedge \theta_n$. Thus $\psi$ reduces to a generalised almost complex structure $\red\psi$ of type $n - k$, i.e. an ordinary almost complex structure on $\cQ$.
\end{example}

\medskip

\subsection{T-Duality and Type Change}~\\[5pt]
As in Section~\ref{ssec:SpinorrelationTduality}, let us consider again the case where $R\colon (\IT \cQ_1, \red H{}_1) \rel (\IT \cQ_2, \red H{}_2)$ is a geometric T-duality relation, hence a generalised isometry, with $\cQ_1$ and $\cQ_2$ both fibred over the same manifold $\cB$ with bundle projections $\pi_i:\cQ_i\to\cB$. While the following arguments can be made for an arbitrary generalised almost complex structure,  for simplicity we assume integrability of the generalised almost complex structures. By a \emph{$D_1$-invariant generalised complex structure} $\cJ_1$ on $\cQ_1$ we mean a generalised complex structure whose $+\ii$-eigenbundle $L_1$ (and hence $-\ii$-eigenbundle $\overline L_1$ and associated spinor $\psi_{1}$) is $D_1$-invariant as in Definition \ref{def:invDirac}. An immediate consequence of Theorem~\ref{thm:relatedspinorline} and Remark~\ref{rem:3rddefinition} is 

\begin{proposition}\label{prop:TdualGCS}
    If $\cJ_1$ is a $D_1$-invariant generalised complex structure on $\cQ_1$, then there is a generalised complex structure $\cJ_2$ on $\cQ_2$ such that $R\colon(\IT \cQ_1, \red H{}_1, \cJ_1) \rel (\IT \cQ_2, \red H{}_2, \cJ_2)$ is a generalised complex relation.
\end{proposition}

\begin{proof}
    The existence of a generalised almost complex structure $\cJ_2$ is assured by Theorem \ref{thm:relatedspinorline}. From Section \ref{ssec:SpinorrelationTduality}, there is an $R$-Clifford relation $\sfS_R$ which is a $(\de_{\red H{}_1}, \de_{\red H{}_2})$-DGO relation. The integrability of $\cJ_2$ then follows by Proposition \ref{prop:integrabilityrelation}.
\end{proof}

The type of a generalised complex structure can be determined by the pure spinor line bundle it generates. For instance, if $\cJ_1$ is a generalised complex structure on $\IT \cQ_1$ of type $k$, then 
\begin{align}
    \psi_{1} = \exp(B_1+\ii\,\omega_1) \wedge \Theta_1 \ ,
\end{align}
where $B_1$ and  $\omega_1$ are real two-forms, while $\Theta_1 = \theta_1 \wedge \cdots \wedge \theta_k$ for some linearly independent complex one-forms $\theta_i$. 

Let $\alpha = B_1 + \ii\,\omega_1$. At a point $c=(c_1,c_2) \in \red C=\cQ_1\times_\cB\cQ_2$, with $b=\pi_1(c_1)=\pi_2(c_2)\in\cB$, there are decompositions
\begin{align}
    \alpha &= \dot \alpha_\cB + \alpha_1 + \alpha_{\rm mix} \ ,\\[4pt]
    \Theta_1 &= \dot\mu_\cB \wedge \mu_1 \wedge \mu_{\rm mix} \ ,
\end{align}
where 
\begin{align}
    \begin{split}
\alpha_\cB \ \in \ \midwedge^2 T_b^*\cB^\IC & \quad , \quad \alpha_1 \ \in \ \midwedge^2 T_{c_1}^*\red\cF{}_1^\IC \quad , \quad \alpha_{\rm mix} \ \in \ ( \pi_1^* T^*\cB^\IC\wedge T^*\red\cF{}_1^\IC)_{c_1} \ , \\[4pt] \mu_{\cB} \ \in \ \midwedge^{k_1} T_b^*\cB^\IC & \quad , \quad \mu_1 \ \in \ \midwedge^{k_2} T_{c_1}^*\red\cF{}_1^\IC \quad , \quad \mu_{\rm mix} \ \in \ \midwedge^{k_3} (\pi_1^*T^*\cB^\IC + T^*\red\cF{}_1^\IC)_{c_1}
\end{split}
\end{align}
are complex forms with $k_1,k_2,k_3\in\IN_0$ non-negative integers. As previously, an overdot denotes the pullback by $\pi_1^*$ while a double overdot will denote the pullback by $\pi_2^*$. Note that $$k = k_1 + k_2 + k_3$$ is the type of $\cJ_1$.

\begin{proposition}\label{prop:Tdualspinor}
Let $\cJ_1$ be a $D_1$-invariant generalised complex structure on $\cQ_1$. The type of the T-dual generalised complex structure $\cJ_2$ on $\cQ_2$ is
\begin{align}\label{eqn:J2type}
    {\rm type}(\cJ_2) = {\rm type}(\cJ_1) - k_3 + m_3 
\end{align} 
where $m_3$ is the smallest positive integer such that $\alpha_{\rm mix}^{m_3} \neq 0$.
The pure spinor associated to $\cJ_2$ is given by
\begin{align}\label{eqn:J2decomp}
    \psi_{2} = \exp(B_2 + \ii\, \omega_2) \wedge \Theta_2 \ ,
\end{align}
where the real two-forms $B_2$ and $\omega_2$ and the complex form $\Theta_2$ are given by Equation \eqref{eqn:J2spinorterms} below.
\end{proposition}

\begin{proof}
We can write
\begin{align}
    \alpha_{\rm mix} &= \dot b_\alpha^1 \wedge f_\alpha^1 + \cdots + \dot b_\alpha^{m_3} \wedge f_\alpha^{m_3}\ ,\\[4pt]
    \mu_{\rm mix} &= (\dot{{b}}_\mu^1  + {f}_\mu^1) \wedge \cdots \wedge (\dot{{b}}_\mu^{k_{3}} + {f}_\mu^{k_{3}})\ ,
\end{align}
for $b_\alpha^i, b_\mu^j \in T_b^*\cB^\IC$ and $f_\alpha^i, f_\mu^j \in \ccG^\star_1{}_{c_1}^\IC$ for $i=1,\ldots,m_3$ and $j=1,\ldots,k_3$, where $\ccG_1^\star = \rmp_{T^*\cQ_1}\,\rmp_1(\ccG)$ as in the construction of the $R$-Clifford relation $\sfS_R$ from Section~\ref{ssec:SpinorrelationTduality}. 

Choose $\widehat\mu_1, \widehat f_\alpha^{\,i}, \widehat\alpha_1, \widehat{{f}}_\mu^{\,i} \in \rmp_1(\ccG^\IC)_{c_1}\cap \midwedge^\bullet T_{c_2}\cQ_2^\IC$ such that 
\begin{align}
    \widehat \mu_1 \cdot \mu_1 = \widehat f_\alpha^{\,i} \cdot f_\alpha^i= \widehat \alpha_1 \cdot \alpha_1 = \widehat{f}_\mu^{\,i} \cdot f_\mu^i = 1 \ ,
\end{align}
while all other Clifford actions among these elements are zero, and moreover $\widehat\mu_1 \cdot \overline{\mu}_1 = 0$. Let also $\widehat \alpha_{\rm mix}$  be given by exchanging $f_\alpha^i$ with \smash{$\widehat f_\alpha^{\,i}$} and $\dot b^i_\alpha$ with $\ddot b^i_\alpha$ in $\alpha_{\rm mix}$, for $i=1,\ldots, m_3$, and similarly for~$\widehat\mu_{\rm mix}$.

Let $m_2 \in \IN$ be the smallest natural number such that $\alpha_1^{m_2} \neq 0$, and set
\begin{align}
    \lambda = \frac{1}{m_2!}\,\alpha_1^{m_2} \wedge f_\alpha^1 \wedge \cdots \wedge f_\alpha^{m_3} \wedge \mu_1 \wedge \overline{\mu}_1 \wedge f_\mu^1 \wedge \cdots \wedge f_\mu^{k_3} \ .
\end{align}
Since $\psi_{1}$ defines a $D_1$-invariant generalised complex structure, it follows that $\lambda$ defines a $D_1$-invariant fibre volume form on the tangent distribution $D_1 = T\red\cF{}_1$.
The considerations above then enable a factorisation $\psi_{1} = \widehat \psi \cdot \lambda$ thanks to

\begin{lemma} \label{lem:formsdefined}
The forms defined above satisfy
    \begin{enumerate}
        \item \ $\displaystyle \exp(\widehat \alpha_1) \cdot \frac{1}{m_2!}\,\alpha_1^{m_2} = \exp(\alpha_1) \ , $ \\
        \item \ $\exp(\widehat \alpha_{\rm mix}) \cdot ( f_\alpha^1 \wedge \cdots \wedge f_\alpha^{m_3}) = (\ddot b_\alpha^1 + f_\alpha^1) \wedge \cdots \wedge (\ddot b_\alpha^{m_3} + f_\alpha^{m_3}) \ , $ \\[-3mm]
        \item \ $\widehat \mu_1 \cdot(\mu_1 \wedge \overline{\mu}_1) = \overline{\mu}_1 \ , $ \ and \\[-3mm]
        \item \ $\widehat\mu_{\rm mix} \cdot (f_\mu^1 \wedge \cdots \wedge f_\mu^{k_3}) = \exp(\ddot{b}_\mu^1 \wedge f_\mu^1 + \cdots +\ddot{b}_\mu^{k_3} \wedge f_\mu^{k_3}) \ .$
    \end{enumerate}
\end{lemma}

\begin{proof}
    (1) Firstly, for each natural number $j\leqslant m_2$
    \begin{align}
        \widehat \alpha_1^j \cdot \alpha_1^{m_2} = m_2\, j\ \widehat\alpha_1^{j-1} \cdot \alpha_1^{m_2-1} = \cdots = \frac{m_2!\, j!}{(m_2-j)!}\,\alpha_1^{m_2-j} \ .
    \end{align}
    Thus
    \begin{equation}
    \begin{aligned}
        \exp(\widehat \alpha_1) \cdot \frac{1}{m_2!}\, \alpha_1^{m_2} = \sum_{j=0}^{m_2}\,\frac1{j!}\,\widehat\alpha_1^j  \cdot \frac{1}{m_2!}\, \alpha_1^{m_2}=\sum_{j=0}^{m_2}\, \frac{1}{m_2!\,j!}\, \frac{m_2!\, j!}{(m_2-j)!}\, \alpha_1^{m_2-j} = \exp( \alpha_1)\ .
    \end{aligned}
    \end{equation}
    
    \noindent (2) We compute
    \begin{equation}
    \begin{aligned}
        \exp(\widehat \alpha_{\rm mix}) \cdot ( f_\alpha^1 \wedge \cdots \wedge f_\alpha^{m_3}) &= \exp(\ddot b_\alpha^1 \wedge \widehat f_\alpha^{\,1} ) \wedge \cdots \wedge \exp(\ddot b_\alpha^{m_3} \wedge \widehat f_\alpha^{\,m_3} ) \cdot ( f_\alpha^1 \wedge \cdots \wedge f_\alpha^{m_3})\\[4pt]
        &=\big(\exp(\ddot b_\alpha^1 \wedge \widehat f_\alpha^{\,1} ) \cdot f_\alpha^1\,\big) \wedge \cdots \wedge \big(\exp(\ddot b_\alpha^{m_3} \wedge \widehat f_\alpha^{\,m_3} )\cdot f_\alpha^{m_3}\,\big) \\[4pt]
        &=\big((1+\ddot b_\alpha^1 \wedge \widehat f_\alpha^{\,1} ) \cdot f_\alpha^1 \big)\wedge \cdots \wedge \big((1+\ddot b_\alpha^{m_3} \wedge \widehat f_\alpha^{\,m_3} )\cdot f_\alpha^{m_3}\big)\\[4pt]
        &=(\ddot b_\alpha^1 + f_\alpha^1) \wedge \cdots \wedge (\ddot b_\alpha^{m_3} + f_\alpha^{m_3}) \ ,
    \end{aligned}
    \end{equation}
    where we use the fact that $\widehat f_\alpha^{\,i} \cdot \ddot b_\alpha^j = 0$ for all $i,j$ and $\widehat f_\alpha^{\,i} \cdot f_\alpha^j = 0$ when $i\neq j$ to commute the exponentials.

    \noindent (3) Follows immediately.

    \noindent (4) Similarly to (2) we compute
    \begin{equation}
    \begin{aligned}
        \widehat\mu_{\rm mix} \cdot (f_\mu^1 \wedge \cdots \wedge f_\mu^{k_3}) &= (\ddot{b}_\mu  + \widehat{f}_\mu^{\,1}) \wedge \cdots \wedge (\ddot{b}_\mu^{k_{3}} + \widehat{f}_\mu^{\,k_{3}}) \cdot (f_\mu^1 \wedge \cdots \wedge f_\mu^{k_3})\\[4pt]
        &= \big((\ddot{b}_\mu  + \widehat{f}_\mu^{\,1}) \cdot f_\mu^1\big)\wedge \cdots \wedge \big((\ddot{b}_\mu^{k_{3}} + \widehat{f}_\mu^{\,k_{3}}) \cdot f_\mu^{k_3}\big)\\[4pt]
        &= (1 + \ddot{b}_\mu\wedge  f_\mu^1)\wedge \cdots \wedge (1+ \ddot{b}_\mu^{k_{3}}\wedge f_\mu^{k_3})\\[4pt]
        &= \exp(\ddot{b}_\mu^1 \wedge f_\mu^1 + \cdots + \ddot{b}_\mu^{k_3} \wedge f_\mu^{k_3}) \ ,
    \end{aligned}
    \end{equation}
    as required.
\end{proof}

We can now complete the proof of Proposition \ref{prop:Tdualspinor}.
By Proposition~\ref{prop:Uisomorphismsmooth}, there is a pure spinor $\psi_{2}$ such that $(\psi_{1}, \psi_{2}) \in \sfS^\IC_{R\,c}$. 
From Lemma~\ref{lem:formsdefined} we obtain $$(\psi_{1},\psi_{2}) = (\widehat \psi , \psi_{2}) \cdot(\lambda,1) \ , $$ where 
\begin{align}
    \widehat \psi = \exp(\dot \alpha_\cB + \widehat\alpha_1 + \widehat \alpha_{\rm mix})\wedge \dot \mu_\cB \wedge \widehat\mu_1 \wedge \widehat \mu_{\rm mix} \ .
\end{align}
Note that $(\dot \beta,\ddot \beta) \in \Cl(R)_c$ for  $\beta \in T_b^*\cB$, by the splitting \eqref{eqn:splittingforRB}.
It follows that
\begin{equation}\label{eqn:tdualcomplexspinor}
\begin{aligned}
    \psi_{2} &= \exp(\ddot \alpha_\cB) \wedge\exp(\alpha_2) \wedge (\ddot b_\alpha^1 + \widetilde f_\alpha^{\,1}) \wedge \cdots \wedge (\ddot b_\alpha^{m_3} + \widetilde f_\alpha^{\,m_3}) \\
    &\qquad \qquad \qquad\qquad\qquad \wedge \ddot\mu_\cB \wedge \overline{\mu}_2 \wedge \exp(\ddot{b}_\mu^1 \wedge \widetilde{f}_\mu^{\,1} + \cdots +\ddot{b}_\mu^{k_3} \wedge \widetilde{f}_\mu^{\,k_3})
\end{aligned}
\end{equation}
for some $\alpha_2\in\midwedge^2 T_{c_2}^*\red\cF{}_2^\IC$ and $\mu_2\in\midwedge^{k_2}T^*_{c_2}\red\cF{}_2^\IC$ such that $\widehat\alpha_1\sim_R\alpha_2$ and $\widehat\mu_1\sim_R\mu_2$ respectively, with $\widetilde f_\alpha^{\,i}$ (resp. $\widetilde f_\mu^{\,j}$) such that $\widehat f_\alpha^{\,i} \sim_R \widetilde f_\alpha^{\,i}$ (resp. \smash{$\widehat f_\mu^{\,j} \sim_R \widetilde f_\mu^{\,j}$}) in $\Cl(R)^\IC_c$, for $i = 1,\ldots, m_3$ and $j = 1, \ldots , k_3$.

Finally, we can write \eqref{eqn:tdualcomplexspinor} in the form
\begin{equation}
\begin{aligned}
    \psi_{2} = \exp(B_2+\ii\,\omega_2) \wedge \Theta_2 \ ,
\end{aligned}
\end{equation}
where we have made the definitions
\begin{equation}\label{eqn:J2spinorterms}
    \begin{aligned}
    B_2 &\coloneqq {\rm Re}(\ddot \alpha_\cB + \alpha_2 + \alpha_{\mu_{\rm mix}})\ ,\\[4pt]
    \omega_2 &\coloneqq {\rm Im}(\ddot \alpha_\cB + \alpha_2 +\alpha_{\mu_{\rm mix}})\ ,\\[4pt]
    \Theta_2 &\coloneqq \mu_{\alpha_{\rm mix}} \wedge \ddot\mu_\cB \wedge \overline{\mu}_2\ ,
\end{aligned}
\end{equation}
with
\begin{align}
    \alpha_{\mu_{\rm mix}}&=\ddot{b}_\mu^1 \wedge \widetilde{f}_\mu^{\,1} + \cdots +\ddot{b}_\mu^{k_3} \wedge \widetilde{f}_\mu^{\,k_3}\ ,\\[4pt]
    \mu_{\alpha_{\rm mix}}&=(\ddot b_\alpha^1 + \widetilde f_\alpha^{\,1}) \wedge \cdots \wedge (\ddot b_\alpha^{m_3} + \widetilde f_\alpha^{\,m_3})\ .
\end{align} 
Hence the type of $\cJ_2$ is $k_1 + k_2 + m_3$, and the result follows.
\end{proof}

\begin{remark}
Since $\lambda$ defines a fibre volume form on the leaves of $\red\cF{}_1$, it follows that $$d := \rk_\IR(D_1) = 2k_2 + k_3 + 2m_2 + m_3 \ . $$ As ${\rm type}(\cJ_2) = {\rm type}(\cJ_1) - k_3 + m_3 $, we then find
\begin{align}
    {\rm type}(\cJ_2) = {\rm type}(\cJ_1) + 2(k_2 + m_2 + m_3)-d \ .
\end{align}
This is in agreement with the discussion of change of type of generalised complex structures under T-duality for principle torus bundles in \cite[Example 4.3]{cavalcanti2011generalized}.
\end{remark}

\begin{example}
Continuing in analogy with \cite{cavalcanti2011generalized}, the expression \eqref{eqn:tdualcomplexspinor} allows us to explicitly determine what becomes of the type in specific instances of generalised complex structures on $\cQ_1$. Suppose that $\dim_\IR (\cQ_1) = 2n$ and that the foliation $\red\cF{}_1$ has real codimension $n$, i.e. $\rk_\IR(D_1) =n$. 

\begin{itemize}
\item[$\Diamond$] If $\cJ_1$ is purely complex and the leaves of  $\red\cF{}_1$ are also complex, then the associated pure spinor has the form
\begin{align}
    \psi_{1} = \dot \mu_\cB \wedge \mu_1 \ ,
\end{align}
that is, there is no $\mu_{\rm mix}$ component. Thus the T-dual spinor corresponding to the complex structure $\cJ_2$ is given by
\begin{align}
    \psi_{2} = \ddot\mu_{\cB} \wedge \overline{\mu}_2 \ .
\end{align}
Hence $\cJ_2$ is also a purely complex structure. \\[-3mm]

\item[$\Diamond$] If $\cJ_1$ is a complex structure given by $J_1$ and the leaves of  $\red\cF{}_1$ are real (i.e. $J_1(T\red\cF{}_1) \cap T\red\cF{}_1 = \set{0}$), then 
\begin{align}
    \psi_{1} = \mu_{\rm mix}
    \end{align}
    and hence
    \begin{align}
    \psi_{2} = \exp\big({\rm Re}(\alpha_{\mu_{\rm mix}}) + \ii\, {\rm Im}(\alpha_{\mu_{\rm mix}})\big) \ .
\end{align}
Thus $\cJ_2$ is symplectic with a $B$-twist. Since $\alpha_{\mu_{\rm mix}}$ has always a component in $\pi_2^*T^*\cB^\IC$, the leaves of  $\red\cF{}_2$ define Lagrangian submanifolds with respect to this symplectic structure. \\[-3mm]

\item[$\Diamond$] If $\cJ_1$ is purely symplectic (with a $B$-twist) and the leaves of $\red\cF{}_1$ are symplectic, then $$\psi_{1} = \exp(\dot\alpha_\cB+\alpha_1) \ , $$ and hence $$\psi_{2} = \exp(\ddot \alpha_\cB+\alpha_2)$$ is also symplectic. \\[-3mm]

\item[$\Diamond$] If $\cJ_2$ is symplectic and the leaves of $\red\cF{}_1$ are Lagrangian, then $$\psi_{1} = \exp(\alpha_{\rm mix}) \ , $$ and hence $$\psi_{2} = \mu_{\alpha_{\rm mix}} \ , $$ which therefore has complex type.
\end{itemize}

We thus recover and generalise the results of \cite[Example 4.3, Table 1]{cavalcanti2011generalized} in the case of torus bundles, which we summarize here in the table
\begin{table}[h]
    \centering
    \begin{tabular}{c|c||c|c}
        Structure on $\cQ_1$ & Leaves of $\red\cF{}_1$ & Structure on $\cQ_2$ & Leaves of $\red\cF{}_2$ \\ \hline \hline
        Complex & Complex & Complex & Complex \\ \hline
        Complex & Real & Symplectic & Lagrangian \\ \hline
        Symplectic & Symplectic & Symplectic & Symplectic \\ \hline
        Symplectic & Lagrangian & Complex & Real
    \end{tabular}
    \label{tab:structures_and_fibres}
\end{table}
\end{example}

\medskip

\section{Generalised K\"ahler Relations}\label{sect:section7}
In this final section we study how generalised K\"ahler structures change under Courant algebroid relations, including many examples of their behaviours under geometric T-duality. Among other things, this enables us to exhibit the T-duality of two-dimensional sigma-models with $\cN=(2,2)$ supersymmetry within our framework, as well as mirror symmetry of Calabi-Yau manifolds, and to produce a new local example of a non-K\"ahler generalised Calabi-Yau structure.

\medskip

\subsection{Relating Generalised K\"ahler Structures}~\\[5pt]
Starting from the notion of a Courant algebroid relation between generalised complex structures from Section~\ref{sect:section6}, it is  natural to induce a similar construction for generalised K{\"a}hler structures.

\begin{definition} \label{def:genkahlerrel}
    Let $R\colon E_1 \rel E_2$ be a Courant algebroid relation supported on $C\subseteq M_1\times M_2$. Let $\cK_1 = (\cJ_1^+, \cJ_1^-)$ and $\cK_2 = (\cJ_2^+,\cJ_2^-)$ be generalised almost K\"ahler structures on $E_1$ and $E_2$ respectively.
    Then $R$ is a \emph{generalised almost K\"ahler relation} if it is a generalised almost complex relation for both $R:(E_1, \cJ^\pm_1) \rel (E_2, \cJ^\pm_2)$.
    If $\cK_1$ and $\cK_2$ are generalised K\"ahler structures, then $R$ is a \emph{generalised K\"ahler relation}.
    If $\cK_1$ and $\cK_2$ are generalised Calabi-Yau structures, then $R$ is a \emph{generalised Calabi-Yau relation}.
\end{definition}

It follows from Proposition~\ref{prop:equivGCRdefs} that generalised almost K\"ahler relations can be equivalently characterised through

\begin{proposition} \label{prop:GKRdecomp}
    A generalised almost K\"ahler relation $R$ decomposes pointwise into four isotropic subbundles as
    \begin{align}
        R_c^\IC = R_c^{+\ii,+\ii} \oplus R_c^{+\ii,-\ii} \oplus R_c^{-\ii,+\ii} \oplus R_c^{-\ii,-\ii}
    \end{align}
    at each $c\in C$, where 
    \begin{align} 
    R_c^{+\ii,+\ii} = R^\IC_c \cap \big((L_1^+\cap L_1^-)\times (L_2^+\cap L_2^-)\big)_c \quad & , \quad R_c^{+\ii,-\ii} = R_c^\IC \cap \big((L_1^+\cap \overline L{}_1^-)\times (L_2^+\cap \overline L_2{}^-)\big)_c \ , \\[4pt]
    R_c^{-\ii,+\ii} = R^\IC_c \cap \big((\overline L_1{}^+\cap L_1^-)\times (\overline L{}_2^+\cap L_2^-)\big)_c \quad & , \quad R_c^{-\ii,-\ii} = R_c^\IC \cap \big((\overline L{}_1^+\cap \overline L{}_1^-)\times (\overline L{}_2^+\cap \overline L{}_2^-)\big)_c \ ,
    \end{align}
    with $L_i^\pm$ the $+\ii$-eigenbundle of $\cJ_i^\pm$ respectively, for $i=1,2$.
\end{proposition}

The generalized metric $V^\pm$ of a generalised almost K\"ahler structure $\cK = (L^+,L^-)$ can be decomposed as $$V^+{}^\IC = (L^+\cap L^- )\oplus (\overline{L}{}^+\cap \overline{L}{}^-) \qquad \text{and} \qquad V^-{}^\IC = (L^+\cap \overline L{}^-) \oplus (\overline{L}{}^+\cap L^-) \ . $$ From this decomposition and Proposition~\ref{prop:GKRdecomp} we get

\begin{lemma}
Let $R\colon (E_1, \cK_1) \rel (E_2, \cK_2)$ be a generalised almost K\"{a}hler relation. Then $R$ is a generalised isometry between the generalised metrics $V_1^+$ and $V_2^+$ defined by $\cK_1$ and $\cK_2$ respectively. 
\end{lemma}

\begin{proof}
    We obtain
    \begin{align}
        R^\IC_c = \big(R_c^{+\ii,+\ii} \oplus R_c^{-\ii,-\ii}\big) \, \oplus \, \big(R_c^{+\ii,-\ii} \oplus R_c^{-\ii,+\ii}\big) = R_c^+{}^\IC \oplus R^-_c{}^\IC = (R^+_c\oplus R_c^-)^\IC \ ,
    \end{align}
    and hence $R_c = R_c^+ \oplus R_c^-$ for all $c\in C$.
\end{proof}

Equivalently, let $\tau_1=-\cJ_1^+\, \cJ_1^-\in{\sf Aut}(E_1)$ and $\tau_2=-\cJ_2^+ \, \cJ_2^-\in {\sf Aut}(E_2)$ be the generalised metrics  defined by $\cK_1$ and $\cK_2$ respectively. From the definition of a generalised almost K\"ahler relation it follows that
\begin{align}
    (\tau_1 \times \tau_2)(R) = (\cJ_1^+\, \cJ_1^- \times \cJ_2^+ \, \cJ_2^-)(R) = (\cJ_1^+\times\cJ_2^+)\,(\cJ_1^-\times\cJ_2^-)(R) = R
\end{align}
and hence $R$ is a generalised isometry.

Since a generalised almost K\"ahler relation $R$ is necessarily a generalised isometry, by \cite[Proposition 4.17]{DeFraja:2023fhe} it follows that both its kernel ${\rm K}_R $ and cokernel $ {\rm CK}_R$ are trivial.

\begin{remark}[\textbf{Relating Bi-Hermitian Structures}]\label{rem:Bihermitian}
    Since a generalised almost K\"ahler relation $R\colon (E_1, \cK_1) \rel (E_2, \cK_2)$ is a generalised isometry, there is a splitting
    \begin{align}
        R = R^+\oplus R^-
    \end{align}
    where $R^\pm = R\cap(V_1^\pm \times V_2^\pm)$ for the generalised metrics $V_1^+$ and $ V_2^+$ associated to $\cK_1$ and $\cK_2$ respectively. Since $\cJ^\pm_i$ commute with the metrics $\tau_i$, it follows that $$(\cJ_1^\pm \times \cJ_2^\pm)(R^\pm) = R^\pm \ . $$ Define $$T^\pm = \rho(R^\pm) \ \subset \ TC \ , $$ so that $\rho\colon R^\pm \to T^\pm$ are isomorphisms covering the identity. 
    
    The bi-Hermitian structures $(g_i, b_i, J_i^+, J_i^-)$ associated with $\cK_i$ act on $T^+$ and $T^-$, and since $\cJ_1^\pm \times \cJ_2^\pm$ preserves $R^\pm$ respectively, it follows that
    \begin{align}
        \rho \circ (\cJ_1^\pm \times \cJ_2^\pm) \circ (\rho |_{R^\pm})^{-1} \colon T^\pm \longrightarrow T^\pm
    \end{align}
    is an isomorphism. In other words $$(J_1^\pm \times J_2^\pm)(T^\pm) = T^\pm \ . $$ In this sense we say that the bi-Hermitian structures are $R$-related.
\end{remark}

Because each generalised almost complex structure $\cJ^\pm$ of a generalised almost K\"ahler structure $\cK=(\cJ^+,\cJ^-)$ induces a grading of the form \eqref{eqn:GCSgrading}, there is a bigrading
\begin{align}
    \midwedge^\bullet T^*M^\IC = \bigoplus_{p,q=0}^n\, \Omega^{p,q} \ .
\end{align}
From Proposition \ref{prop:gradingforGCR} we may then infer

\begin{proposition}\label{prop:bigradingGKR}
    There is a bigrading 
    \begin{align}
        \cS_R^\IC = \bigoplus_{p,q=0}^n\, \cS_R^{p,q}
    \end{align}
    where $$\cS_R^{p,q} = \big((\Omega_1^{p,q} \times \Omega_2^{p,q}) \otimes |\det \sfS_R^* |^{\frac{1}{2^n}}\big) \cap \cS_R^\IC \ . $$
\end{proposition}

\medskip

\subsection{T-Duality for Generalised K\"{a}hler Structures}~\\[5pt]
Discussing T-duality for generalised K{\"a}hler structures follows naturally from Proposition \ref{prop:TdualGCS} concerning T-duality for generalised complex structures.

Suppose that $R \colon (\IT \cQ_1 , \red H{}_1) \rel (\IT \cQ_2, \red H{}_2)$ is a geometric T-duality relation. By a \emph{$D_1$-invariant generalised K\"ahler structure} $\cK_1=(\cJ_1^+,\cJ_1^-)$ on $\cQ_1$ we mean a generalised K\"ahler structure both of whose generalised complex structures $\cJ_1^+$ and $\cJ_1^-$ are $D_1$-invariant. Then from Proposition \ref{prop:TdualGCS} we immediately arrive at

\begin{proposition}\label{prop:TdualGKS}
    If $\cQ_1$ has a $D_1$-invariant generalised K\"ahler structure $\cK_1$, then there is a generalised K\"ahler structure $\cK_2$ on $\cQ_2$ such that $R \colon (\IT \cQ_1, \red H{}_1, \cK_1) \rel (\IT \cQ_2, \red H{}_2, \cK_2)$ is a generalised K\"ahler relation.
\end{proposition}

This approach to T-duality for generalised K{\"a}hler structures fully recovers the approach proposed in \cite{Lindstrom:2007sq}.

\begin{example}[\textbf{T-Dual Bi-Hermitian Structures for Torus Bundles}]
    In the setting of T-duality for principal torus bundles, there is an explicit formula for the T-dual bi-Hermitian structure given by \cite[Example 4.7]{cavalcanti2011generalized}
    \begin{align}
        J_2^\pm = \red\rho{}_2^\pm \, \mathscr{R} \, \red\rho{}_1^\pm{}^{-1} \, J_1^\pm \, \big(\red\rho{}_2^\pm \, \mathscr{R} \, \red\rho{}_1^\pm{}^{-1}\big)^{-1} \ .
    \end{align}
    where the projections $\red\rho{}^\pm_i:V_i^\pm\to T\cQ_i$ are isomorphisms from the generalised metrics
    \begin{align}
        V_i^\pm = \{X_i + b_i(X_i)\pm g_i(X_i) \ \vert \ X_i\in T\cQ_i\} \ ,
    \end{align}
    and $\ccR$ is the $C^\infty(\cB)$-module isomorphism \eqref{eq:ccRiso} induced by $R$. 
    
    For $(X_1, X_2) \in T^\pm=\rho(R^\pm)$, we compute
    \begin{equation}
    \begin{aligned}
        (J_1^\pm \times J_2^\pm)(X_1,X_2) &= 
        \big(\red\rho{}_1^\pm\, \cJ_1^\pm \, \red\rho{}_1^\pm{}^{-1} X_1 \, , \, \red\rho{}_2^\pm\, \mathscr{R} \, \red\rho{}_1^\pm{}^{-1}\, J_1^\pm\, \red\rho{}_1^\pm \, \mathscr{R}^{-1} \, \red\rho{}_2^\pm{}^{-1} X_2\big)\\[4pt]
        &=\big(\red\rho{}_1^\pm\, \cJ_1^\pm X_1^\pm \, , \, \red\rho{}_2^\pm \, \mathscr{R} \, \red\rho{}_1^\pm{}^{-1}\, \red\rho{}_1^\pm\, \cJ_1^\pm \, \red\rho{}_1^\pm{}^{-1}\, \red\rho{}_1^\pm\, \mathscr{R}^{-1} X_2^\pm\big)\\[4pt]
        &=\big(\red\rho{}_1^\pm\, \cJ_1^\pm X_1^\pm\, , \,  \red\rho{}_2^\pm\, \mathscr{R}\, \cJ_1^\pm\, \mathscr{R}^{-1} \, X_2^\pm\big)\\[4pt]
        &=\big(\red\rho{}_1^\pm \times \red\rho{}_2^\pm\big)\big(\cJ_1^\pm X^\pm_1 \, , \, \mathscr{R}\, \cJ_1^\pm X_1^\pm\big) \ ,
    \end{aligned}
    \end{equation}
    where $X_i^\pm = \red\rho{}_i^\pm{}^{-1}X_i\in V_i^\pm$ for $i=1,2$.
    Since $\mathscr{R}$ is induced by the relation $R$, it follows that $$(J_1^\pm \times J_2^\pm)(T^\pm) = T^\pm \ , $$ as anticipated by Remark~\ref{rem:Bihermitian}.
\end{example}

\begin{remark}[\textbf{T-Duality for Supersymmetric Sigma-Models}]
    The Lagrangian density of the $\mathcal{N}=(2,2)$ supersymmetric sigma-model of \cite{Gates:1984nk} is determined by the Poisson structure $$\Pi = [J^+,J^-]\,g^{-1}$$  associated with the generalised K{\"a}hler structure $\cK = (\cJ^+, \cJ^-)$ \cite{Lindstrom:2005zr}, where $(g,b,J^+,J^-)$ is the corresponding bi-Hermitian structure. In the setting of Proposition \ref{prop:TdualGKS}, we may think of the generalised K{\"a}hler structure $\cK_1$ as the one determining the Lagrangian density of a supersymmetric sigma-model via $\Pi_1 = [J_1^+,J_1^-]\,g_1^{-1}$ so that its T-dual sigma-model is obtained directly by constructing the T-dual generalised K{\"a}hler structure $\cK_2$ along the lines of Proposition \ref{prop:TdualGCS} and Theorem~\ref{thm:relatedspinorline}. 
    
    The Lagrangian formulation of the $\cN=(2,2)$ supersymmetric sigma-model is tied to the existence of a generalised K{\"a}hler potential.
    A generalised K{\"a}hler structure $\cK=(\cJ^+, \cJ^-)$ is associated with the real Poisson structures $$\Pi_\pm = \tfrac{1}{2}\,(J^+ \pm J^-)\,g^{-1} \ . $$ If these Poisson structures are integrable, then a generalised K{\"a}hler potential can be defined from a bisection of the resulting $(1,1)$-morphism of the integrating double symplectic  groupoids \cite{Alvarez:2024wki}. Thus a suitable integration of Courant algebroid relations is needed to describe the ``morphism'' of generalised K{\"a}hler potentials dictated by T-duality, as hinted to in \cite{Alvarez:2024wki}. This lies outside the scope of the present paper. 
\end{remark}

\begin{proposition}\label{prop:TdualGCY}
    Let $R\colon (\IT \cQ_1,\red H{}_1) \rel (\IT \cQ_2,\red H{}_2)$ be a geometric T-duality relation. If $\cK_1$ is a $D_1$-invariant generalised Calabi-Yau structure on $\cQ_1$, then there is a generalised Calabi-Yau structure $\cK_2$ on $\cQ_2$ such that $R\colon (\IT \cQ_1, \red H{}_1, \cK_1) \rel (\IT \cQ_2, \red H{}_2, \cK_2)$ is a generalised Calabi-Yau relation.
\end{proposition}

\begin{proof}
    Since $\cK_1$ is in particular a generalised  K\"ahler structure, the existence of a generalised K\"ahler structure $\cK_2$ on $\cQ_2$ is guaranteed by Proposition \ref{prop:TdualGKS}. The integrability of the generalised K\"ahler structures is determined by the $(\de_{\red H{}_1}, \de_{\red H{}_2})$-DGO relation $\sfS_R$ constructed in Section~\ref{ssec:SpinorrelationTduality}, and since there are nowhere vanishing sections $\psi_1^\pm\in \sfGamma(\Omega_1^\pm)$ such that $\de_{\red H{}_1} \psi^\pm_1 = 0$, it follows from Proposition~\ref{prop:Uisomorphismsmooth} that $\de_{\red H{}_2}\psi^\pm_2 = 0$ for $\psi^\pm_2$ such that $(\psi^\pm_1, \psi^\pm_2) \in \sfS_R$. The Mukai pairings are preserved by Proposition~\ref{prop:pairingspinor},  and hence $\cK_2$ is a generalised Calabi-Yau structure.
\end{proof}

\begin{example}[\textbf{Mirror Symmetry}]
\label{ex:mirror}
    Consider the setting~\cite{cavalcanti2011generalized} of a Calabi-Yau manifold $\cQ_1$ with T-duality taken along a special Lagrangian fibration $\red \cF{}_1$, i.e. $\dim_\IR \red \cF{}_1 = \frac 12 \dim_\IR \cQ_1$, the symplectic form $\omega_1$ vanishes along $\red \cF{}_1$, and the complex structure $J_1$ is real. This induces a generalised Calabi-Yau structure with $\cJ_1^+ = \cJ_{\omega_1}$ and $\cJ_1^- = \cJ_{J_1}$. From Example~\ref{tab:structures_and_fibres}, the T-dual generalised Calabi-Yau structure has swapped symplectic and complex structures, $\cJ_2^+ = \cJ_{J_2}$ and $\cJ_2^- = \cJ_{\omega_2}$. The bigrading of Proposition \ref{prop:bigradingGKR} allows us to associate $\Omega_1^{p,q}$ with $\Omega_2^{p,q}$.

    For an ordinary Calabi-Yau structure, there is an isomorphism of the ordinary $(p,q)$ decomposition with the decomposition into $\Omega^{p,q}$ given by
    \begin{align}
        \Omega_J^{p-q}\cap \Omega_\omega^{n-p-q} = \exp(\ii\,\omega) \cdot \exp(\ii\,\omega^{-1}) \cdot \midwedge^{p,q} T^*M^\IC \ .
    \end{align}
    Thus in T-duality we may associate $\midwedge^{p,q} T^*\cQ_1^\IC$ and $\midwedge^{n-p,q} T^*\cQ_2^\IC$, giving the usual relation between the Hodge diamonds of Betti numbers in mirror symmetry.
\end{example}

\begin{example}[\textbf{Semi-Flat Calabi-Yau Metrics}]
\label{ex:semi-flat}
    Following~\cite{Garcia-Fernandez:2020ope}, let $\cQ_1= \cB\times \sfT^k$ where $\cB\subset \IR^k$ is the unit $k$-ball. Take local coordinates $\set{x^1,\ldots, x^k}$ on $\cB$ and $\set{y^1,\ldots,y^k}$ on $\sfT^k$. For any $c>0$ there is a unique solution $u \colon \cB \to \IR$ to the Monge-Amp\`ere equation
    \begin{align}
        \det(u_{ij}) = c
    \end{align}
    with $u|_{\partial \cB} = 0$ and $u_{ij} >0$ in $\cB\setminus\partial\cB$, where $u_{ij} = \frac{\partial^2 u}{\partial {x^i}\, \partial {x^j}}$. In the following we set $c=1$ without loss of generality.
    
    We can define a $\sfT^k$-invariant Riemannian metric on $\cQ_1$ by
    \begin{align}
        g = \sum_{i,j}\, u_{ij}\, \big(\de x^i \otimes \de x^j + \de y^i \otimes \de y^j \big) \ .
    \end{align}
    With the complex structure $J\in{\sf Aut}(T\cQ_1)$ defined by $J(\frac{\partial}{\partial x^i}) = \frac\partial{\partial y^i}$ and $J(\frac\partial{\partial y^i}) = -\frac\partial{\partial x^i}$, the metric $g$ is K\"ahler with K\"ahler form given by
    \begin{align}
        \omega = \sum_{i,j}\, u_{ij}\ \de x^i \wedge \de y^j= \frac \ii2\, \sum_{i,j}\, u_{ij}\ \de z^i \wedge \de \overline z^j  \ ,
    \end{align}
    where $z^i=x^i+\ii\,y^i$.
    
    Thus we can define a generalised K\"ahler structure $\cK_1 = (\cJ_J, \cJ_\omega)$ with generalised metric given by $g$. This has Ricci form equal to
    \begin{align}
        {\rm Ric}_g = -\ii\, \overline \partial\, \partial \log \det(u_{ij}) = 0 \ ,
    \end{align}
    and hence $g$ is Ricci flat.
    The associated pure spinors are $$\psi_J = \de z^1 \wedge \cdots \wedge \de z^k \qquad \text{and} \qquad \psi_\omega = \exp(\ii\,\omega) \ , $$ which are both closed and hence $\cK_1$ defines a generalised Calabi-Yau structure on $(\IT\cQ_1,0)$. 

    Consider $\cQ_2 = \cB \times \tilde{\sfT}^k$, with coordinates $\set{\tilde y^1 , \ldots, \tilde y^k}$ on the dual torus $\tilde{\sfT}^k$, and define the correspondence space $$M = \cQ_1 \times_\cB \cQ_2 \ . $$
    We take as $B$-field on $M$ the two-form $$B = \de y^1 \wedge \de \tilde y^1 + \de y^2 \wedge \de \tilde y^2 \ , $$ and hence we may write the T-duality relation
    \begin{align}
        R = {\rm Span}_{C^\infty(\cB)}\set{(\tfrac\partial{\partial x^i}, \tfrac\partial{\partial x^i})\,,\, (\de x^i,\de x^i)\,,\,(\tfrac\partial{\partial y^i},\de \tilde y^i)\,,\, (\de y^i,\tfrac\partial{\partial\tilde y^i})} \ .
    \end{align}
    
    Since $\cJ_J$ and $\cJ_\omega$ are both $\sfT^k$-invariant, we can transport them via Proposition \ref{prop:TdualGKS} to a T-dual generalised K\"ahler structure $\cK_2$ on $(\IT\cQ_2,0)$. To see what the T-dual structure looks like, we can apply the constructions of Proposition \ref{prop:Tdualspinor}. There are two equivalent approaches to this, which illustrate two ways of viewing the T-dual structure.
    
    Firstly, we can write
    \begin{align}
        \psi_J = \de z^1 \wedge \cdots \wedge \de z^k = \mu_{\rm mix} \qquad \text{and} \qquad
        \psi_\omega = \exp(\ii\,\omega)= \exp(\alpha_{\rm mix})
    \end{align}
    where
    \begin{align}
        \alpha_{\rm mix} = \ii\,\sum_{i,j}\, u_{ij}\ \de x^i \wedge \de y^j = \ii\,\de x^1 \wedge f^1 +\cdots + \ii\,\de x^k \wedge f^k
    \end{align}
    for $f^i = \sum_j\,u_{ij}\, \de y^j$. Then \smash{$\widehat f{}^{\,i} = \sum_j\, u^{ij}\, \tfrac\partial{\partial\tilde y^j}$}, where $(u^{ij})$ is the matrix inverse of $(u_{ij})$, and thus $\widetilde f{}^{\,i} = \sum_j\,u^{ij}\,\de \tilde y^j$. 
    Hence the pure spinors \smash{$\tilde \psi _J$} and $\tilde \psi_\omega$ T-dual to $\psi_J$ and $\psi_\omega$ are given by
    \begin{align}
        \tilde \psi_J &= \exp\big(\ii\, \de x^1 \wedge \de \tilde y^1 +\cdots + \ii\,\de x^k \wedge \de \tilde y^k\big) \ , \\[4pt]
        \tilde \psi_{\omega} &=\Big(\de x^1 + \ii\,\sum_j\,u^{1j}\, \de \tilde y^j\Big)\wedge \cdots \wedge \Big(\de x^k + \ii\, \sum_j\,u^{kj}\, \de \tilde y^j\Big) \ .
    \end{align}
    
    Thus $\tilde \psi_J$ and $\tilde \psi_\omega$ correspond to the symplectic structure $\tilde \omega$ and complex structure $\tilde J$ respectively given by (in the basis $(\tfrac\partial{\partial x^i},\tfrac\partial{\partial\tilde y^i})$)
    \begin{align}
        \tilde \omega = \begin{pmatrix}
            0 & -\unit\\
            \unit & 0 \\
        \end{pmatrix} \qquad \text{and} \qquad
        \tilde J = \begin{pmatrix}
            0 & U^{-1}\\
            U & 0
        \end{pmatrix} \ ,
    \end{align}
    where $U_{ij} = u_{ij}$. This defines a K\"ahler structure with K\"ahler metric
    \begin{align}
         \tilde g = \sum_{i,j}\, u_{ij}\,\de x^i \otimes \de x^j + u^{ij}\,\de \tilde y^i \otimes \de\tilde y^j
    \end{align}
    which agrees with the Buscher rules. This again defines a generalised Calabi-Yau structure $\cK_2=(\cJ_{\tilde\omega},\cJ_{\tilde J})$, though it is not immediately clear that $\tilde g$ is classically Ricci flat. To see this, we may use a different approach whereby we rewrite the two-form $\alpha_{\rm mix}$.

    To this end, write
    \begin{align}
        \alpha_{\rm mix} = \ii\,b^1 \wedge \de y^1 +\cdots + \ii\,b^k \wedge \de y^k
    \end{align}
    where $b^i = \sum_j\,u_{ji}\, \de x^j $. Then the T-dual pure spinors are given by
    \begin{align}
        {\tilde \psi}_J &= \exp\big(\ii\, \de x^1 \wedge \de \tilde y^1 +\cdots + \ii\,\de x^k \wedge \de \tilde y^k\big) \ , \\[4pt]
        {\tilde \psi}_\omega &= \Big( \sum_j\, u_{j1}\, \de x^j + \ii\,\de \tilde y^1\big) \wedge \cdots \wedge \Big(\sum_j\, u_{jk}\, \de x^j + \ii\,\de \tilde y^k\Big) \ .
    \end{align}
    If we define new coordinates $\tilde x^i = u_i \coloneqq \frac{\partial u}{\partial x_i}$, so that $\frac{\partial \tilde x^i}{\partial x^j} = u_{ij}$, then 
    \begin{align}
        \de \tilde x^i = \sum_j\,u_{ij}\, \de x^j \quad \Longrightarrow \quad \de x^i =\sum_j\,u^{ij}\, \de \tilde x^j \ ,
    \end{align}
     hence
    \begin{align}
        {\tilde \psi}_J = \exp\Big(\ii\, \sum_j\, u^{ij}\,\de \tilde x^j \wedge \de \tilde y^i\Big) \qquad \text{and} \qquad
        {\tilde \psi}_\omega = \big( \de \tilde x^1 + \ii\,\de \tilde y^1\big) \wedge \cdots \wedge \big(\de \tilde x^k + \ii\,\de \tilde y^k\big) \ .
    \end{align}
    
    The associated K\"ahler metric in these coordinates is therefore
    \begin{align}
         \tilde g = \sum_{i,j}\,u^{ij}\,\big(\de \tilde x^i \otimes \de \tilde x^j + \de \tilde y^i \otimes \de \tilde y^j\big) \ ,
    \end{align}
    whose Ricci form is $${\rm Ric}_{\tilde g} = -\ii\,\overline \partial\, \partial \log \det (u^{ij}) = 0 \ , $$ so $\tilde g$ is Ricci flat.
\end{example}

\begin{example}[\textbf{Harmonic Functions and Symplectic Structures}]
\label{ex:harmonic-symplectic}
    Consider $\cQ_1 = \cB \times \sfT^2$ for some two-manifold $\cB$. Take local coordinates $\set{x^1,x^2}$ on $\cB$ and $\set{y^1,y^2}$ on $\sfT^2$. We can endow $\cQ_1$ with the complex structures
    \begin{align}
        J^+ = \begin{pmatrix}
            0 & -\unit \\
            \unit & 0
        \end{pmatrix} \qquad \text{and}
        \qquad J^- = \begin{pmatrix}
            \eta & 0\\
            0 & \eta
        \end{pmatrix}
    \end{align}
    where $$\eta = \begin{pmatrix}
        0 & -1 \\
        1 & 0
    \end{pmatrix} \ . $$ Consider the metric $$g = u\, g_{\text{\tiny E}} \ , $$ where $g_{\text{\tiny E}}$ is the flat Euclidean metric and $u\in C^\infty(\cB)$. 
    
    Then $J^\pm$ preserve the metric $g$, and we may define the fundamental two-forms $\omega_\pm = g\,J^\pm$ which are given by
    \begin{align}
        \omega_+ = u\,(\de x^1 \wedge \de x^2 + \de y^1 \wedge \de y^2) \qquad \text{and} \qquad \omega_- = u \, (\de x^1 \wedge \de y^1+ \de x^2 \wedge \de y^2) \ .
    \end{align}
    We may calculate
    \begin{align}
        \de_-^c \omega_- = -\de_+^c \omega_+ = (u_2 \,\de x^1 - u_1\, \de x^2)\wedge \de y^1 \wedge \de y^2 =: \red H \ ,
    \end{align}
    where $u_i = \frac{\partial u}{\partial x^i}$.
    The three-form $\red H\in\sfOmega^3(\cQ_1)$ is closed if and only if $$\Delta u=0 \ , $$ where $\Delta$ is the Laplacian on $\cB$. 
    
    It then follows that $\cK = (\cJ^+, \cJ^-)$ is an $\red H$-twisted generalised K\"ahler structure, where
    \begin{align}
        \cJ^\pm = \frac 12\,
        \begin{pmatrix}
            J^+ +J^- & -(\omega_+^{-1} \mp \omega_-^{-1})\\
            \omega_+ \mp \omega_- & -(J^+{}^{\tt t} \pm J^-{}^{\tt t})
        \end{pmatrix} \ .
    \end{align}
    The corresponding $+\ii$-eigenbundles are given by
    \begin{align}
        L^\pm = {\rm Span}_{C^\infty(\cQ_1)^\IC}\lbrace &\mp \tfrac\partial{\partial x^1} +\ii\,\tfrac\partial{\partial x^2} \pm \ii\,\tfrac\partial{\partial y^1} + \tfrac\partial{\partial y^2}\,,\,\mp \, \ii\,\tfrac\partial{\partial x^1} - \tfrac\partial{\partial x^2} \pm \ii\, u\, \de x^1 + u\, \de x^2\,,\\
        & \quad \ii\,\tfrac\partial{\partial x^1} + \tfrac\partial{\partial y^1} + \ii\, u\, \de x^1 + u\, \de y^1\,,\, \ii\, \tfrac\partial{\partial x^2} \mp \ii\, \tfrac\partial{\partial y^1} \pm u\, \de x^1 + u\,\de y^2 \, \rbrace \ ,
    \end{align}
    which can be written as $$L^\pm = \gr\big(\exp(\alpha_\pm)\big|_{E_\pm}\big) \oplus \ann(E_\pm) \ , $$ where $$E_\pm = \red\rho{}_1(L^\pm) \qquad \text{and} \qquad \alpha_\pm = u \, (\ii\,\de x^1 + \de y^1)\wedge(\ii\,\de x^1 \pm \de x^2) \ . $$ 
    
    Hence the generating spinors are given by
    \begin{align}
        \psi^\pm = \exp(\alpha_\pm) \wedge \theta_\pm \ ,
    \end{align}
    where $$\theta_\pm = \pm\,  \ii\,\de x^1 + \de x^2 \pm \de y^1 - \ii\, \de y^2 $$ are closed one-forms. We then find
    \begin{align}
        \de_{\red H} \psi^\pm &= \de\alpha_\pm\wedge \exp(\alpha_\pm) \wedge \theta_\pm + \red H\wedge \exp(\alpha_\pm) \wedge \theta_\pm\\[4pt]
        &=\de\alpha_\pm\wedge  \theta_\pm + \red H\wedge  \theta_\pm\\[4pt]
        &=(\pm\,\ii\,u_1 + u_2)\,\de x^1 \wedge \de x^2 \wedge \de y^1 \wedge \de y^2 + (\mp\,\ii\,u_1 - u_2)\,\de x^1 \wedge \de x^2 \wedge \de y^1 \wedge \de y^2\\[4pt]
        &=0 \ ,
    \end{align}
    and therefore $\cK$ is a generalised Calabi-Yau structure on $(\IT\cQ_1,\red H)$.
    
    We now apply T-duality in the $y^2$ direction. Thus let $\cQ_2$ be a principal circle bundle over $\cB\times \sfT^1$, with coordinate $\tilde y^2$ on the fibres. On the correspondence space $$M = \cQ_1 \times_{\cB \times \sfT^1} \cQ_2$$ we take the $B$-field
    \begin{align}
        B = \de y^2 \wedge \varTheta
    \end{align}
    where $$\varTheta = \de \tilde y^2 - y^1\, u_2\, \de x^1 + y^1\, u_1\, \de x^2 \ . $$ Note that $\cQ_2$ therefore has Chern class
    \begin{align}
        [\de\varTheta] = \big[u_2\,\de x^1 \wedge \de y^1 - u_1\, \de x^2 \wedge \de y^1\big] \ \in \ \sfH^2(\cB\times\sfT^1 ,\IR) \ .
    \end{align}
    Then $\de B = \red H$, and there is a T-duality relation 
    \begin{equation}
\begin{tikzcd}
 R  \colon (\IT \cQ_1, \red H) \ar[r,dashed] & (\IT \cQ_2,0)
 \end{tikzcd}
\end{equation} 
which is generated as a $C^\infty(\cB\times\sfT^1)$-module by
    \begin{align}
        \set{(\tfrac\partial{\partial x^i}, \tfrac\partial{\partial x^i} - (-1)^i\,y^1\, u_{3-i}\, \tfrac\partial{\partial \tilde y^2})\,,\,(\de x^i, \de x^i)\,,\,(\tfrac\partial{\partial y^1}, \tfrac\partial{\partial y^1})\,,\, (\de y^1, \de y^1)\,,\, (\tfrac\partial{\partial y^2}, \varTheta)\,,\, (\de y^2, \tfrac\partial{\partial \tilde y^2})}_{i=1,2} \ .
    \end{align}
    
    By Proposition \ref{prop:Tdualspinor}, we may write the T-dual pure spinors as $$\tilde \psi^\pm = \exp(\tilde \alpha_\pm) \ , $$ where 
    \begin{align}
        \tilde \alpha_\pm = \mp\, u\, \de x^2 \wedge \de y^1 - \ii\,u\,\de x^1 \wedge \de y^1 \pm \ii\,u \, \de x^1 \wedge \de x^2 -\ii\,(\pm\, \ii\,\de x^1 + \de x^2 \pm \de y^1)\wedge \varTheta \ .
    \end{align}
    By T-duality, this defines a generalised Calabi-Yau structure on $(\IT\cQ_2,0)$, i.e. $\de \tilde \psi_\pm =0$.
    We may decompose $\tilde \alpha_\pm = \pm \tilde B + \ii\, \tilde\omega_\pm$, where
    \begin{align}
        \tilde B = y^1\, u_1\,\de x^1 \wedge \de x^2 - u\, \de x^2 \wedge \de y^1 + \de x^1 \wedge \de \tilde y^2 \ .
    \end{align}
    Then $\de \tilde B = 0$ and hence $\de \tilde \omega_\pm=0$, so that $\tilde \psi_\pm$  define $\pm\tilde B$-twists of genuine symplectic structures on~$\cQ_2$.
\end{example}

\medskip

\appendix

\section{A Note on Fibred Integration}\label{appendix}

In the main text we use some graded commutativity properties of fibred integration which we summarise here. Let $M$ be a fibred manifold with $\dim_\IR M = n$ and oriented fibres $\cF$ of codimension $n-k$. Denote by $\varpi:M\to M/\cF$ the smooth quotient map.

The fibred integral over $\cF$ of a $p$-form $\omega$ on $M$ is the $(p{-}k)$-form on $M/\cF$ defined as
\begin{align}
    \left (\int_{\cF}\, \omega\right )(X_1,\dots,X_{p-k}) = \int_{\cF}\, \mu_\omega \ ,
\end{align}
where $\mu_\omega$ is a fibre volume form such that
\begin{align}
    \omega(Y_1,\dots,Y_k, X^*_1,\dots, X^*_{p-k}) = \mu_\omega(Y_1,\dots,Y_k)
\end{align}
for vertical vector fields $Y_1,\dots,Y_k$, and vector fields $X_1,\dots,X_{p-k}\in T(M/\cF)$ with horizontal lifts $X_1^*,\dots,X_{p-k}^*\in TM$, i.e. $\varpi_*(X_i^*)=X_i$. 

Thus if $X\in T(M/\cF)$, then
\begin{align}
   \left ( \iota_X \int_{\cF}\,\omega \right)(X_2,\dots,X_{p-k}) = \left (\int_{\cF}\, \omega \right)(X,X_2,\dots,X_{p-k}) \ .
\end{align}
On the other hand
\begin{align}
    \left (\int_{\cF}\, \iota_{ X^*} \omega \right ) (X_2,\dots,X_{p-k}) = \int_{\cF}\, \mu_{\iota_{X^*}\omega}
\end{align}
where 
\begin{align}
    \mu_{\iota_{X^*}\omega}(Y_1,\dots,Y_k) = (\iota_{ X^*}\omega)(Y_1,\dots,Y_k,  X^*_2,\dots, X^*_{p-k}) = (-1)^k\, \omega(Y_1,\dots,Y_k,  X^*,  X^*_2,\dots,  X^*_{p-k}) \ .
\end{align}
Thus
\begin{align}
    \iota_X \int_\cF\, \omega = (-1)^k\, \int_\cF\, \iota_{ X^*} \omega \ .
\end{align}

Similarly
\begin{align}
    \xi \wedge \int_\cF\, \omega = (-1)^k\, \int_\cF\, \varpi^* \xi \wedge \omega \qquad \text{and} \qquad \de \int_\cF\, \omega = (-1)^k\, \int_\cF\, \de \omega
\end{align}
for any form $\xi$ on $M/\cF$. 

In particular, this has the consequence of introducing a grading on the smooth Fourier-Mukai transform $\varrho$ in \eqref{eqn:FMTproperties}.

\bibliographystyle{ourstyle}
\bibliography{bibprova5}

\end{document}